\documentclass[openbib,12pt,notitlepage]{article}
\usepackage{eurosym}
\usepackage{amsfonts}
\usepackage{amsmath}
\usepackage{amssymb}
\usepackage{graphicx}
\usepackage{color}
\usepackage{natbib}
\usepackage{newtxtext}
\usepackage{newtxmath}
\usepackage[cal=cm]{mathalfa}
\usepackage{footmisc}
\usepackage{caption}
\usepackage{comment}
\usepackage[colorlinks,citecolor=blue]{hyperref}

\newtheorem{theo}{Theorem}
\newtheorem{prop}{Proposition}
\newtheorem{assum}{Assumption}
\newtheorem{define}{Definition}
\newtheorem{remark}{Remark}

\newtheorem{lemma}{Lemma}
\newtheorem{eg}{Example}

\newtheorem{cor}{Corollary}

\newenvironment{proof}[1][Proof]{\noindent \textbf{#1\ } }{\  \rule{0.5em}{0.5em}}

\providecommand{\func}[1]{\operatorname{#1}}

\makeindex
\input epsf

\begin{document}

\title{\textbf{Contracting with Imperfect Commitment: Minimal Canonical
Contracts\thanks{%
Han gratefully acknowledges financial support from Social Sciences and
Humanities Research Council of Canada.}}}
\author{Seungjin Han\thanks{
Department of Economics, McMaster University, Canada, hansj@mcmaster.ca}
\and Siyang Xiong\thanks{
Department of Economics, University of California, Riverside, United States,
siyang.xiong@ucr.edu}}
\maketitle

\begin{abstract}
Contract theory typically assumes full commitment by the principal, but many
contracts fix some payoff-relevant decisions while leaving others
discretionary. We ask when imperfect commitment is equivalent to full
commitment. For contracts in which a committed baseline is followed by a
bounded discretionary adjustment, as in commercial-insurance schedule rating
or civil penalties, bounded discretion is allocation-neutral. When
contractible and non-contractible decisions are distinct instruments, the
equivalence fails. We characterize optimal single-principal contracts and
show that simple-offer equilibria are robust under competing principals. The
methodological contribution is an extended taxation principle that makes
these analyses more tractable.
\end{abstract}


\section{Introduction}

Contract theory typically assumes full commitment by the principal: every
payoff-relevant decision is fixed in advance, often as a function of the
agent's report. In many real contracts, commitment is imperfect. Parties fix
some payoff-relevant decisions in advance but retain discretion over others.
A firm may commit to a base salary, but not to every dimension of the job:
workload, pace, task assignment, scheduling demands, or the care with which
work is performed. A buyer may specify price and delivery date, while
leaving quality, effort, or adaptation to be determined later. A regulator
may set a formal standard while retaining discretion over monitoring or
enforcement. In each case, some decisions are fixed ex ante, while other
payoff-relevant choices are made ex post, often after private information
has been revealed.

This paper asks when imperfect commitment is equivalent to full commitment.
The full-commitment benchmark may be a legitimate shortcut, but only if the
corresponding limited-commitment environment supports the same allocations.

A common contractual form fixes a baseline in advance while giving the
principal discretion to choose a bounded adjustment later. A compensation
contract specifies a base salary in advance and gives the firm discretion to
choose an end-of-year bonus within a contractually bounded range. A
commercial insurer files a manual premium computed from codifiable rating
factors such as class, location, coverage, and exposure, leaving an
underwriter with discretion to apply a bounded schedule credit or debit for
risk characteristics that are harder to write down.\footnote{%
For example, \href{https://app.leg.wa.gov/wac/default.aspx?cite=284-24-100}{%
Washington Administrative Code 284-24-100} allows schedule rating for
several commercial lines, limits the schedule credit or debit to 25 percent,
and requires the adjustment to be based on specific information supporting
the rating decision. \href{https://codes.ohio.gov/ohio-administrative-code/rule-3901-1-22%
}{Ohio Administrative Code 3901-1-22} similarly describes schedule rating as
judgment credits or debits for risk characteristics not otherwise reflected
in the basic premium, subject to a 25 percent maximum debit and credit.} An
inspector at the Occupational Safety and Health Administration (OSHA)
calculates a gravity-based penalty from relatively observable features of a
workplace safety and health violation and may then apply bounded adjustments
for employer size, history, good-faith safety efforts, and prompt abatement.%
\footnote{%
OSHA's Field Operations Manual, Chapter 6, describes gravity as the primary
basis for calculating the basic penalty. It then permits bounded
adjustments, including a history adjustment, a good-faith reduction, a size
reduction, and a quick-fix reduction, with documentation and consistency
requirements for the use of discretion. See Occupational Safety and Health
Administration, Field Operations Manual, Chapter 6, \href{https://www.osha.gov/fom/chapter-6%
}{https://www.osha.gov/fom/chapter-6}.}

The first result is that, in this class of contracts, bounded discretion is
allocation-neutral: every final-action allocation that can be implemented
under full commitment can also be implemented when the principal retains
bounded discretion.

The intuition is that the principal need not commit directly to the later
adjustment in order to reproduce the full-commitment outcome. The baseline
can position the range of feasible final actions around the intended
outcome. If the target lies below the principal's posterior ideal, the
baseline places the target at the upper edge of the adjustment range; if it
lies above the posterior ideal, the baseline places the target at the lower
edge. The principal's later best response then selects the same final action
that full commitment would have selected. Because final actions following
messages are unchanged, the agent's incentives are unchanged as well.

Is the equivalence universal? It is not. The argument behind the first
result relies on a specific structural feature of bounded-adjustment
contracts: the contractible baseline and the discretionary adjustment
combine into one final action, and both parties' payoffs depend on that
final action rather than on its decomposition. The adjustment may be
additive, as in a fixed dollar credit or debit, or proportional, as in a
percentage adjustment to a baseline; in either case, the baseline can be
chosen so that the principal's later best response implements the intended
final action. This feature fails in many economic relationships. A worker's
later effort is not a revision of the wage. Enforcement intensity is not a
revision of a formal standard. Product quality is not a revision of price.
When the contractible and discretionary choices are economically distinct
instruments, full-commitment analysis need not remain a valid guide, and
optimal contracts must be characterized directly.

Both halves of this analysis - the equivalence in bounded-adjustment
contracts and the failure of equivalence in distinct-instrument contracts -
rest on the same representation. Imperfect commitment changes the role of
communication. Under full commitment, a message can be treated as selecting
a final action. Under imperfect commitment, the same message also shapes the
beliefs that govern the principal's later discretionary choice. Two messages
that lead to the same contractible action may therefore induce different
continuation behavior. Ordinary menu representations are too coarse to
capture this. The paper develops an extended taxation principle: mechanisms
can be represented as menus of contractible baselines together with
non-binding recommendations for later discretionary choices. The
recommendations are not enforceable; they record the continuation behavior,
and the beliefs, induced by the original mechanism. This representation
preserves the tractability of menu analysis while keeping the
belief-mediated effects of communication that are central under imperfect
commitment.

For the distinct-instrument case, the paper analyzes two applications. In a
single-principal model, it characterizes optimal contracts when some terms
are fixed in advance but later productive decisions are not, and the agent
retains an exit option after observing the discretionary choice. In a
competing-principals model, equilibria supported by simple contract offers
remain robust to arbitrary mechanism deviations, in the sense that no
deviation is profitable across all continuation equilibria following it.
Limited commitment thus need not destroy tractability, but it changes which
contracts are optimal and which are not.

\paragraph{Related Literature}

This paper contributes to the literature on contracting under imperfect
commitment. Under imperfect commitment, a mechanism does not merely select
contractible actions; it also shapes the beliefs that govern later
discretionary choices. Standard revelation-principle and menu-representation
arguments need not apply.

Early work derived optimal contracts in specific environments, including the
durable-goods and regulation models of \cite{ht1988} and \cite{lt1990}. \cite%
{bs2001} provide the key breakthrough by recovering a canonical
direct-mechanism representation under finite type spaces. Subsequent work
extends this analysis: \cite{vs2006} to durable-goods models with continuous
types and short-term contracts, \cite{km2008} to \cite{vcjs1982} settings
with report-contingent transfers, and \cite{bk2012,bk2017} to exit options
and mediation.

The present paper shifts from the revelation-principle route to a
taxation-principle route. It develops an extended taxation principle:
without loss of generality, the principal can restrict attention to menus of
contractible actions augmented with non-binding recommendations over
non-contractible actions. The recommendations are not enforceable, but they
record the continuation behavior and beliefs that each message induces,
which a standard menu of contractible actions alone discards. This
methodological shift broadens the class of environments that can be handled.
Existing tractable approaches rely on different kinds of structure: finite
type spaces in \cite{bs2001}, quasilinearity in \cite{vs2006}, or additional
commitment through information-design devices in \cite{ds2021}. The extended
taxation-principle approach instead accommodates continuous type spaces,
non-quasilinear preferences, and settings in which the principal lacks the
additional commitment power assumed in those approaches. A related but
distinct literature studies limited commitment in dynamic mechanism design,
where the principal commits to the current contract but not to future ones (%
\citealp{liu2019,ds2024}). That strand addresses commitment across time; the
present paper addresses commitment across contractible and non-contractible
dimensions within a single interaction.

The baseline-and-revision interpretation also relates the paper to
delegation. In \cite{mmts1991} and \cite{am2008}, the principal controls the
feasible set from which an informed party chooses under full commitment; in
our paper, the principal controls a baseline, and the later choice is a
bounded revision. The equivalence result gives a limited-commitment reading
of these classical delegation results.

Applied work shows that limited commitment is economically important across
a range of applications: \cite{gim2013} on procurement renegotiation, \cite%
{am2016} on public utility provision, \cite{bm2007} on renegotiation-proof
commitment contracts, and \cite{gi2021} on unemployment-insurance design
when governments cannot fully commit. The contribution here is
complementary: the paper provides a general tool rather than focusing on a
particular application.

Finally, the paper contributes to a broader methodological question: how can
one establish canonical contract spaces? Two major approaches are the
revelation principle and the taxation principle. We relate to the latter,
developed by \cite{guesnerie1981,guesnerie1995} and \cite{rochet1986}. \cite%
{mp2001} and \cite{ms2002} identify one important advantage of the taxation
principle over the revelation principle: in common-agency problems, the
taxation principle applies even when the revelation principle generally does
not. We identify another advantage. Under imperfect commitment, a
taxation-principle approach substantially expands tractability relative to
revelation-principle methods, both in single-principal and in common-agency
environments. In particular, to the best of our knowledge, this paper
provides the first tractable and general canonical-contract treatment of
imperfect commitment in common-agency problems.

The paper proceeds as follows. Section \ref{sec:examples} presents the
motivating institutional examples. Section \ref{sec:revisable} studies the
clean revisable-action case and establishes the equivalence with full
commitment. Section \ref{sec:model} introduces the general model that
encompasses this case and allows contractible and discretionary decisions to
be distinct instruments. Section \ref{sec:canonical} develops the extended
taxation principle, and Section \ref{sec:labor_apps} applies the framework
to distinct discretionary actions and exit.

\section{Examples\label{sec:examples}}

The introduction emphasized that imperfect commitment matters because many
contracts fix only part of the relevant decision while leaving another part
to be chosen later. The following examples illustrate the scope of the
framework. The first is the clean revisable-action case emphasized in the
introduction: a baseline is fixed from hard information and later revised
within a bounded range using softer information. The second studies optimal
contracting when a contractible employment term can be fixed in advance but
later productive decisions are not. The third extends the logic to
common-agency environments, where a single agent interacts with multiple
principals. In each case, the central question is whether full-commitment
outcomes can still be implemented, and, when they cannot, what form optimal
or equilibrium contracts take.

\subsection{Example 1: Schedule Rating, Civil Penalties, and Revisable
Policy Choices\label{sec:Example1}}

Commercial insurance schedule rating is a direct instance of the
revisable-action structure. The filed manual premium is a committed
baseline, determined from codifiable rating factors, while the underwriter's
schedule credit or debit is a bounded discretionary revision for risk
characteristics not fully priced in the manual rate. If the adjustment is a
fixed amount, the final premium is $z=x+y$; if it is a percentage of the
manual premium, the final premium is $z=x\eta$, or equivalently $z=x(1+y)$,
with the adjustment factor restricted to a bounded interval. OSHA
civil-penalty policy has the same form: the gravity-based penalty is the
baseline, and later adjustments for size, history, good faith, and prompt
abatement are bounded revisions, often proportional reductions from that
baseline.

In this mapping, the agent's message is the underwriting submission,
inspection record, or employer evidence that contains both codifiable items
and judgmental information. The hard part of the message fixes the baseline;
the soft part shapes the posterior used for the bounded revision. The
revision is therefore not a response to a new demand shock, but a later use
of information that was too judgmental to hardwire fully into the baseline.

Section \ref{sec:revisable} states the model additively for notational
economy. The proportional case uses the same endpoint logic: with positive
baselines and $0<\underline{\eta }\leq \eta \leq \bar{\eta}$, the feasible
final-action interval is $[\underline{\eta }x,\bar{\eta}x]$.

Policy delegation provides a textbook version of the same idea. In \cite%
{am2008}, a legislature, the principal, and a committee, the agent, choose a
policy $z\in \mathbb{R}$. Their payoffs are
\begin{equation*}
v(z,\theta )=-\left( \frac{z}{\theta }-1\right) ^{2}\quad \text{and}\quad
u(z,\theta )=-\left( \frac{z}{\theta }-1-b\right) ^{2},
\end{equation*}%
where $b\in \mathbb{R}$ is the committee's bias and $\theta $ is privately
observed by the committee. The legislature has a prior belief over $\theta $%
, for example a truncated normal distribution on a compact support.

The standard formulation assumes commitment to the policy instrument. With
revisable actions, the legislature instead commits to a baseline directive $%
x $ after the committee's report and later chooses a bounded implementation
revision $y$, so
\begin{equation*}
z=x+y,\qquad y\in[-\alpha,\alpha].
\end{equation*}

The question is whether the legislature can attain the full-commitment
allocation when only $x$ is contractible and $y$ remains discretionary. The
equivalence result in Section \ref{sec:revisable} identifies conditions
under which the contractible component can organize communication and induce
continuation behavior that replicates the full-commitment outcome.

\subsection{Example 2: Labor Contracting with Non-Contractible Work Intensity%
}

\label{sec:motivating_example}

The second example illustrates optimal contracting when some employment
terms can be fixed in advance but later productive decisions cannot. A firm,
the principal, considers hiring a worker, the agent. The firm can commit to
a contractible employment term $x \in \mathbb{R}$, while work intensity or
speed $y \in \mathbb{R}_{+}$ is non-contractible and chosen after the
contract is accepted.

If the worker is employed by the firm, payoffs are
\begin{equation*}
v(x,y,\theta)=y\theta-x^2 \quad \text{and} \quad u(x,y,\theta)=\frac{%
x\theta-y^2}{\sqrt{\theta}},
\end{equation*}
where $\theta$ is the worker's privately observed type. The firm has a prior
belief over $\theta$, and both parties receive zero from their outside
options.

This example captures a common feature of employment relationships. A firm
can typically commit to a compensation term, but it often cannot fully
commit to the eventual intensity of the job. Workload, pace, task
assignment, scheduling demands, travel requirements, and other job
attributes are frequently determined only after the worker has been hired.
As in \cite{bk2012,bk2017}, the worker has an exit option: after accepting a
contract and observing the firm's choice of the non-contractible component,
she may exit and obtain her reservation payoff.

The key question is how the firm should design the contract when it cannot
commit to the later productive decision. Unlike in quasilinear environments
where the contractible component is interpreted purely as a monetary
transfer, here $x$ represents a broader contractible employment term, and
the worker's utility need not be quasilinear. The framework in Subsection %
\ref{sec:appI} characterizes the optimal limited-commitment contract and
shows how the contractible component, the discretionary action, and the
worker's exit option jointly shape the solution.

\subsection{Example 3: Labor Contracting with a Common Agent}

\label{sec:Example3}

The third example illustrates imperfect commitment in a common-agency
environment. A worker, such as a consultant, may contract with two firms.
Firm $j \in \{1,2\}$ chooses a contractible support or service term $x_j$
and a non-contractible work-intensity component $y_j$. The firms' payoffs
are
\begin{equation*}
v_j=(1+\beta x_{-j})y_j\theta-x_j^2,\qquad j=1,2,
\end{equation*}
where $\beta>0$. The worker's payoff is
\begin{equation*}
u=u_1+u_2 = \frac{x_1\theta-y_1^2}{\sqrt{\theta}} + \frac{x_2\theta-y_2^2}{%
\sqrt{\theta}},
\end{equation*}
where $\theta$ is privately observed by the worker and follows a commonly
known prior distribution.

This example captures settings in which a common agent provides services to
multiple principals and the principals' payoffs are interdependent. A
consultant, contractor, platform, or expert may allocate attention across
several clients, and one principal's contractible support may affect the
productivity of work performed for another.

The question is which equilibrium allocations are robust when each principal
may offer arbitrary mechanisms rather than just simple contract offers.
Under full commitment, the taxation principle provides a canonical menu
representation in common-agency environments. With imperfect commitment,
ordinary menus are too coarse because non-contractible actions are chosen
after communication and depend on the beliefs induced by the mechanism. The
extended taxation principle that we will develop provides the appropriate
representation by adding non-binding recommendations over discretionary
actions. As shown in Subsection \ref{sec:appII}, this allows us to identify
conditions under which equilibria supported by simple contract offers
survive deviations to more complex mechanisms.

\section{The Equivalence with Revisable Actions}

\label{sec:revisable}

We begin with the paper's cleanest positive result: in a class of
revisable-action environments, limited commitment changes implementation but
not the equilibrium allocations over final actions. The principal commits to
a baseline action and later revises it within a bounded range. Because the
baseline can be chosen to position the feasible set of final actions around
the intended outcome, the full-commitment allocation set remains valid. This
equivalence gives a limited-commitment interpretation of classical
delegation results such as \cite{mmts1991} and \cite{am2008}.

\paragraph{Motivating examples.}

The commercial-insurance schedule rating and OSHA examples in Section \ref%
{sec:Example1} have this structure. A manual premium or gravity-based
penalty is the committed baseline, and a later credit, debit, or penalty
adjustment is a discretionary revision within a bounded range. Public-policy
delegation provides another example: a legislature commits to a baseline
directive, while a committee or regulator later revises implementation
within a bounded range.

\paragraph{Model.}

There are two players: a sender (agent) and a receiver (principal). The
sender privately observes a type $\theta \in \Theta $, where $\Theta $ is an
interval. The receiver can commit to a contractible action $x\in \mathbb{R}$
conditional on the sender's message. After observing the message, the
receiver may revise this action by choosing
\begin{equation*}
y\in \lbrack -\alpha ,\alpha ],
\end{equation*}%
where $\alpha \geq 0$ is exogenous. The final action is therefore
\begin{equation*}
z=x+y.
\end{equation*}

The sender's and receiver's payoffs are
\begin{equation*}
u(x,y,\theta)=u(z,\theta), \qquad v(x,y,\theta)=v(z,\theta).
\end{equation*}
Thus both players care only about the final action $z$, not about how it is
decomposed into the contractible baseline $x$ and the later revision $y$.

The parameter $\alpha$ measures the degree of discretion. When $\alpha=0$,
the receiver cannot revise the contractible action, so the model reduces to
a full-commitment delegation problem. Larger values of $\alpha$ give the
receiver more ex post discretion. If the contractible baseline were absent,
this would move the problem toward cheap talk. The point of the result below
is that a contractible baseline, combined with communication, can preserve
the full-commitment allocation set even when the receiver retains bounded
discretion.

\paragraph{Equilibrium and final-action allocations.}

To state the result, it is enough to describe equilibrium informally. After
observing his type, the sender sends a message. The receiver then chooses a
baseline as a function of that message and, after observing the same
message, chooses a bounded revision. An equilibrium therefore specifies
messaging behavior, posterior beliefs after messages, and baseline and
revision choices that are mutually optimal. These equilibrium objects
induce, for each type $\theta$, a distribution over final actions $z=x+y$.

Let $\Gamma (\alpha )$ denote the set of type-contingent distributions over
final actions that can arise in this way when the revision is constrained to
lie in $[-\alpha ,\alpha ]$. Deterministic allocations are the special case
in which each type induces a degenerate distribution on a single final
action. Sections \ref{sec:model} and \ref{sec:canonical} later provide the
fully general formal definitions and show that the same allocations can be
represented in the canonical menu-with-recommendations contract space.

\paragraph{Assumptions and interpretation.}

For every posterior belief $p\in\Delta(\Theta)$, let the receiver's interim
expected payoff be
\begin{equation*}
V(z\mid p):=\int_\Theta v(z,\theta)\,p(d\theta).
\end{equation*}
Assume that $v(\cdot,\theta)$ is strictly concave for every $\theta\in\Theta$%
. Also assume that, for every posterior belief $p\in\Delta(\Theta)$, $%
V(\cdot\mid p)$ admits a maximizer
\begin{equation*}
r(p)\in \arg\max_{z\in\mathbb{R}} V(z\mid p),
\end{equation*}
and that the induced rule $p\mapsto r(p)$ is measurable. Since $V(\cdot\mid
p)$ is strictly concave, $r(p)$ is unique. The number $r(p)$ is the
receiver's posterior ideal final action.

\paragraph{Main result.}

The statement uses the admissible mechanisms of the general model; Appendix %
\ref{sec:proof_appIII} records the minor regularity convention needed for
the transformations in the proof.

\begin{prop}
\label{prop:revisable_equivalence_main}Suppose that $v(\cdot ,\theta )$ is
strictly concave for every $\theta \in \Theta $ and that, for every
posterior belief $p\in \Delta (\Theta )$, $V(\cdot \mid p)$ admits a
maximizer $r(p)$, with $p\mapsto r(p)$ measurable. Then, for every $\alpha
\in \lbrack 0,\infty )$,
\begin{equation*}
\Gamma (\alpha )=\Gamma (0).
\end{equation*}%
Equivalently, a stochastic final-action allocation is an equilibrium
allocation under limited commitment if and only if it is an equilibrium
allocation under full commitment.
\end{prop}

The proof, given in Appendix \ref{sec:proof_appIII}, first establishes the
equivalence for unrestricted mechanisms and then uses the canonical-contract
reduction developed later in Sections \ref{sec:model} and \ref{sec:canonical}
to express the same final-action allocations in the
menu-with-recommendations contract space. One direction is immediate: any
revisable-action equilibrium can be converted into a full-commitment
equilibrium by replacing each pair $(x,y)$ with the same final action $z=x+y$%
. Conversely, start from a full-commitment equilibrium. For each message,
choose the baseline so that the intended final action lies at the
appropriate endpoint of the feasible revision interval. The bounded revision
then makes that final action the receiver's constrained optimum. Thus
limited commitment changes implementation, but not the attainable
final-action allocations.

The proposition has an immediate optimization consequence.

\begin{cor}
\label{cor:revisable_optimality}Suppose an objective criterion depends only
on the stochastic final-action allocation. Under the conditions of
Proposition \ref{prop:revisable_equivalence_main}, the value and the set of
optimal allocations are the same under limited commitment and full
commitment.
\end{cor}

Corollary \ref{cor:revisable_optimality} implies that one can solve the
full-commitment problem and then implement the same allocation with a
contractible baseline and a bounded revision. Thus, in schedule-rating,
civil-penalty, and policy-delegation settings, the full-commitment solution
remains valid when the payoff-relevant object is the final premium, penalty,
or policy and the discretionary component is a bounded revision of that same
object.

The same endpoint-placement logic also applies to proportional revisions,
such as percentage credits or debits relative to a baseline. If the final
action is $z=x\eta$ with $\eta\in[\underline{\eta},\bar{\eta}]$ and $0<%
\underline{\eta}\leq\bar{\eta}$, then for a given baseline $x\geq 0$ the
feasible final-action interval is $[\underline{\eta}x,\bar{\eta}x]$. To
implement a target final action below the receiver's posterior ideal, choose
the baseline so that the target is the upper endpoint; to implement a target
above the receiver's posterior ideal, choose the baseline so that it is the
lower endpoint. Thus proportional revision uses the same endpoint argument
as the additive model, provided the implied baseline is admissible.

\paragraph{Implication: Melumad and Shibano (1991).}

We now apply the equivalence result to the quadratic environment of \cite%
{mmts1991}. The sender's and receiver's preferences over the final action $z$
are
\begin{equation*}
u(z,\theta )=-(z-\theta )^{2},\qquad v(z,\theta )=-(z-k-a\theta )^{2},
\end{equation*}%
where $a\in (0,1)$ and $k\in \mathbb{R}$. The receiver's payoff is strictly
concave in $z$ for every type $\theta $, and for any posterior belief $p$
the receiver's expected payoff has the unique ideal point
\begin{equation*}
r(p)=k+a\int_{\Theta }\theta \,p(d\theta ).
\end{equation*}%
Thus Proposition \ref{prop:revisable_equivalence_main} applies directly.
Applying Proposition 3 of \cite{mmts1991}, when
\begin{equation*}
k\in \left( -\frac{a}{2},1-\frac{a}{2}\right) ,
\end{equation*}%
the sender-optimal equilibrium allocation in the full-commitment model takes
the form
\begin{equation*}
z^{\ast \ast }(\theta )=%
\begin{cases}
\theta _{1} & \text{for }\theta \in \lbrack 0,\theta _{1}), \\
\theta & \text{for }\theta \in \lbrack \theta _{1},\theta _{2}], \\
\theta _{2} & \text{for }\theta \in (\theta _{2},1],%
\end{cases}%
\end{equation*}%
where
\begin{equation*}
\theta _{1}=\max \left\{ 0,\frac{2k}{2-a}\right\} ,\qquad \theta _{2}=\min
\left\{ \frac{2k+a}{2-a},1\right\} .
\end{equation*}%
Thus the same allocation remains optimal under limited commitment for every $%
\alpha \geq 0$; only the implementation changes.

\paragraph{Implication: Alonso and Matouschek (2008).}

A similar implication applies to the delegation environment studied by \cite%
{am2008}. In the specification described in Section \ref{sec:Example1}, on a
positive compact type support, the receiver's primitive payoff is strictly
concave in the final policy for every type, and each posterior expected
payoff has a unique ideal point. Corollary \ref{cor:revisable_optimality}
therefore implies that, in the corresponding revisable-action environment,
the optimal limited-commitment allocation is again exactly the same as the
full-commitment allocation.

\section{A More General Model of Imperfect Commitment}

\label{sec:model}

Section \ref{sec:revisable} presented the cleanest implication of the
framework: when a contractible baseline and a later discretionary adjustment
jointly determine a single final action, limited commitment can reproduce
full commitment. We now introduce the general model behind that argument.
The model encompasses the revisable-action environment but also allows the
contractible component $x$ and the later discretionary component $y$ to be
distinct payoff-relevant instruments, and it allows several principals to
contract with the same informed agent.

The key object is a message that does two jobs at once. It selects the
contractible component through the mechanism, and it carries information
that shapes the posterior beliefs used in the later discretionary choice.
Thus two messages can induce the same contractible action but lead to
different continuation choices. This is why ordinary menus of contractible
actions are not rich enough under imperfect commitment. The definitions
below formalize this continuation problem and then use it to study canonical
contract spaces.

All primitive sets and admissible message spaces are standard Borel spaces,
and all contracts, strategies, beliefs, and payoff functions are measurable.
For a measurable set $K$, let $\Delta(K)$ denote the set of probability
distributions on $K$, and let $\mathcal{B}(K)$ denote its Borel $\sigma$%
-algebra. Appendix \ref{app:formal_continuation_equilibrium} records the
formal measurable definition of continuation equilibrium. Existence is
imposed only when a result asserts attainment; otherwise value statements
are understood with suprema.

\subsection{Environment}

There is one agent and a finite set $\mathcal{J}$ of principals. The
one-principal model is the special case $|\mathcal{J}|=1$. The
payoff-relevant state is $\theta \in \Theta$, drawn from a commonly known
prior $\mu \in \Delta(\Theta)$. The agent privately observes $\theta$.

Each principal $j\in \mathcal{J}$ chooses an action $z_{j}=(x_{j},y_{j})\in
X_{j}\times Y_{j},$ where $x_{j}$ is contractible and $y_{j}$ is not. Let $%
X:=\times _{j\in \mathcal{J}}X_{j}$ and $Y:=\times _{j\in \mathcal{J}}Y_{j}$%
. Principal $j$'s feasible non-contractible actions may depend on the
contractible action that principal $j$ has chosen. This dependence is
represented by a non-empty-valued correspondence
\begin{equation*}
F_{j}:X_{j}\rightrightarrows Y_{j},
\end{equation*}%
where $F_{j}(x_{j})\neq \emptyset $ for every $x_{j}\in X_{j}$. Thus, after
choosing $x_{j}$, principal $j$ can choose only some $y_{j}\in F_{j}(x_{j})$.

The agent's Bernoulli utility is
\begin{equation*}
u:X\times Y\times \Theta \rightarrow \mathbb{R},
\end{equation*}
and principal $j$'s Bernoulli utility is
\begin{equation*}
v_{j}:X\times Y\times \Theta \rightarrow \mathbb{R}.
\end{equation*}
The agent's outside option is
\begin{equation*}
U:\Theta \rightarrow \mathbb{R}.
\end{equation*}
The outside option in the baseline model is an interim participation payoff.
Applications with richer exit timing, such as exit after the discretionary
action is observed, specialize the payoff functions so that the relevant
exit decision is already embedded in the continuation payoff, as shown later.

For any $f\in \Delta (K)$, let supp$\left[ f\right] $ denote the support of $%
f$ whenever $K$ is endowed with a topology. We impose a mild non-triviality
condition.

\begin{assum}
\label{assm:non-trivial}%
\begin{equation*}
\left\vert \left\{ (x_{j},y_{j})\in X_{j}\times Y_{j}:y_{j}\in
F_{j}(x_{j})\right\} \right\vert \geq 2,\qquad \forall j\in \mathcal{J}.
\end{equation*}
\end{assum}

If this condition failed for some principal, that principal would have only
one feasible action pair and could be removed from the model. We impose no
restriction on whether $\Theta$, $X_j$, or $Y_j$ is finite or infinite.

\subsection{Contracts and continuation play}

A contract offered by principal $j$ is a mechanism $G_{j}:M_{j}\rightarrow
X_{j},$ where $M_{j}$ is an admissible standard Borel message space and $%
G_{j}$ is measurable. A contract is admissible only if its image $%
L_{j}:=G_{j}(M_{j})$ is Borel, the set $\{(x_{j},y_{j}):x_{j}\in L_{j},\,
y_{j}\in F_{j}(x_{j})\}$ is Borel in $X_{j}\times Y_{j}$, the feasible
correspondence admits a measurable selector on $L_j$, and $G_{j}$ admits a
measurable selector on its image. These regularity requirements are
automatic in finite environments and are maintained in the general standard
Borel case. The mechanism specifies only the contractible action. The
non-contractible action is chosen later.

The timing is as follows. First, principals simultaneously and publicly
choose contracts. Second, after observing the contract profile, the agent
sends messages publicly. Third, after observing those messages, principals
simultaneously choose their non-contractible actions.

Fix a contract profile
\begin{equation*}
G=(G_{j}:M_{j}\rightarrow X_{j})_{j\in \mathcal{J}},\qquad M:=\times _{j\in
\mathcal{J}}M_{j}.
\end{equation*}%
The agent's messaging strategy is $q:\Theta \rightarrow \Delta (M)$.
Principal $j$'s continuation strategy is a map $\gamma _{j}:M\rightarrow
Y_{j}$ such that
\begin{equation*}
\gamma _{j}(m)\in F_{j}(G_{j}(m_{j}))\quad \text{for every }m=(m_{k})_{k\in
\mathcal{J}}\in M.
\end{equation*}%
Let $\gamma =(\gamma _{j})_{j\in \mathcal{J}}$.

In the one-principal case, continuation equilibrium has the following
economic content. Given $G$, the agent chooses messages optimally, taking
into account both the contractible action $G(m)$ and the continuation action
$\gamma(m)$. The principal, after each message, chooses $\gamma(m)$
optimally given the posterior belief induced by the agent's strategy. The
multi-principal definitions below are the same idea with a message profile
and simultaneous continuation choices.

\subsection{Equilibrium for a fixed contract space}

The object of interest is the set of allocations that can arise when
principals are allowed to choose arbitrary mechanisms. As in standard
canonical-mechanism arguments, we do not take a literal universal set of
contracts as a primitive. Instead, we first define equilibrium relative to
an arbitrary contract space and then ask whether a smaller, more tractable
contract space is without loss for equilibrium allocations.

Let $\mathcal{G}_{j}$ be a non-empty set of admissible contracts available
to principal $j$, and let $\mathcal{G}:=\times _{j\in \mathcal{J}}\mathcal{G}%
_{j}.$ Given a contract profile $G=(G_{j}:M_{j}\rightarrow X_{j})_{j\in
\mathcal{J}}\in \mathcal{G}$, let $M:=\times _{j\in \mathcal{J}}M_{j}$. The
agent's messaging strategy is denoted by $q:\Theta \rightarrow \Delta (M)$.
Principal $j$'s continuation strategy is denoted by $\gamma _j$ and must
select a feasible non-contractible action after every message profile.
Principal $j$'s posterior belief after message profile $m$ is denoted by $%
p_j(m)$. A profile $[G,(\gamma,q,p)]$ collects the contract profile,
continuation actions, the agent's messaging strategy, and posterior beliefs.

\begin{define}[$\mathcal{G}$-continuation equilibrium]
\label{def:weak-equilibrium}A profile $[G,(\gamma ,q,p)]$ is a $\mathcal{G}$%
-continuation equilibrium if beliefs are Bayesian consistent with the
agent's messaging strategy, the agent's messages are optimal and satisfy
interim participation, and each principal's discretionary action is optimal
after every message profile given principal $j$'s posterior belief and the
other principals' continuation actions.
\end{define}

The formal version of this definition is recorded in Appendix \ref%
{app:formal_continuation_equilibrium}.

For any $\mathcal{G}$-continuation equilibrium $[G,(\gamma ,q,p)]$,
principal $j$'s ex ante payoff is
\begin{equation*}
V_j[G,(\gamma ,q,p)]:=\int_{\Theta }\left[ \int_{M}v_j\left(
(G_k(m_k))_{k\in \mathcal{J}},(\gamma _k(m))_{k\in \mathcal{J}},\theta
\right) q(\theta )[dm]\right] \mu \lbrack d\theta ].
\end{equation*}

The equilibrium induces an allocation
\begin{equation*}
z:\Theta \rightarrow \Delta (X\times Y)
\end{equation*}%
defined by, for every measurable $E\subseteq X\times Y$,
\begin{equation*}
z(\theta )[E]=q(\theta )\left[ \left\{ m\in M:\left( (G_{k}(m_{k}))_{k\in
\mathcal{J}},(\gamma _{k}(m))_{k\in \mathcal{J}}\right) \in E\right\} \right]
.
\end{equation*}%
When $|\mathcal{J}|=1$, this continuation equilibrium is the standard
solution concept in the limited-commitment literature (see, e.g., \cite%
{bs2001}).

With multiple principals, equilibrium also requires optimal contract choice
at the first stage.

\begin{define}[$\mathcal{G}$-robust equilibrium]
\label{def:equilibrium:robust}A profile $[G,(\gamma ,q,p)]$ is a $\mathcal{G}
$-robust equilibrium if it is a $\mathcal{G}$-continuation equilibrium and,
for every principal $j\in \mathcal{J}$ and every alternative contract $%
G_{j}^{\prime }\in \mathcal{G}_{j}$, there exists a $\mathcal{G}$%
-continuation equilibrium $[(G_{j}^{\prime },G_{-j}),(\gamma ^{\prime
},q^{\prime },p^{\prime })]$ such that
\begin{equation*}
V_{j}[G,(\gamma ,q,p)]\geq V_{j}[(G_{j}^{\prime },G_{-j}),(\gamma ^{\prime
},q^{\prime },p^{\prime })].
\end{equation*}
\end{define}

This is the no-safe-deviation criterion used in competing-mechanism
environments (see, e.g., \cite{mp2001}). A deviating principal does not
choose which continuation equilibrium follows the deviation. Thus a
deviation defeats robustness only if all continuation equilibria following
it give the deviating principal a higher payoff. If one continuation
equilibrium leaves the deviating principal no better off, the original
profile survives that deviation within the contract space $\mathcal{G}$. All
robustness statements below use this equilibrium-selection criterion.

\subsection{Unrestricted equilibrium}

Definitions \ref{def:weak-equilibrium} and \ref{def:equilibrium:robust} are
stated relative to a given contract space. Our main interest, however, is in
equilibrium allocations when principals are not restricted to a particular
class of mechanisms.

To avoid treating a universal set of mechanisms as a primitive, we define
unrestricted equilibrium notions indirectly.

\begin{define}[continuation equilibrium]
\label{def:equilibrium:continuation} A profile $\left[G,(\gamma,q,p)\right]$
is a \emph{continuation equilibrium} if it is a $\mathcal{G}$-\emph{%
continuation equilibrium} for every contract space $\mathcal{G}$ such that $%
G\in\mathcal{G}$.
\end{define}

\begin{define}[robust equilibrium]
\label{def:equilibrium:robust:without} A profile $\left[G,(\gamma,q,p)\right]
$ is a \emph{robust equilibrium} if it is a $\mathcal{G}$-\emph{robust
equilibrium} for every contract space $\mathcal{G}$ such that $G\in\mathcal{G%
}$.
\end{define}

These definitions let principals be evaluated against arbitrary admissible
deviations without requiring an explicit universal contract space. A
continuation equilibrium imposes optimal behavior in the continuation game
after the chosen contract profile. A robust equilibrium additionally
requires that no principal can profit by deviating to any admissible
mechanism. The purpose of the extended taxation principle developed below is
to replace the unrestricted contract space with a tractable class of menus
with recommendations, without changing the set of equilibrium allocations.

\section{The Extended Taxation Principle}

\label{sec:canonical}

The previous section defined the general limited-commitment environment.
This section proves the representation result used above and in the
applications below. The unrestricted problem allows principals to choose
arbitrary contracts. For analysis, however, we would like to reduce
attention to a smaller and more transparent class of contracts.

The central observation is that, under limited commitment, a menu over
contractible actions alone is generally too coarse. Two messages may induce
the same contractible action but support different posterior beliefs and
therefore different continuation choices of the non-contractible action. A
canonical contract must therefore keep track not only of the contractible
choice, but also of the continuation outcome it is meant to sustain. This is
why the right object is a \emph{menu with recommendations} rather than a
standard menu of contractible actions.

\subsection{Canonicality}

Fix the contractual primitives
\begin{equation*}
\left\langle \Theta ,\mathcal{J},(X_{j},Y_{j},F_{j})_{j\in \mathcal{J}%
}\right\rangle .
\end{equation*}%
These primitives describe the type space, the set of principals, and which
parts of each principal's action are contractible or discretionary. A
canonical contract space should depend only on this contractual structure,
not on the particular prior or payoff functions. We summarize the latter by
the payoff environment
\begin{equation*}
\mathcal{E}:=\langle \mu ,U,u,(v_{k})_{k\in \mathcal{J}}\rangle .
\end{equation*}

Given $\mathcal{E}$ and a contract space $\mathcal{G}$, let
\begin{equation*}
\mathcal{Z}^{\text{continuation-}\mathcal{G}\text{-}\mathcal{E}}
\end{equation*}
denote the set of allocations induced by $\mathcal{G}$-continuation
equilibria, and let
\begin{equation*}
\mathcal{Z}^{\text{robust-}\mathcal{G}\text{-}\mathcal{E}}
\end{equation*}
denote the set of allocations induced by $\mathcal{G}$-robust equilibria.
Finally, let
\begin{equation*}
\mathcal{Z}^{\text{robust-}\mathcal{E}}
\end{equation*}
denote the set of allocations induced by unrestricted robust equilibria.

\begin{define}
\label{define:canonical:weak} A contract space $\mathcal{G}$ is \emph{%
canonical for continuation equilibrium} if, for every payoff environment $%
\mathcal{E}$ and every contract space $\mathcal{G}^{\prime}$,
\begin{equation*}
\mathcal{Z}^{\text{continuation-}\mathcal{G}^{\prime}\text{-}\mathcal{E}}
\subseteq \mathcal{Z}^{\text{continuation-}\mathcal{G}\text{-}\mathcal{E}}.
\end{equation*}
\end{define}

\begin{define}
\label{define:canonical:robust} A contract space $\mathcal{G}$ is \emph{%
canonical for robust equilibrium} if, for every payoff environment $\mathcal{%
E}$,
\begin{equation*}
\mathcal{Z}^{\text{robust-}\mathcal{G}\text{-}\mathcal{E}} = \mathcal{Z}^{%
\text{robust-}\mathcal{E}}.
\end{equation*}
\end{define}

The first notion says that restricting attention to $\mathcal{G}$ entails no
loss for continuation-equilibrium allocations. The second says that $%
\mathcal{G}$ captures exactly the allocations that can arise in unrestricted
robust equilibrium.

\begin{remark}
Because equilibrium strategies depend on the contract space, the relevant
object is a canonical contract space for equilibrium \emph{allocations}, not
for equilibrium \emph{profiles}.
\end{remark}

\subsection{Menus with recommendations}

Many canonical contract spaces exist. The goal is to find one that is both
without loss and as small as possible. Minimality here means minimal
cardinality, not minimality under set inclusion.

\begin{define}
A canonical contract space $\mathcal{G}$ is \emph{minimal for continuation
equilibrium} if, for any canonical contract space $\mathcal{G}^{\prime}$ for
continuation equilibrium,
\begin{equation*}
|\mathcal{G}^{\prime}|\geq |\mathcal{G}|.
\end{equation*}
\end{define}

\begin{define}
A canonical contract space $\mathcal{G}$ is \emph{minimal for robust
equilibrium} if, for any canonical contract space $\mathcal{G}^{\prime}$ for
robust equilibrium,
\begin{equation*}
|\mathcal{G}^{\prime}|\geq |\mathcal{G}|.
\end{equation*}
\end{define}

Let $\mathcal{B}(X_{k})$ and $\mathcal{B}(X_{k}\times Y_{k})$ denote the
Borel $\sigma $-algebras on $X_{k}$ and $X_{k}\times Y_{k}$ respectively.
For a non-empty Borel message space $M_{k}\in \mathcal{B}(X_{k}\times
Y_{k}), $ let $\widehat{G}|_{M_{k}}$ denote the function with a domain $%
M_{k} $ such that
\begin{equation*}
\widehat{G}|_{M_{k}}(x_{k},y_{k})=x_{k},\text{ }\forall (x_{k},y_{k})\in
M_{k}
\end{equation*}%
For each non-empty Borel set $L_{k}\subseteq X_{k}$, let
\begin{equation*}
M_{k}(L_{k}):=\{(x_{k},y_{k}):x_{k}\in L_{k},\ y_{k}\in F_{k}(x_{k})\}.
\end{equation*}%
Define the admissible measurable menus of contractible actions by
\begin{align*}
\mathcal{L}_{k}^{\ast }:=\{L_{k}\in \mathcal{B}(X_{k}):& \ L_{k}\neq
\varnothing ,\ M_{k}(L_{k})\in \mathcal{B}(X_{k}\times Y_{k}), \\
& \text{and }M_{k}(L_{k})\text{ admits a measurable selector }%
s_{k}:L_{k}\rightarrow M_{k}(L_{k})\}.
\end{align*}%
We define
\begin{gather*}
\mathcal{M}_{k}^{\ast }:=\{M_{k}(L_{k}):L_{k}\in \mathcal{L}_{k}^{\ast }\},
\\
\mathcal{G}_{k}^{\ast }:=\left\{ \langle G_{k}:M_{k}\longrightarrow
X_{k}\rangle :M_{k}\in \mathcal{M}_{k}^{\ast }\text{ and }G_{k}=\widehat{G}%
|_{M_{k}}\right\} , \\
\mathcal{G}^{\ast }:=\times _{k\in \mathcal{J}}\mathcal{G}_{k}^{\ast }.
\end{gather*}

A contract in $\mathcal{G}_{j}^{\ast }$ offers the agent a non-empty set of
feasible pairs $(x_{j},y_{j})$ such that, whenever one pair with
contractible action $x_{j}$ is available, all feasible recommendations $%
y_{j}\in F_{j}(x_{j})$ are available as well. The agent therefore chooses a
contractible action and, at the same time, recommends a continuation action.
For this reason, contracts in $\mathcal{G}_{j}^{\ast }$ can be interpreted
as \emph{menus with recommendations}. The closure condition means that each
message set is determined by its projection onto contractible actions: when $%
x_{j}$ is made available, all feasible recommendations for $x_{j}$ are
included. Thus one can think of a contract as a menu of contractible
actions. For any $L_{j}\in \mathcal{L}_{j}^{\ast }$, the corresponding
message set is
\begin{equation*}
\left\{ (x_{j},y_{j}):x_{j}\in L_{j},\ y_{j}\in F_{j}(x_{j})\right\} ,
\end{equation*}%
and the mechanism commits only to the selected $x_{j}$. The recommended $%
y_{j}$ is not binding; it records the continuation action that the message
is meant to induce. When $X_{j}$ is finite, $\mathcal{L}_{j}^{\ast
}=2^{X_{j}}\diagdown \{\varnothing \}$.

\setcounter{eg}{3}

\begin{eg}
\label{example1} Suppose $\mathcal{J}=\{j_{1},j_{2}\}$. For each principal $%
j\in \mathcal{J}$, let
\begin{equation*}
X_{j}=\{x_{j},x_{j}^{\prime }\}\qquad \text{and}\qquad
Y_{j}=\{y_{j},y_{j}^{\prime },y_{j}^{\prime \prime },y_{j}^{\prime \prime
\prime }\}
\end{equation*}%
with
\begin{equation*}
F_{j}(x_{j})=\{y_{j},y_{j}^{\prime }\}\qquad \text{and}\qquad
F_{j}(x_{j}^{\prime })=\{y_{j}^{\prime },y_{j}^{\prime \prime
},y_{j}^{\prime \prime \prime }\}.
\end{equation*}%
Then
\begin{eqnarray}
\mathcal{M}_{j}^{\ast } &=&\left\{
\begin{array}{l}
M_{j}^{(1)}:=\{(x_{j},y_{j}),(x_{j},y_{j}^{\prime })\}, \\
M_{j}^{(2)}:=\{(x_{j}^{\prime },y_{j}^{\prime }),(x_{j}^{\prime
},y_{j}^{\prime \prime }),(x_{j}^{\prime },y_{j}^{\prime \prime \prime })\},
\\
M_{j}^{(3)}:=\{(x_{j},y_{j}),(x_{j},y_{j}^{\prime }),(x_{j}^{\prime
},y_{j}^{\prime }),(x_{j}^{\prime },y_{j}^{\prime \prime }),(x_{j}^{\prime
},y_{j}^{\prime \prime \prime })\}%
\end{array}%
\right\} ,  \notag \\
\mathcal{G}_{j}^{\ast } &=&\left\{ \widehat{G}|_{M_{j}^{(1)}},\widehat{G}%
|_{M_{j}^{(2)}},\widehat{G}|_{M_{j}^{(3)}}\right\} .  \label{ktt1}
\end{eqnarray}%
Thus principal $j$ has exactly three contracts in $\mathcal{G}_{j}^{\ast }$:
one for each non-empty subset of contractible actions.
\end{eg}

The next two theorems show that $\mathcal{G}^{\ast }$ is not only canonical,
but minimal. The upper-bound idea is that any arbitrary mechanism can be
replicated by recording, for each message, both the contractible action it
selects and the continuation action it induces. The lower-bound idea is that
when there are enough agent types, every admissible measurable menu of
contractible actions can be made behaviorally necessary in some payoff
environment. A canonical contract space that omitted one of these menus
would fail for that environment. All cardinalities are set-theoretic
cardinalities; when action spaces are infinite, the relevant menus are the
admissible measurable menus in $\mathcal{L}_{j}^{\ast }$, not arbitrary
subsets of $X_{j}$. The assumption $|\Theta |\geq |X|$ is used for this
lower-bound construction; the formal proof is in the appendix.

\begin{theo}
\label{thm:minimal:weak} Suppose $|\Theta|\geq |X|$ and $|\mathcal{J}|=1$.
Then $\mathcal{G}^\ast$ is a minimal canonical contract space for
continuation equilibrium.
\end{theo}

\begin{theo}
\label{thm:minimal:weak-robust} Suppose $|\Theta|\geq |X|$ and $|\mathcal{J}%
|\geq 2$. Then $\mathcal{G}^\ast$ is a minimal canonical contract space for
robust equilibrium.
\end{theo}

The economic content is simple. This is the limited-commitment analogue of
the taxation principle. Under full commitment, a menu of contractible
outcomes is enough because the principal controls the entire outcome. Under
imperfect commitment, a menu of contractible actions alone discards
continuation-relevant information. The additional recommendation records the
beliefs and continuation behavior that a mechanism can induce without making
the discretionary action enforceable. In Example \ref{example1}, this means
that principal $j$ need only consider the three contracts in (\ref{ktt1}).
Equivalently, principal $j$ offers a non-empty subset of contractible
actions and allows the agent to recommend any feasible continuation action
conditional on that choice. In any $\mathcal{G}^{\ast }$-continuation
equilibrium, principals follow the agent's non-binding recommendations on
the equilibrium path; the recommendations thus record the continuation
behavior that each contractible action is designed to induce rather than
constrain it.

\subsection{Implementation of robust-equilibrium allocations}

Theorem \ref{thm:minimal:weak-robust} is an allocation result: the reduced
contract space $\mathcal{G}^\ast$ captures the unrestricted robust
equilibrium allocations. A profile-level distinction remains. A profile that
is robust relative to $\mathcal{G}^\ast$ need not itself be robust against
every admissible contract outside $\mathcal{G}^\ast$, because a deviation
may remove the recommendation messages used to support post-deviation
continuation play.

This issue can be resolved without changing the induced allocation. Enlarge
each canonical mechanism by an auxiliary message coordinate that records the
original message profile, while keeping the committed contractible action
unchanged. The added coordinate restores the off-path information needed for
the no-safe-deviation test, and the formal construction yields the following
profile-implementation result.

\begin{prop}
\label{prop:robust:equilibrium} Suppose $|J|\geq 2$. For any $\mathcal{G}%
^{\ast }$-robust equilibrium $\left[ (G_{k}:M_{k}\rightarrow X_{k})_{k\in
J},(\gamma ,q,p)\right] $ that induces allocation $z$, let $M=\times _{k\in
J}M_{k}.$ Then there exists an unrestricted robust equilibrium $\left[
(G_{k}^{\circ }:M_{k}\times M\rightarrow X_{k})_{k\in J},(\gamma ^{\prime
},q^{\prime },p^{\prime })\right] $ that also induces $z$, where
\begin{equation*}
G_{k}^{\circ }(m_{k},\widetilde{m})=G_{k}(m_{k}),\qquad \forall (m_{k},%
\widetilde{m})\in M_{k}\times M,\quad \forall k\in J.
\end{equation*}
\end{prop}


Proposition \ref{prop:robust:equilibrium} shows that the canonical contract
space is sufficient for characterizing robust-equilibrium allocations, even
if some equilibrium profiles require a richer off-path message space for
unrestricted implementation.

\section{Distinct Discretionary Actions and Exit}

\label{sec:labor_apps}

\subsection{A Single Principal's Optimal Contract\label{sec:appI}}

Section \ref{sec:revisable} considered bounded revisions of a committed
baseline. This section turns to environments in which the discretionary
choice is a distinct payoff-relevant instrument rather than a revision of
the same final action. The principal contracts on one instrument, cannot
commit to another, and the agent can exit after observing the discretionary
choice. The full-commitment equivalence need not apply, but the extended
taxation principle still sharply simplifies the contracting problem.

Labor contracting is the running interpretation, with compensation as the
contractible action and workload as the non-contractible action. The same
structure also covers settings in which a loan term, insurance coverage,
wholesale price, or regulatory standard is contractible while monitoring,
claim handling, service quality, or enforcement is chosen later.

\paragraph{Model.}

There is one principal and one agent. The agent's type is $%
\theta\in\Theta\subset\mathbb{R}$, where $\Theta$ is an interval, and the
prior $\mu$ is atomless.

The principal's utility function and the agent's utility function are
\begin{eqnarray}
&&v(x,y,\theta),  \label{v_quasi} \\
&&u(x,y,\theta),  \label{u_quasi}
\end{eqnarray}
respectively. Here, $x\in\mathbb{R}$ is a contractible action and $y\in%
\mathbb{R}_{+}$ is a non-contractible action. The agent's reservation
utility is normalized to zero. If trade does not occur, the principal also
receives zero.

After the principal chooses $y$, the agent can continue or walk away. Define
the reduced-form payoffs
\begin{equation*}
\bar{u}(x,y,\theta ):=%
\begin{cases}
u(x,y,\theta ) & \text{if }u(x,y,\theta )\geq 0, \\
0 & \text{otherwise,}%
\end{cases}%
\text{ \hspace{0.15in}and }\hspace{0.15in}\bar{v}(x,y,\theta ):=%
\begin{cases}
v(x,y,\theta ) & \text{if }u(x,y,\theta )\geq 0, \\
0 & \text{otherwise.}%
\end{cases}%
\end{equation*}%
The assumptions below are imposed on the primitive payoffs $u$ and $v$ on
the continuation branch; exit enters through these reduced-form payoffs.

\paragraph{Assumptions.}

We impose the following monotonicity and sorting assumptions.

\begin{assum}
\label{assu_monotone_in_theta} $u(x,y,\theta)$ is strictly increasing in $%
\theta$, and $v(x,y,\theta)$ is non-decreasing in $\theta$.
\end{assum}

Higher types are weakly more valuable to the principal and strictly better
off at any fixed contractual outcome.

\begin{assum}
\label{assu_monotone_in_y} $v(x,y,\theta)$ is increasing in $y$, and $%
u(x,y,\theta)$ is decreasing in $y$.\footnote{%
If instead $v(x,y,\theta)$ is decreasing in $y$ and $u(x,y,\theta)$ is
increasing in $y$, one can redefine $y^{\prime}=-y$ and apply the same
analysis. For example, in a monopoly problem, $y$ may represent product
quality: the seller's payoff is decreasing in quality, whereas the buyer's
utility is increasing in quality.}
\end{assum}

The non-contractible action creates the key tension: the principal prefers a
higher $y$, whereas the agent prefers a lower $y$.

We also impose a global single-crossing condition on the agent's preferences.

\begin{assum}[Single Crossing Property]
\label{assu_single_crossing} For any two distinct pairs $a=(x,y)$ and $%
b=(x^{\prime},y^{\prime})$, the set $\{\theta:u(a,\theta)\geq u(b,\theta)\}$
is an interval, possibly empty or all of $\Theta$, and the indifference set $%
\{\theta:u(a,\theta)=u(b,\theta)\}$ contains at most one type.
\end{assum}

Assumption \ref{assu_single_crossing} is an ordinal single-crossing
restriction on the agent's preferences over action pairs. It is implied by
the usual differentiable Spence-Mirrlees single-crossing condition in
environments where preference differences cross at most once; see \cite%
{gs1996}.

\paragraph{Equilibrium class: atomic participating action pairs.}

The applications of interest use finite menus or piecewise-pooling
mechanisms. For the general statement, we impose an atomicity condition that
captures those cases exactly and keeps the pooling argument exact rather
than approximate. Fix a contract in $\mathcal{G}^{\ast }$ and a truthful
continuation equilibrium in which the principal follows the recommendation
on the equilibrium path. Let $L$ be the menu of contractible actions and let
$q:\Theta \rightarrow \Delta (L\times \mathbb{R}_{+})$ be the agent's
continuation strategy over chosen contractible actions and recommended
non-contractible actions. For each type $\theta $, let
\begin{equation*}
P_{\theta }:=\{(x,y)\in L\times \mathbb{R}_{+}:u(x,y,\theta )\geq 0\}.
\end{equation*}%
Define the aggregate measure over participating action pairs by
\begin{equation*}
Q^{P}(A):=\int_{\Theta }q(\theta )[A\cap P_{\theta }]\,\mu (d\theta )
\end{equation*}%
for every Borel set $A\subseteq L\times \mathbb{R}_{+}$. We call the
equilibrium an \emph{atomic participating-action-pair equilibrium} if,
whenever $Q^{P}$ is not concentrated, up to null sets, on a single action
pair, there exist (at least) two distinct action pairs $a,b\in L\times
\mathbb{R}_{+}$ such that $Q^{P}(\{a\})>0$ and $Q^{P}(\{b\})>0$. This
condition is automatic when the participating action-pair measure has finite
or countable support, including finite or countable menu-with-recommendation
message spaces and pure piecewise-pooling mechanisms. Diffuse
thin-randomization schemes can be treated by approximation under stronger
regularity; Appendix \ref{sec:proof_proposition_single} records that
interpretation.

\paragraph{Main results.}

By Theorem \ref{thm:minimal:weak}, we can restrict attention to menus with
recommendations. The agent chooses a contractible action $x$ and recommends
a non-contractible action $y$; by Lemma \ref{lem:extension_w2}, we can focus
on truthful continuation equilibria in which the principal follows the
recommendation on the equilibrium path.

The first result shows that, despite the apparent richness of the
contracting problem, the principal can restrict attention to a particularly
simple class of contracts.

\begin{prop}
\label{proposition_single} Under Assumptions \ref{assu_monotone_in_theta}--%
\ref{assu_single_crossing}, together with the maintained local uniform
continuity and integrability conditions used in Lemma \ref%
{lem:rent_extraction_common_slack}, the single-principal optimum over atomic
participating-action-pair equilibria is attainable with a single
contractible-action offer $x$.
\end{prop}

The proposition is exact for the atomic equilibrium class. The restriction
selects the finite and pooling mechanisms that carry the economic argument;
Appendix \ref{sec:proof_proposition_single} explains how diffuse
thin-randomization schemes can be read by approximation under stronger
regularity.

To visualize the force of Proposition \ref{proposition_single}, fix a
truthful continuation equilibrium induced by a menu with recommendations $G$%
, with continuation strategy $q:\Theta \rightarrow \Delta (L\times \mathbb{R}%
_{+})$ over chosen contractible actions and recommended non-contractible
actions, where $L$ denotes the menu of contractible actions that the agent
can choose given $G$, and let $Q^{P}$ be the aggregate participating
action-pair measure defined above. For each $(x,y)$ with $Q^{P}(\{(x,y)\})>0$%
, define
\begin{equation*}
\Theta (x,y):=\left\{ \theta \in \Theta :q(\theta )[\{(x,y)\}]>0\text{ and
the agent does not exit the contract}\right\} .
\end{equation*}%
Thus, $\Theta (x,y)$ is the set of types that choose $(x,y)$ with positive
probability and remain in the relationship.

\begin{figure}[tbp]
\centering
\includegraphics[scale=0.6]{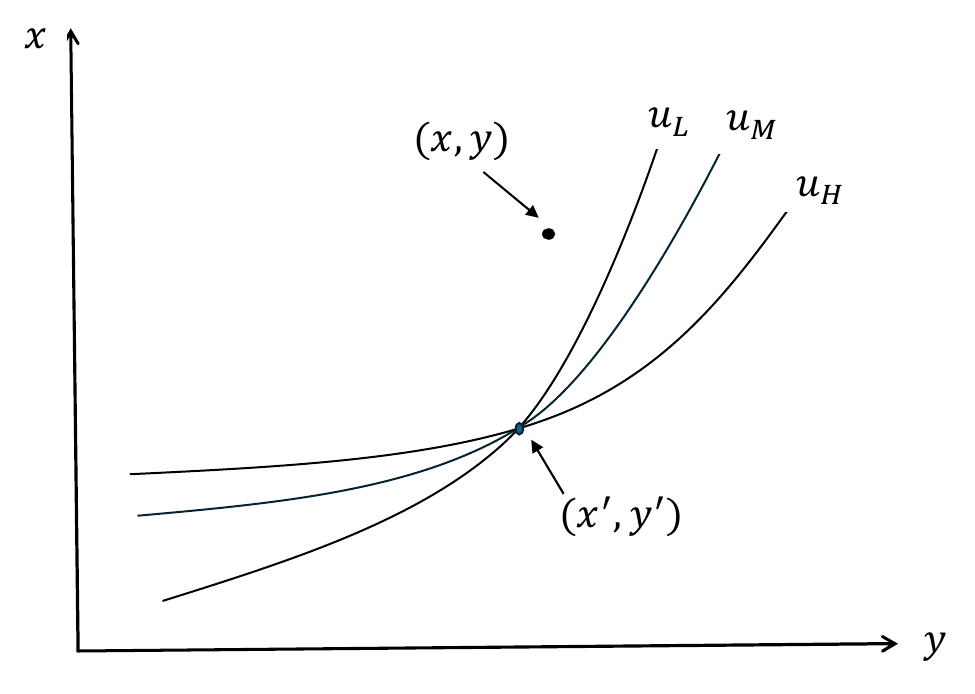}
\caption{A schematic illustration of why two positive-measure pooling
regions are inconsistent}
\label{fig:pooling}
\end{figure}

Figure \ref{fig:pooling} illustrates why two positive-measure pooling
regions must be ordered by type. Suppose instead that they are interlaced:
low and high types choose one pair, say $(x,y)$, while an intermediate type
chooses another pair, $(x^{\prime },y^{\prime })$. Since $(x,y)$ is
acceptable to both the low and high types, single crossing implies that it
lies above the intermediate type's relevant indifference curve as well. The
intermediate type would therefore strictly prefer $(x,y)$ to $(x^{\prime
},y^{\prime })$, contradicting incentive compatibility. Hence any two
positive-measure pooling regions must be ordered.

Nor can ordered positive-measure pooling regions be sustained. In any higher
pooling region, participating types obtain rents above their reservation
utilities. The principal can therefore raise the non-contractible action
slightly, which increases the principal's payoff and lowers the agent's
payoff, while preserving participation for those types. Appendix \ref%
{sec:proof_proposition_single} formalizes this step as a rent-extraction
lemma from common slack. Thus, within the atomic equilibrium class, any menu
with recommendations generates the payoff of a single participating action
pair, and hence of a single contractible-action offer.

This is the distinct-instrument counterpart to the revisable-action
equivalence. Because $y$ is not a revision of $x$, the conclusion is not
that full commitment is reproduced. Instead, exit and sorting make rich
communication payoff-equivalent to a simple contractible-action offer.

Within the atomic equilibrium class, Proposition \ref{proposition_single}
reduces the optimal contracting problem to a single offer $x$. Once such an
offer is made, the relevant continuation decision is whether the agent stays
or exits after observing the principal's choice of $y$. For any $(x,y)\in
\mathbb{R}\times \mathbb{R}_{+}$, let
\begin{equation*}
S(x,y):=\{\theta \in \Theta :u(x,y,\theta )\geq 0\}
\end{equation*}%
be the set of types that remain. Since $u(x,y,\theta )$ is increasing in $%
\theta $, this set is either empty, or all of $\Theta $, or an upper
interval. When the cutoff is interior, its lower endpoint $\theta (x,y)$
satisfies $u(x,y,\theta (x,y))=0$.

Define $C:=\left\{ (x,y)\in \mathbb{R}\times \mathbb{R}_{+}:\exists \theta
\in \Theta \text{ such that }u(x,y,\theta )\geq 0\right\} ,$and, for each $%
x\in \mathbb{R}$,
\begin{equation*}
C(x):=\{y\in \mathbb{R}_{+}:(x,y)\in C\}.
\end{equation*}%
Let $B:=\{x\in \mathbb{R}:C(x)\neq \emptyset \}.$

The second result characterizes the principal's optimal single offer.

\begin{prop}
\label{proposition_single_optimal} A single contractible-action offer $%
x^{\ast }$ is optimal within the atomic equilibrium class if and only if
\begin{equation}
x^{\ast }\in \arg \max_{x\in B}\left[ \max_{y\in C(x)}\int_{\Theta }\mathbf{1%
}_{\{u(x,y,\theta )\geq 0\}}v(x,y,\theta )\,\mu \left[ d\theta \right] %
\right] .  \label{optimal_contract}
\end{equation}
\end{prop}

Proposition \ref{proposition_single_optimal} gives a simple two-step
procedure. For each contractible action $x$, the principal first chooses the
non-contractible action $y$ to maximize expected payoff, taking into account
the set of types that remain after observing $y$. The principal then chooses
the contractible action $x$ that yields the highest value of this
continuation problem.

The result is not tied to labor markets: it applies whenever the primitive
payoffs satisfy the monotonicity and single-crossing conditions, without
requiring separability.

\paragraph{Application to the labor-contracting example.}

We now return to the labor example introduced in Section \ref%
{sec:motivating_example}. Let
\begin{equation*}
v(x,y,\theta)=y\theta-x^{2},\qquad u(x,y,\theta)=\frac{x\theta-y^{2}}{\sqrt{%
\theta}},
\end{equation*}
where $x$ is base compensation, $y$ is work intensity, and $\theta\sim%
\mathrm{Unif}[3,4]$.

Applying Proposition \ref{proposition_single_optimal}, we restrict attention
to a single compensation offer $x$. If $x\leq 0$, no type accepts except
trivially at $(0,0)$, so we focus on $x>0$. The cutoff type is determined by
\begin{equation*}
x\theta -y^{2}=0\qquad \Longrightarrow \qquad \theta (x,y)=\frac{y^{2}}{x}.
\end{equation*}%
Hence, the worker stays if and only if $\theta \geq y^{2}/x$.

Let $t:=\theta(x,y)=y^{2}/x$. Since $\theta\in[3,4]$,
\begin{equation*}
\Pr(\text{stay}) = \Pr(\theta\geq t) = \left\{
\begin{array}{ll}
1 & \text{if } t\leq 3, \\
4-t & \text{if } 3<t<4, \\
0 & \text{if } t\geq 4.%
\end{array}
\right.
\end{equation*}
Using $y=\sqrt{xt}$, the principal's expected payoff can be written as a
function of $(x,t)$.

For $t\in \lbrack 3,4)$, expected profit is
\begin{equation*}
\Pi (x,t)=\Pr (\theta \geq t)\,\mathbb{E}\left[ y\theta -x^{2}\mid \theta
\geq t\right] =(4-t)\left( \sqrt{xt}\cdot \frac{t+4}{2}-x^{2}\right)
=(4-t)\left( K(t)\sqrt{x}-x^{2}\right) ,
\end{equation*}%
where
\begin{equation*}
K(t):=\frac{\sqrt{t}(t+4)}{2}.
\end{equation*}

For $t\leq 3$, all types stay and the principal's payoff is increasing in $t$%
, so the best no-exit cutoff is $t=3$. For each $t\in[3,4)$, choose $%
x^{\ast}(t)$ to maximize $\Pi(x,t)$, and then choose $t^{\ast}$ to maximize $%
\Pi(x^{\ast}(t),t)$. The first-order condition gives $x^{\ast}(t)=\left(
K(t)/4\right) ^{2/3}$, and substituting this expression shows that the
continuation value is maximized at $t^{\ast}=3$. The unique solution is
\begin{equation*}
x^{\ast}=\left(\frac{7\sqrt{3}}{8}\right)^{2/3}=1.32, \qquad y^{\ast}=\sqrt{%
3x^{\ast}}=1.99,
\end{equation*}
so all types remain in equilibrium.

Thus the optimal rich mechanism in this class is payoff-equivalent to a
single offer of base compensation, with work intensity chosen so that the
lowest participating type is just indifferent.

\subsection{Robustness of Simple Contracts in Common Agency \label{sec:appII}%
}

We now ask whether the single-offer structure survives common agency, where
several principals compete for the same privately informed agent and
important performance dimensions remain non-contractible. The answer is yes
under principal-specific exit and separable agent utility across
relationships.

\paragraph{A motivating example.}

Consider a worker who can work for two firms. Each firm can contract on a
contractible support term $x_{j}$, but not on eventual work intensity $y_{j}$%
. After observing realized intensities, the worker decides separately for
each firm whether to remain. The question is whether a firm can profit by
replacing a simple support offer with an arbitrary mechanism, such as a menu
of support terms and recommendations. Proposition \ref{prop_pooling_common}
shows that it cannot do so in the robust-equilibrium sense.

\paragraph{Model.}

There are two principals, indexed by $j=1,2$, and one agent. The agent's
type is $\theta\in\Theta=[\theta_{\min},\theta_{\max}]\subset\mathbb{R}$,
and the common prior is an atomless distribution $\mu\in\Delta(\Theta)$.

For each principal $j$, let $x_{j}\in \mathbb{R}$ denote the contractible
action and $y_{j}\in \mathbb{R}_{+}$ denote the non-contractible action. Let
$x=(x_{1},x_{2})$ and $y=(y_{1},y_{2})$. Principal $j$'s payoff is $%
v_{j}(x,y,\theta )$, while the agent's utility is $u(x,y,\theta ).$

The agent can remain with or walk away from each principal separately. If
the agent exits from principal $j$, both the agent and principal $j$ receive
zero from that bilateral relationship.

\paragraph{Assumptions.}

The common-agency extension uses a bilateral version of the single-principal
conditions, together with separability across relationships. We state the
condition in bundled form here; the clause-by-clause formulation is recorded
in the Supplemental Appendix.

\begin{assum}[Common-agency reduction]
\label{Ass_multi_bilateral_reduction} The agent's utility is separable
across principals, $u(x,y,\theta)=u_1(x_1,y_1,\theta)+u_2(x_2,y_2,\theta)$.
Each principal's payoff can be written as $v_j(x,y,\theta)=\widetilde
v_j(x_j,y_j,x_{-j},\theta)$, and if the agent exits from principal $j$ then
principal $j$ receives zero. For each principal $j$ and each fixed $x_{-j}$,
the bilateral primitives $u_j(x_j,y_j,\theta)$ and $\widetilde
v_j(x_j,y_j,x_{-j},\theta)$ satisfy the single-principal monotonicity,
opposite-preference, single-crossing, local-regularity, and atomicity
conditions in Section \ref{sec:appI}. Following any unilateral deviation by
principal $j$, the no-safe-deviation test admits a continuation equilibrium
in which the payoff contribution from the non-deviating relationship is
independent of principal $j$'s message.
\end{assum}

This is a sufficient condition, not a claim that common-agency environments
are generically separable. Its role is to isolate a class of applications in
which each bilateral relationship has its own exit decision and a unilateral
deviation cannot improve by manipulating continuation payoffs in the other
relationship. The separability and principal-separable-continuation clauses
then make the deviation comparable to the single-principal problem: holding
the other principal's contractible offer fixed, the deviating principal
faces a bilateral menu-with-recommendations problem of the kind analyzed in
Section \ref{sec:appI}.

\paragraph{Main result.}

Consider the \emph{single contractible-action offer game}, in which each
principal offers only one contractible action. The next proposition shows
that equilibria of this simple game are robust to arbitrary mechanism
deviations in the no-safe-deviation sense of Definition \ref%
{def:equilibrium:robust}.

\begin{prop}
\label{prop_pooling_common} Under Assumption \ref%
{Ass_multi_bilateral_reduction}, every equilibrium allocation in the single
contractible-action offer game is a robust-equilibrium allocation.
\end{prop}

The logic is direct. A complex deviation defeats the simple-offer
equilibrium only if every continuation equilibrium after the deviation gives
the deviator a higher payoff. Under principal-separable continuation, one
post-deviation continuation equilibrium holds fixed the payoff contribution
from the non-deviating relationship. The deviating principal's communication
problem then reduces to the single-principal problem in Section \ref%
{sec:appI}, holding the other principal's contract fixed. Hence every
complex deviation has a continuation equilibrium that is payoff-equivalent
to a simple contractible-action deviation, and no such deviation is
profitable in an equilibrium of the simple-offer game. The formal proof is
in the Supplemental Appendix.

\paragraph{Application to the common-agency labor example.}

\label{example_common_agency}

We now apply Proposition \ref{prop_pooling_common} to the two-firm labor
example in Section \ref{sec:Example3}. Suppose that the worker can work for
two firms. Let $x_{j}\geq 0$ and $y_{j}\geq 0$ denote the contractible
support and non-contractible work intensity chosen by firm $j$, for $j=1,2$.
The firms' profit functions are
\begin{equation*}
v_{j}=(1+\beta x_{-j})y_{j}\theta -x_{j}^{2},\qquad j=1,2,
\end{equation*}%
where $\beta =17/21$. Thus, each firm's payoff depends on the other firm's
contractible action. The worker's utility function is
\begin{equation*}
u=u_{1}+u_{2}=\frac{x_{1}\theta -y_{1}^{2}}{\sqrt{\theta }}+\frac{%
x_{2}\theta -y_{2}^{2}}{\sqrt{\theta }}.
\end{equation*}%
As in the single-principal example, assume that $\theta \sim \mathrm{Unif}%
[3,4]$.

For each firm $j$, the worker remains if and only if $\theta\geq
t_j:=y_j^2/x_j$. Holding $x_{-j}$ fixed and writing $A_j=1+\beta x_{-j}$,
the same cutoff calculation as in the single-principal example gives the
best response
\begin{equation*}
BR_{j}^{x}(x_{-j})=\left[ \left( 1+\beta x_{-j}\right) \frac{7\sqrt{3}}{8}%
\right] ^{2/3}.
\end{equation*}%
The derivation is collected in the Supplemental Appendix.

With $\beta =17/21$, this best response has the closed-form fixed point
\begin{equation*}
BR_{j}^{x}(3)=\left[ \left( 1+\frac{17}{21}\cdot 3\right) \frac{7\sqrt{3}}{8}%
\right] ^{2/3}=3.
\end{equation*}%
Therefore the simple-offer game has a symmetric equilibrium
\begin{equation*}
x_{1}^{\ast }=x_{2}^{\ast }=3,\qquad y_{1}^{\ast }=y_{2}^{\ast }=3,\qquad
t_{1}^{\ast }=t_{2}^{\ast }=3.
\end{equation*}%
All worker types remain with both firms, and the lowest type is just
indifferent. By Proposition \ref{prop_pooling_common}, this simple-offer
equilibrium is robust to deviations to richer mechanisms.

\section{Conclusion}

\label{sec:conclusion}

This paper studies when full commitment is a valid benchmark for contracts
that leave some decisions to later discretion. The answer depends on the
structure of the discretionary decision. When the principal's later choice
is a bounded revision of the same final action, limited commitment changes
implementation but not final-action allocations. When the contractible and
discretionary choices are distinct instruments, the equivalence can fail.

The technical device behind these results is the \emph{extended taxation
principle}, which identifies a minimal canonical contract space under
general preferences. Menus specify contractible actions, while
recommendations carry the continuation information needed to discipline
discretionary choices.

The analysis applies this test in two directions. In revisable-action
environments, continuation equilibria implement the same final-action
distributions as full commitment, whether revisions take the form of bounded
additive adjustments or bounded percentage adjustments. When discretionary
choices are distinct instruments and the agent can exit, the equivalence
need not hold. Even there, in structured applications optimal
single-principal contracts are simple, and, under separability, this
simplicity is robust in common agency even against arbitrary mechanisms.

The main text focuses on publicly observable contracting. This distinction
is immaterial with a single principal, but under competition public messages
help determine the beliefs and continuation actions that discipline
deviations. The final section after the Supplemental Appendix, ``Private
Contracting with Imperfect Commitment,'' treats the corresponding
private-contracting environment, in which each principal observes only his
own mechanism and the message sent to him. It shows that the
single-principal conclusions are unchanged and identifies the private
canonical contract space needed for robust equilibrium with multiple
principals.

Together, these results show that limited commitment need not make contract
design intractable. Once the correct canonical contract space is identified,
familiar full-commitment benchmarks can be evaluated as equilibrium
implications rather than imposed as assumptions. Discretion is not harmless,
but its consequences can be analyzed within a disciplined contract space,
clarifying when full commitment is justified and when it is not.\newpage

\appendix

\renewcommand{\theHsection}{appendix.\Alph{section}} \renewcommand{%
\theHsubsection}{appendix.\Alph{section}.\arabic{subsection}}

\begin{center}
{\Large \textbf{Appendix}}
\end{center}

\noindent The appendix collects proof material that is used directly in the
main text. Appendix A proves the revisable-action equivalence in Proposition %
\ref{prop:revisable_equivalence_main}. Appendix B records the formal
continuation-equilibrium definition underlying Definition \ref%
{def:weak-equilibrium}. Appendix C contains the canonical-contract
definitions, the continuation-equilibrium reduction lemma, and the main
canonicality arguments; remaining auxiliary robust-equilibrium facts are
collected in the Supplemental Appendix. Appendix D contains the proof of the
single-contractible-action result in Proposition \ref{proposition_single}.
The proof of Proposition \ref{proposition_single_optimal}, the proof of
Proposition \ref{prop_pooling_common}, and the less central auxiliary proofs
are collected in the Supplemental Appendix.

\section{Proof of Proposition \protect\ref{prop:revisable_equivalence_main}
\label{sec:proof_appIII}}

\paragraph{Admissibility convention.}

The equivalence result is stated for admissible continuation equilibria
whose transformed mechanisms in the proof remain admissible: replacing a
limited-commitment mechanism by its final-action baseline, and replacing a
full-commitment mechanism by the shifted baseline that supports the same
final action under revision. This is only a regularity convention, automatic
in finite or countable menus and in standard compact-continuous
specifications.

\paragraph{Proof}

Fix $\alpha \geq 0$. For a posterior belief $p\in \Delta (\Theta )$, write
\begin{equation*}
V(z\mid p)=\int_{\Theta }v(z,\theta )\,p(d\theta ).
\end{equation*}%
Since $v(\cdot ,\theta )$ is strictly concave for every $\theta $, the
posterior expected payoff $V(\cdot \mid p)$ is strictly concave. By
assumption it admits a maximizer $r(p)$. Hence $r(p)$ is unique, and $%
V(\cdot \mid p)$ is strictly increasing on $(-\infty ,r(p)]$ and strictly
decreasing on $[r(p),\infty )$.

Given a continuation equilibrium $[G,(\gamma ,q,p)]$, define $\lambda
_{\theta }\in \Delta (\mathbb{R})$ by
\begin{equation*}
\lambda _{\theta }(B)=\int_{M}\mathbf{1}_{\{G(m)+\gamma (m)\in
B\}}\,q(\theta )[dm]
\end{equation*}%
for every Borel set $B\subseteq \mathbb{R}$. Deterministic allocations $%
z:\Theta \rightarrow \mathbb{R}$ are the special case in which $\lambda
_{\theta }$ is degenerate at $z(\theta )$ for every type $\theta $. For each
$\alpha $, let $\Gamma (\alpha )$ denote the set of continuation equilibrium
stochastic final-action allocations when the revision bound is $\alpha $.

We first show that $\Gamma (\alpha )\subseteq \Gamma (0)$. Consider any
continuation equilibrium in the model with revision bound $\alpha $, denoted
by $[G,(\gamma ,q,p)]$, where $G:M\rightarrow \mathbb{R}$ is the
contractible baseline and $\gamma :M\rightarrow \lbrack -\alpha ,\alpha ]$
is the receiver's continuation strategy. For each message $m\in M$, let
\begin{equation*}
z_{m}=G(m)+\gamma (m)
\end{equation*}%
be the final action induced by $m$. Construct a full-commitment mechanism $%
\widetilde{G}:M\rightarrow \mathbb{R}$ by $\widetilde{G}(m)=z_{m}$ for every
$m\in M$, and set the continuation action equal to zero after every message.
Since $G$ and $\gamma $ are measurable, $\widetilde{G}$ is measurable. By
the admissibility convention in Section \ref{sec:revisable}, this
transformed map is an admissible full-commitment baseline. Keep the same
sender strategy $q$ and the same belief system $p$.

Because both players' payoffs depend only on the final action, every mixed
message strategy yields the same payoff under the two mechanisms: for every
type $\theta $ and every $\sigma \in \Delta (M)$,
\begin{equation*}
\int_{M}u(\widetilde{G}(m),0,\theta )\,\sigma (dm)=\int_{M}u(G(m),\gamma
(m),\theta )\,\sigma (dm).
\end{equation*}%
Taking $\sigma =q(\theta )$ preserves the sender's equilibrium payoff, and
taking arbitrary $\sigma $ preserves the entire set of deviation payoffs.
Hence the sender's incentive and participation constraints are preserved.
Receiver optimality is automatic when $\alpha =0$, since the receiver has no
discretionary revision, and beliefs remain Bayesian consistent because the
sender strategy is unchanged. Finally, for every type $\theta $ and every
Borel set $B\subseteq \mathbb{R}$,
\begin{equation*}
\int_{M}\mathbf{1}_{\{\widetilde{G}(m)\in B\}}\,q(\theta )[dm]=\int_{M}%
\mathbf{1}_{\{G(m)+\gamma (m)\in B\}}\,q(\theta )[dm].
\end{equation*}%
Thus the same stochastic final-action allocation is induced under full
commitment.

We next show that $\Gamma (0)\subseteq \Gamma (\alpha )$. Consider any
continuation equilibrium in the full-commitment model, denoted by $%
[G,(0,q,p)]$, where $G:M\rightarrow \mathbb{R}$ and the receiver's
continuation action is identically zero. For each message $m\in M$, let $%
z_{m}=G(m)$ be the final action induced by $m$, and let $r_{m}=r(p_{m})$ be
the receiver's posterior ideal final action after message $m$, where $%
p_{m}:=p(m)$. Since $m\mapsto p_{m}$ and $p\mapsto r(p)$ are measurable, $%
m\mapsto r_{m}$ is measurable. Since $m\mapsto z_{m}$ and $m\mapsto r_{m}$
are measurable, the sets $\{m:z_{m}<r_{m}\}$, $\{m:z_{m}=r_{m}\}$, and $%
\{m:z_{m}>r_{m}\}$ are measurable. Define a new contractible baseline $%
\widehat{G}:M\rightarrow \mathbb{R}$ and a continuation strategy $\widehat{%
\gamma }:M\rightarrow \lbrack -\alpha ,\alpha ]$ by
\begin{equation*}
\widehat{G}(m)=%
\begin{cases}
z_{m}-\alpha , & \text{if }z_{m}<r_{m}, \\
z_{m}, & \text{if }z_{m}=r_{m}, \\
z_{m}+\alpha , & \text{if }z_{m}>r_{m},%
\end{cases}%
\qquad \widehat{\gamma }(m)=%
\begin{cases}
\alpha , & \text{if }z_{m}<r_{m}, \\
0, & \text{if }z_{m}=r_{m}, \\
-\alpha , & \text{if }z_{m}>r_{m}.%
\end{cases}%
\end{equation*}%
Then
\begin{equation*}
\widehat{G}(m)+\widehat{\gamma }(m)=z_{m},\qquad \forall m\in M.
\end{equation*}%
The preceding measurability observation implies that $\widehat{G}$ and $%
\widehat{\gamma }$ are measurable. By the admissibility convention in
Section \ref{sec:revisable}, the transformed baseline $\widehat{G}$ is an
admissible contract.

We verify receiver optimality. After message $m$, the receiver can choose
any $y\in \lbrack -\alpha ,\alpha ]$, and hence can choose any final action
in the interval $[\widehat{G}(m)-\alpha ,\widehat{G}(m)+\alpha ].$

If $z_{m}<r_{m}$, then $\widehat{G}(m)=z_{m}-\alpha $, so this interval is $%
[z_{m}-2\alpha ,z_{m}]$. Since the whole interval lies weakly below $r_{m}$
and $V(\cdot \mid p_{m})$ is strictly increasing on $(-\infty ,r_{m}]$, the
best feasible final action is $z_{m}$, implemented by $y=\alpha $. If $%
z_{m}>r_{m}$, then the feasible interval is $[z_{m},z_{m}+2\alpha ]$. Since
this interval lies weakly above $r_{m}$ and $V(\cdot \mid p_{m})$ is
strictly decreasing on $[r_{m},\infty )$, the best feasible final action is $%
z_{m}$, implemented by $y=-\alpha $. If $z_{m}=r_{m}$, then the receiver's
posterior ideal is feasible and is implemented by $y=0$. Thus $\widehat{%
\gamma }$ is optimal after every message.

Keep the same sender strategy $q$ and the same belief system $p$. Since each
message induces the same final action as under full commitment, for every
type $\theta$ and every mixed deviation $\sigma\in\Delta(M)$,
\begin{equation*}
\int_M u(\widehat{G}(m),\widehat{\gamma}(m),\theta)\,\sigma(dm) =\int_M
u(G(m),0,\theta)\,\sigma(dm).
\end{equation*}
Thus the sender's incentive and participation constraints are preserved.
Beliefs remain Bayesian consistent because the sender strategy is unchanged.
Finally, for every type $\theta $ and every Borel set $B\subseteq \mathbb{R}$%
,
\begin{equation*}
\int_{M}\mathbf{1}\{\widehat{G}(m)+\widehat{\gamma }(m)\in B\}\,q(\theta
)[dm]=\int_{M}\mathbf{1}\{G(m)\in B\}\,q(\theta )[dm].
\end{equation*}%
Thus the same stochastic final-action allocation is induced in the model
with revision bound $\alpha $.

Combining the two inclusions gives $\Gamma (\alpha )=\Gamma (0)$ for every $%
\alpha \geq 0$. The proof establishes the equality for unrestricted
mechanisms. By Proposition \ref{prop:canonical:weak} in Appendix \ref%
{sec:prelim_analysis:2}, restricting attention to the canonical
menu-with-recommendations contract space $\mathcal{G}^{\ast }$ preserves the
corresponding action-pair allocations, and therefore also their stochastic
final-action allocations.

\section{Formal Continuation Equilibrium\label%
{app:formal_continuation_equilibrium}}

\paragraph{Formal continuation-equilibrium definition.}

Fix a contract profile $G=(G_{j}:M_{j}\rightarrow X_{j})_{j\in \mathcal{J}%
}\in \mathcal{G}$ with $M:=\times _{j\in \mathcal{J}}M_{j}$. Define $%
\mathcal{Q}^{M}:=[\Delta (M)]^{\Theta }$ and
\begin{equation*}
\mathcal{Y}_{j}^{M}:=\{\gamma _{j}:M\rightarrow Y_{j}:\gamma _{j}(m)\in
F_{j}(G_{j}(m_{j}))\text{ for every }m\in M\},\qquad \mathcal{Y}^{M}:=\times
_{j\in \mathcal{J}}\mathcal{Y}_{j}^{M}.
\end{equation*}%
Let $\mathcal{P}_{j}^{M}:=[\Delta (\Theta )]^{M}$ and $\mathcal{P}%
^{M}:=\times _{j\in \mathcal{J}}\mathcal{P}_{j}^{M}$. A belief profile $%
p=(p_{j})_{j\in \mathcal{J}}$ is Bayesian consistent with $q:\Theta
\rightarrow \Delta (M)$ if, for every $j$, and all measurable $%
E_{1}\subseteq \Theta $ and $E_{2}\subseteq M$,
\begin{equation*}
\int_{E_{1}}q(\theta )[E_{2}]\,\mu \lbrack d\theta ]=\int_{\Theta }\left(
\int_{E_{2}}p_{j}(m)[E_{1}]\,q(\theta )[dm]\right) \mu \lbrack d\theta ].
\end{equation*}%
A profile $[G,(\gamma ,q,p)]$ is a $\mathcal{G}$-continuation equilibrium if
$G\in \mathcal{G}$, $(\gamma ,q,p)\in \mathcal{Y}^{M}\times \mathcal{Q}%
^{M}\times \mathcal{P}^{M}$, $p$ is Bayesian consistent with $q$, and the
following two conditions hold. First, for every $\theta \in \Theta $,
\begin{equation*}
\int_{M}u\left( (G_{j}(m_{j}))_{j\in \mathcal{J}},(\gamma _{j}(m))_{j\in
\mathcal{J}},\theta \right) q(\theta )[dm]\geq \max \left\{ U(\theta
),\sup_{m\in M}u\left( (G_{j}(m_{j}))_{j\in \mathcal{J}},(\gamma
_{j}(m))_{j\in \mathcal{J}},\theta \right) \right\} .
\end{equation*}%
Second, for every $j\in \mathcal{J}$, every $m=(m_{k})_{k\in \mathcal{J}}\in
M$, and every $y_{j}\in F_{j}(G_{j}(m_{j}))$,
\begin{equation*}
\int_{\Theta }v_{j}\left( (G_{k}(m_{k}))_{k\in \mathcal{J}},(\gamma
_{k}(m))_{k\in \mathcal{J}},\theta \right) p_{j}(m)[d\theta ]\geq
\int_{\Theta }v_{j}\left( (G_{k}(m_{k}))_{k\in \mathcal{J}},(y_{j},(\gamma
_{k}(m))_{k\neq j}),\theta \right) p_{j}(m)[d\theta ].
\end{equation*}

\section{Canonical Contracts}

This appendix proves the minimal-canonical-contract results in Theorems \ref%
{thm:minimal:weak} and \ref{thm:minimal:weak-robust}. We first introduce an
auxiliary contract space. For any $M_k\subseteq X_k\times Y_k$, write
\begin{equation*}
L(M_k):=\{x_k\in X_k:\exists y_k\in Y_k\text{ such that }(x_k,y_k)\in M_k\}.
\end{equation*}
Let $\pi_{X_k}:X_k\times Y_k\rightarrow X_k$ denote the coordinate
projection, $\pi_{X_k}(x_k,y_k)=x_k$.
\begin{align*}
\mathcal{M}_{k}^{\sharp }:=\{M_{k}\in \mathcal{B}(X_{k}\times
Y_{k})\diagdown \{\varnothing\}:&\ M_k\subseteq \{(x_k,y_k):y_k\in
F_k(x_k)\}, \\
&\ L(M_k)\in\mathcal{L}_k^\ast, \\
&\ \pi_{X_k}|_{M_k}:M_k\rightarrow L(M_k)\text{ admits} \\
&\ \text{a measurable right inverse}\},\quad \forall k\in\mathcal{J}, \\
\mathcal{G}_{k}^{\sharp }:=\{\langle G_{k}:M_{k}\longrightarrow
X_{k}\rangle: &\ M_{k}\in \mathcal{M}_{k}^{\sharp }\text{ and} \\
&\ G_{k}(x_{k},y_{k})=x_{k},\ \forall (x_{k},y_{k})\in M_{k}\},\quad \forall
k\in \mathcal{J}, \\
\mathcal{M}^{\sharp }&:=\times _{k\in \mathcal{J}}\mathcal{M}_{k}^{\sharp },
\\
\mathcal{G}^{\sharp }&:=\times _{k\in \mathcal{J}}\mathcal{G}_{k}^{\sharp }.
\end{align*}%
A contract in $\mathcal{G}_{j}^{\sharp }$ offers a non-empty set of feasible
pairs $(x_j,y_j)$; the agent chooses $x_j$ and recommends $y_j$. Unlike $%
\mathcal{G}_{j}^{\ast }$, this class need not include every feasible
recommendation after a given $x_j$. The Borel-image and selector
requirements in the definition ensure that every contract in $\mathcal{G}%
_{j}^{\sharp }$ is admissible under the convention in Section \ref{sec:model}%
. Thus%
\begin{equation*}
\mathcal{M}_{j}^{\ast }\subseteq \mathcal{M}_{j}^{\sharp }\text{, and }%
\mathcal{G}_{j}^{\ast }\subseteq \mathcal{G}_{j}^{\sharp }\text{, }\forall
j\in \mathcal{J}\text{,}
\end{equation*}%
Every contract in $\mathcal{G}_{j}^{\ast }$ is therefore in $\mathcal{G}%
_{j}^{\sharp }$. The inclusions are strict whenever $F_j(x_j)$ contains at
least two feasible recommendations for some $x_j$; if every $F_j(x_j)$ is a
singleton, the two classes coincide. In Example \ref{example1}, $\left\vert
\mathcal{G}_{j}^{\ast }\right\vert =3$ and $\left\vert \mathcal{G}%
_{j}^{\sharp }\right\vert =31$.

We introduce a useful binary relation on contracts. For each $j\in \mathcal{J%
}$ and any contract $G_{j}:M_{j}\longrightarrow X_{j}$, define
\begin{equation*}
G_{j}(M_{j}):=\left\{ x_{j}\in X_{j}:\exists m_{j}\in M_{j}\text{, }%
G_{j}\left( m_{j}\right) =x_{j}\right\} \text{.}
\end{equation*}%
i.e., $G_{j}(M_{j})$ is the image set of $G_{j}:M_{j}\longrightarrow X_{j}$.
For any two contracts $G_{j}:M_{j}\longrightarrow X_{j}$ and $G_{j}^{\prime
}:M_{j}^{\prime }\longrightarrow X_{j}$, define%
\begin{equation*}
G_{j}\sim G_{j}^{\prime }\Longleftrightarrow G_{j}(M_{j})=G_{j}^{\prime
}(M_{j}^{\prime })\text{.}
\end{equation*}%
For any contract $G_{j}:M_{j}\longrightarrow X_{j}$, let $G_{j}^{G_{j}}$
denote the unique contract in $\mathcal{G}_{j}^{\ast }$ such that $G_{j}\sim
G_{j}^{G_{j}}$, i.e.,
\begin{eqnarray*}
M_{j}^{G_{j}}:= &&\left\{ \left( x_{j},y_{j}\right) \in X_{j}\times
Y_{j}:x_{j}\in G_{j}(M_{j})\text{ and }y_{j}\in F_{j}(x_{j})\right\} \text{,}
\\
G_{j}^{G_{j}} &:&M_{j}^{G_{j}}\longrightarrow X_{j}\text{ such that }%
G_{j}^{G_{j}}\left( x_{j},y_{j}\right) =x_{j}\text{, }\forall \left(
x_{j},y_{j}\right) \in M_{j}^{G_{j}}\text{.}
\end{eqnarray*}

\subsection{Auxiliary canonical-contract facts}

\label{sec:prelim_analysis:2}

The proof uses auxiliary facts about the relation between arbitrary
contracts, menus with recommendations, and robust-equilibrium deviations.
The first fact shows that every continuation-equilibrium allocation can be
induced by a menu-with-recommendations contract.

\begin{lemma}
\label{lem:extension_w2}For any contract space $\mathcal{G}$, any allocation
$z$ induced by a $\mathcal{G}$-continuation equilibrium is also induced by a
$\mathcal{G}^{\ast }$-continuation equilibrium.
\end{lemma}

\begin{proof}
Suppose that $z$ is induced by a $\mathcal{G}$-continuation equilibrium $%
\left[ G\text{, }\left( \gamma ,q,p\right) \right] $, where $%
G_{k}:M_{k}\rightarrow X_{k}$ for each $k\in \mathcal{J}$. For each
principal $k$, let
\begin{equation*}
L_k:=G_k(M_k)
\end{equation*}
and replace $G_k$ by the menu-with-recommendations contract $%
G_k^{G_k}:M_k^{G_k}\rightarrow X_k$, where
\begin{equation*}
M_k^{G_k}:=\{(x_k,y_k)\in X_k\times Y_k:x_k\in L_k,\ y_k\in F_k(x_k)\},
\qquad G_k^{G_k}(x_k,y_k)=x_k .
\end{equation*}
By the measurable-selection regularity imposed in the main text, choose a
measurable selector $s_k:L_k\rightarrow M_k$ such that $G_k(s_k(x_k))=x_k$
for every $x_k\in L_k$.

Let $M^{\prime}=\times_{k\in\mathcal{J}}M_k^{G_k}$ and define the map
\begin{equation*}
T:M\rightarrow M^{\prime},\qquad T(m)=\left(
G_k(m_k),\gamma_k(m)\right)_{k\in\mathcal{J}} .
\end{equation*}
The new communication strategy is the push-forward of $q$ by $T$:
\begin{equation*}
q^{\prime}(\theta)[E]=q(\theta)[T^{-1}(E)]
\end{equation*}
for every measurable $E\subset M^{\prime}$. Hence $q^{\prime}$ induces
exactly the same allocation as the original equilibrium.

It remains to define beliefs and continuation actions at the larger message
space. For every $m^{\prime}=(x_k,y_k)_{k\in\mathcal{J}}\in M^{\prime}$, let
\begin{equation*}
r(m^{\prime}):=\left( s_k(x_k)\right)_{k\in\mathcal{J}}\in M
\end{equation*}
be an original message profile that induces the same contractible actions.
If $m^{\prime}\in T(M)$, set
\begin{equation*}
\gamma^{\prime}_k(m^{\prime})=y_k,\qquad k\in\mathcal{J},
\end{equation*}
and choose $p^{\prime}_k(m^{\prime})$ to be a version of the conditional
distribution of $\theta$ given $T(m)=m^{\prime}$. If $m^{\prime}\notin T(M)$%
, ignore the off-path recommendations and set
\begin{equation*}
\gamma^{\prime}_k(m^{\prime})=\gamma_k(r(m^{\prime})),\qquad
p^{\prime}_k(m^{\prime})=p_k(r(m^{\prime})),\qquad k\in\mathcal{J}.
\end{equation*}
These off-path actions are feasible because $G_k(r_k(m^{%
\prime}))=x_k=G_k^{G_k}(m^{\prime}_k)$.

Bayesian consistency holds by construction of the regular conditional
beliefs on the image of $T$; beliefs outside that image are unrestricted by
Bayes' rule. Principal optimality also holds. On $T(M)$, each $%
\gamma^{\prime}_k(m^{\prime})=y_k$ is the continuation action chosen after
all original messages mapped into $m^{\prime}$, and optimality is preserved
under mixtures of the corresponding original posteriors. Outside $T(M)$, the
pair $(\gamma^{\prime}_k(m^{\prime}),p^{\prime}_k(m^{\prime}))$ is copied
from the original message profile $r(m^{\prime})$.

Finally, the enlarged message space creates no profitable deviation for the
agent. A message in $T(M)$ gives the same action profile as some original
message profile. A message outside $T(M)$ gives the same contractible-action
profile and the same continuation-action profile as the original message
profile $r(m^{\prime})$. Thus every payoff the agent can obtain under the
new menu was already attainable by a message deviation in the original
equilibrium. The agent's incentive and participation constraints therefore
remain satisfied. Hence
\begin{equation*}
\left[ (G_k^{G_k})_{k\in\mathcal{J}},(\gamma^{\prime},q^{\prime},p^{\prime})%
\right]
\end{equation*}
is a $\mathcal{G}^{\ast}$-continuation equilibrium inducing $z$.
\end{proof}

\begin{prop}
\label{prop:canonical:weak}$\mathcal{G}^{\ast }$ is a canonical contract
space for continuation equilibrium.
\end{prop}

\begin{proof}
This is a direct consequence of Lemma \ref{lem:extension_w2} given the
definition of a canonical contract space for continuation equilibrium.
\end{proof}

The remaining auxiliary facts, Lemmas \ref{lem:extension_w}, \ref%
{lem:max:robust:off-equilibrium-pah}, and \ref%
{lem:max:robust:on-equilibrium-pah}, together with Propositions \ref%
{prop:canonical:lower-bound} and \ref{prop:canonical:upper-bound}, provide
the robust-equilibrium lower and upper bounds used below. Their formal
statements and proofs are collected in the Supplemental Appendix.

For reference, the robust-equilibrium arguments use the following relabeling
relation. For any two contracts $G_{j}:M_{j}\longrightarrow X_{j}$ and $%
G_{j}^{\prime }:M_{j}^{\prime }\longrightarrow X_{j}$, write $G_{j}^{\prime
}\sqsupset _{j}G_{j}$ if there exists a measurable surjection $\iota
_{j}:M_{j}^{\prime }\rightarrow M_{j}$, admitting a measurable right
inverse, such that $G_{j}^{\prime }(m_{j}^{\prime })=G_{j}(\iota
_{j}(m_{j}^{\prime }))$ for every $m_{j}^{\prime }\in M_{j}^{\prime }$.

\subsection{A lower bound for canonical contract space}

\label{sec:lowerbound_canonical}

First, we show that $\mathcal{G}^{\ast }$ is a lower bound for the canonical
contract space for the two equilibrium notions.

\begin{prop}
\label{thm:lower:weak}Suppose $\left\vert \Theta \right\vert \geq \left\vert
X\right\vert $. If $\mathcal{G}$ is a canonical contract space for
continuation equilibrium, we have $\left\vert \mathcal{G}\right\vert \geq
\left\vert \mathcal{G}^{\ast }\right\vert $.
\end{prop}

\begin{prop}
\label{thm:lower:weak-robust}Suppose $\left\vert \Theta \right\vert \geq
\left\vert X\right\vert $. If $\mathcal{G}$ is a canonical contract space
for robust equilibrium, we have $\left\vert \mathcal{G}\right\vert \geq
\left\vert \mathcal{G}^{\ast }\right\vert $.
\end{prop}

\begin{proof}
We provide a unified proof for both propositions. Suppose $\left\vert \Theta
\right\vert \geq \left\vert X\right\vert $. Recall that $\left\vert \mathcal{%
G}_{j}^{\ast }\right\vert =\left\vert \mathcal{L}_{j}^{\ast }\right\vert $
for any $j\in \mathcal{J}$. Consider any canonical contract space $\mathcal{G%
}$ for either continuation equilibrium or robust equilibrium. Fix any $j\in
\mathcal{J}$. We now prove $\left\vert \mathcal{G}_{j}\right\vert \geq
\left\vert \mathcal{L}_{j}^{\ast }\right\vert =\left\vert \mathcal{G}%
_{j}^{\ast }\right\vert $, which completes the proof.

Consider the image map $\phi _{j}:\mathcal{G}_{j}\longrightarrow
2^{X_{j}}\diagdown \left\{ \varnothing \right\} $ such that%
\begin{equation*}
\phi _{j}\left( G_{j}:M_{j}\longrightarrow X_{j}\right) =G_{j}\left(
M_{j}\right) \text{, }\forall G_{j}\in \mathcal{G}_{j}\text{.}
\end{equation*}%
We now prove that $\mathcal{L}_{j}^{\ast }$ is contained in the range of $%
\phi_j$, which implies $\left\vert \mathcal{G}_{j}\right\vert \geq
\left\vert \mathcal{L}_{j}^{\ast }\right\vert =\left\vert \mathcal{G}%
_{j}^{\ast }\right\vert $. Specifically, fix any $\widetilde{X}_{j}\in
\mathcal{L}_{j}^{\ast }$, and we prove $\phi _{j}\left( G_{j}\right) =%
\widetilde{X}_{j}$ for some $G_{j}\in \mathcal{G}_{j}$.

Since $\widetilde X_j$ is a Borel subset of the standard Borel space $X_j$
and $\left\vert \Theta \right\vert \geq \left\vert X\right\vert \geq
\left\vert \widetilde{X}_{j}\right\vert $, the standard-Borel isomorphism
theorem gives a measurable surjection $\varphi _{j}:\Theta \longrightarrow
\widetilde{X}_{j}$. Set $U(\theta)=0$ for every $\theta$. Consider the
following environment: for any $\left( x,y,\theta ,k\right) \in X\times
Y\times \Theta \times \mathcal{J}$,
\begin{gather}
x_{j}\notin \widetilde{X}_{j}\Longrightarrow u\left( x,y,\theta \right)
=-v_{k}\left( x,y,\theta \right) =8\text{,}  \label{kkt1} \\
\notag \\
\text{and }x_{j}\in \widetilde{X}_{j}\Longrightarrow u\left( x,y,\theta
\right) =v_{k}\left( x,y,\theta \right) =\left\{
\begin{tabular}{ll}
$1\text{,}$ & if $x_{j}=\varphi _{j}\left( \theta \right) $, \\
&  \\
$0\text{,}$ & otherwise.%
\end{tabular}%
\right. \text{.}  \label{kkt2}
\end{gather}%
Thus the agent strictly prefers actions in $X_{j}\diagdown \widetilde{X}_{j}$
to actions in $\widetilde{X}_{j}$, while principals have the opposite
preference. All payoffs depend only on principal $j$'s contractible action.

For any $k\in \mathcal{J}\diagdown \left\{ j\right\} $, fix any $\overline{x}%
_{k}\in X_{k}$. Define%
\begin{eqnarray*}
M_{j} &:=&\widetilde{X}_{j}\text{ and }M_{k}:=\left\{ \overline{x}%
_{k}\right\} \text{, }\forall k\in \mathcal{J}\diagdown \left\{ j\right\}
\text{,} \\
\widetilde{G}_{j} &:&M_{j}\rightarrow X_{j}\text{ such that }\widetilde{G}%
_{j}\left( x_{j}\right) =x_{j}\text{, }\forall x_{j}\in M_{j}\text{,} \\
\widetilde{G}_{k} &:&M_{k}\rightarrow X_{k}\text{ such that }\widetilde{G}%
_{k}\left( x_{k}\right) =x_{k}\text{, }\forall x_{k}\in M_{k}\text{, }%
\forall k\in \mathcal{J}\diagdown \left\{ j\right\} \text{,}
\end{eqnarray*}%
and%
\begin{equation*}
\widetilde{\mathcal{G}}_{k}:=\left\{ \widetilde{G}_{k}\right\} \cup \mathcal{%
G}_{k},\ \forall k\in\mathcal{J},\qquad \widetilde{\mathcal{G}}:=\times
_{k\in \mathcal{J}}\widetilde{\mathcal{G}}_{k}\text{.}
\end{equation*}%
There is a $\widetilde{\mathcal{G}}$-continuation equilibrium in which
principal $j$ offers $\widetilde{G}_{j}$, type $\theta$ chooses $%
x_j=\varphi_j(\theta)$, and the other principals choose their fixed
contractible actions. This allocation is also supported by an unrestricted
robust equilibrium. Indeed, the equilibrium payoff of every principal is
one. If principal $j$ deviates to a contract whose image contains some $%
x_j\notin \widetilde X_j$, then there is a continuation equilibrium after
the deviation in which the agent chooses such an action, yielding payoff $-8$
to every principal. If the deviating contract of principal $j$ has image
contained in $\widetilde X_j$, then every continuation equilibrium following
the deviation gives principal $j$ payoff at most one. Deviations by any
principal $k\neq j$ do not affect payoffs, because all payoffs depend only
on $x_j$. Hence the allocation generated by $\widetilde G$ is both a $%
\widetilde{\mathcal{G}}$-continuation-equilibrium allocation and an
unrestricted robust-equilibrium allocation. Since $\varphi _{j}:\Theta
\longrightarrow \widetilde{X}_{j}$ is surjective, we have $\phi _{j}\left(
\widetilde{G}_{j}\right) =\widetilde{X}_{j}$.

If $\mathcal{G}$ is canonical for continuation equilibrium, canonicality
applied to $\widetilde{\mathcal{G}}$ yields a $\mathcal{G}$-continuation
equilibrium inducing the same allocation. If $\mathcal{G}$ is canonical for
robust equilibrium, canonicality for robust equilibrium yields a $\mathcal{G}
$-robust equilibrium inducing the same allocation. In either case, denote
the resulting $\mathcal{G}$-equilibrium by $\left[ G=\times _{k\in \mathcal{J%
}}\left( G_{k}\right) \text{, }\left( \gamma ,q,p\right) \right] $. For each
$x_j\in\widetilde X_j$, there is a type $\theta$ such that $%
\varphi_j(\theta)=x_j$, and the induced allocation for that type chooses $%
x_j $. Therefore $\widetilde{X}_{j}\subset \phi _{j}\left( G_{j}\right) $.
Finally, we must have $\phi _{j}\left( G_{j}\right) \subset \widetilde{X}%
_{j} $. Otherwise, there exists $x_{j}^{\prime }\in \phi _{j}\left(
G_{j}\right) \diagdown \widetilde{X}_{j}$, and the agent finds it strictly
profitable to deviate to a message inducing $x_{j}^{\prime }$, contradicting
the agent incentive constraint in the selected $\mathcal{G}$-equilibrium.
Therefore, $\phi _{j}\left( G_{j}\right) =\widetilde{X}_{j}$.
\end{proof}

\subsection{Proof of Theorem \protect\ref{thm:minimal:weak}}

\label{sec:proof_thm_canon_continuation}

Propositions \ref{prop:canonical:weak} and \ref{thm:lower:weak} immediately
imply Theorem \ref{thm:minimal:weak}.

\subsection{Proof of Theorem \protect\ref{thm:minimal:weak-robust}}

\label{sec:minimal_contract_muti_principals}

Given $\left\vert \mathcal{J}\right\vert \geq 2$, Proposition \ref%
{thm:lower:weak-robust} gives the lower bound: any canonical contract space
for robust equilibrium must be at least as large as $\mathcal{G}^{\ast}$. It
remains to show that $\mathcal{G}^{\ast}$ is itself canonical for robust
equilibrium.

The proof has two steps. First, Proposition \ref{prop:canonical:upper-bound}
shows that any robust-equilibrium allocation can be represented by a robust
equilibrium in the menu-with-recommendations space $\mathcal{G}^{\ast}$.
Second, any $\mathcal{G}^{\ast}$-robust equilibrium can be embedded in an
enlarged message space that preserves the allocation while restoring
robustness to arbitrary deviations. The embedding uses the public message
profile to encode missing information when one principal deviates to a
coarser contract. The detailed construction and verification are in the
Supplemental Appendix. These two steps prove that $\mathcal{G}^{\ast}$ is a
canonical contract space for robust equilibrium, completing the proof. $%
\blacksquare $

\section{Proof of Proposition \protect\ref{proposition_single}}

\label{sec:proof_proposition_single}

We first record the regularity step used below.

\begin{lemma}[Rent extraction from common slack]
\label{lem:rent_extraction_common_slack} Fix a posterior belief $p$ and a
pair $(x,y)$. Suppose $p(S(x,y))>0$ and there is $\kappa >0$ such that
\begin{equation*}
u(x,y,\theta )\geq \kappa
\end{equation*}%
for $p$-almost every $\theta \in S(x,y)=\{\theta \in \Theta :u(x,y,\theta
)\geq 0\}$. Suppose also that $u$ is locally uniformly continuous in $y$ at $%
(x,y)$ on this participating posterior set: for every $\varepsilon >0$ there
exists $\delta >0$ such that, for every $t\in \lbrack 0,\delta ]$,
\begin{equation*}
|u(x,y+t,\theta )-u(x,y,\theta )|<\varepsilon
\end{equation*}%
for $p$-almost every $\theta \in S(x,y)$. If $u$ is decreasing in $y$ and $v$
is strictly increasing in $y$, and there is $\bar{\delta}>0$ such that $%
v(x,y+t,\cdot )$ is $p$-integrable on $S(x,y)$ for every $t\in \lbrack 0,%
\bar{\delta}]$, then for some $t>0$,
\begin{equation*}
\int_{\Theta }\mathbf{1}_{\{u(x,y+t,\theta )\geq 0\}}v(x,y+t,\theta
)\,p(d\theta )>\int_{\Theta }\mathbf{1}_{\{u(x,y,\theta )\geq
0\}}v(x,y,\theta )\,p(d\theta ).
\end{equation*}
\end{lemma}

\begin{proof}
Choose $\delta >0$ from local uniform continuity for $\varepsilon =\kappa $.
Then for every $t\in (0,\min \{\delta ,\bar{\delta}\}]$, $p$-almost all
types in $S(x,y)$ remain willing to participate after $(x,y+t)$. Types
outside $S(x,y)$ do not enter after the increase because $u$ is decreasing
in $y$. Hence the participating set is unchanged, up to a $p$-null set.
Since $p(S(x,y))>0$, $v$ is strictly increasing in $y$, and the relevant
functions are integrable,
\begin{equation*}
\int_{S(x,y)}v(x,y+t,\theta )\,p(d\theta )>\int_{S(x,y)}v(x,y,\theta
)\,p(d\theta )
\end{equation*}%
for every sufficiently small $t>0$. This is the desired strict improvement
in the reduced-form posterior payoff.
\end{proof}

\paragraph{Diffuse action-pair distributions.}

The atomic participating-action-pair restriction is an exact condition. It
rules out equilibria in which the agent's mixed communication strategy makes
participating mass spread over zero-probability action pairs. This
restriction is mainly technical. Under additional compactness, boundedness,
and uniform-continuity conditions, such diffuse action-pair schemes can be
approximated by finite partitions of $L\times \mathbb{R}_{+}$: each cell is
replaced by a representative action pair, and the induced payoffs, incentive
constraints, and continuation optimality conditions are preserved up to an
arbitrarily small error. Thus the exact proposition applies to the atomic
equilibrium class, while diffuse thin-randomization equilibria are covered
by the corresponding approximation or closure interpretation. The proof
below uses the exact condition to avoid carrying these approximation errors.

\bigskip

Now we prove Proposition \ref{proposition_single}. Let $w\in \{0,1\}$ denote
the agent's exit decision after observing the principal's non-contractible
action, with $w=1$ meaning that the agent stays. By Lemma \ref%
{lem:extension_w2}, it is enough to consider a menu with recommendations and
a truthful continuation equilibrium in which the principal follows the
agent's recommendation on the equilibrium path. Fix such a menu, let $L$ be
the set of available contractible actions, and let $q:\Theta \rightarrow
\Delta (L\times \mathbb{R}_{+})$ be the agent's continuation strategy over
chosen contractible actions and recommended non-contractible actions. For
each type $\theta $, let
\begin{equation*}
P_{\theta }:=\{(x,y)\in L\times \mathbb{R}_{+}:u(x,y,\theta )\geq 0\}.
\end{equation*}%
Let
\begin{equation*}
Q(A):=\int_{\Theta }q(\theta )[A]\,\mu (d\theta )
\end{equation*}%
be the aggregate distribution over selected action pairs, and let
\begin{equation*}
Q^{P}(A):=\int_{\Theta }q(\theta )[A\cap P_{\theta }]\,\mu (d\theta )
\end{equation*}%
be the aggregate measure over participating action pairs. For each action
pair $a=(x,y)$ with $Q(\{a\})>0$, let $p_{a}$ denote the posterior after
action pair $a$. For each action pair $a=(x,y)$, define
\begin{equation*}
\Theta (a):=\{\theta \in \Theta :q(\theta )[\{a\}]>0\text{ and the agent
stays after }y\}.
\end{equation*}%
All statements about type sets are understood up to $\mu $-null sets.

Consider two distinct selected action pairs $a=(x,y)$ and $b=(x^{\prime
},y^{\prime })$ with $Q^P(\{a\})>0$ and $Q^P(\{b\})>0$. For each such action
pair, $\mu(\Theta(a))>0$, because
\begin{equation*}
Q^P(\{a\})=\int_\Theta \mathbf{1}_{\{\theta\in \Theta(a)\}}q(\theta)[\{a\}]
\,\mu(d\theta).
\end{equation*}
In particular, each such action pair has positive aggregate probability and
hence a well-defined on-path posterior. Assumption \ref{assu_single_crossing}
implies that the set of types weakly preferring $a$ to $b$ is an interval,
with at most one indifference type. Hence the two regions cannot be
interlaced: if, after relabeling, $\theta _{L},\theta _{H}\in \Theta (a)$
and the region $\Theta (b)$ contains a positive-measure set of types between
them, then the at-most-one indifference-type condition lets us choose $%
\theta _{M}\in \Theta (b)$ with $\theta _{L}<\theta _{M}<\theta _{H}$ who
strictly prefers $b$ to $a$. Incentive compatibility says that $\theta _{L}$
and $\theta _{H}$ weakly prefer $a$ to $b$. Since the set of types weakly
preferring $a$ to $b$ is an interval, $\theta _{M}$ must also weakly prefer $%
a$, a contradiction. Thus any two distinct positive-measure action-pair
regions must be ordered.

Ordered positive-measure pooling regions cannot be sustained either.
Suppose, without loss of generality, that every type in $\Theta (b)$ is
higher than every type in $\Theta (a)$. For any $\theta^{\prime }\in \Theta
(b)$, incentive compatibility and monotonicity in type give
\begin{equation*}
u(x^{\prime },y^{\prime },\theta ^{\prime })\geq u(x,y,\theta ^{\prime
})\geq \sup_{\theta \in \Theta (a)}u(x,y,\theta )>0.
\end{equation*}%
Therefore, participation is slack on a positive posterior-measure subset of
the higher pooling region; in fact, the displayed inequality gives a common
positive slack bound for all types that participate after action pair $b$.
By Lemma \ref{lem:rent_extraction_common_slack}, the principal can raise $%
y^{\prime }$ slightly following action pair $b$ and strictly increase the
principal's posterior payoff. This contradicts continuation optimality.
Hence there cannot be two distinct positive-measure action-pair regions.

The atomic participating-action-pair restriction rules out the remaining
diffuse case. Since the preceding paragraphs rule out two distinct action
pairs with positive participating mass, this restriction implies that the
participating action-pair measure must be concentrated, up to null sets, on
a single action pair. If there is no participating action-pair mass, the
no-trade payoff is attainable by a single offer that induces exit. Otherwise
let $(\bar{x},\bar{y})$ be the action pair on which $Q^P$ is concentrated.
Any type with $u(\bar{x},\bar{y},\theta)>0$ must choose this participating
action pair rather than exit; otherwise that type has a profitable message
deviation to $(\bar{x},\bar{y})$. Thus, except for null sets and indifferent
cutoff types, the participating types coincide with $S(\bar{x},\bar{y})$.
The original menu payoff is no larger than the payoff from the single offer $%
\bar{x}$ followed by the best continuation choice of $y$ after that offer.
Since every single contractible-action offer is itself feasible, the menu
value within the atomic equilibrium class equals the single-offer value. $%
\blacksquare $

\bibliographystyle{econometrica}
\bibliography{common-agency}

\clearpage
\setcounter{page}{1} \renewcommand{\thepage}{SA-\arabic{page}}

\begin{center}
{\Large \textbf{Supplemental Appendix}}\\[0.75em]
{\large \textbf{Contracting with Imperfect Commitment: Minimal Canonical
Contracts}}\\[0.5em]
Seungjin Han and Siyang Xiong
\end{center}

\noindent This supplemental appendix contains formal assumptions and proofs
deferred from the main paper.

\section*{Additional Proofs}

\subsection*{Auxiliary facts for robust canonicality}

The continuation-equilibrium canonicality result is proved in Appendix \ref%
{sec:prelim_analysis:2}. Here we record the remaining robust-equilibrium
comparison facts used in Appendix \ref{sec:minimal_contract_muti_principals}
and in the detailed proof of Theorem \ref{thm:minimal:weak-robust} below.
The proof uses two comparison results: $\mathcal{G}^{\sharp }$ gives a lower
bound for robust-equilibrium allocations, while $\mathcal{G}^{\ast }$ gives
an upper bound. For the lower bound, introduce a second binary relation on
contracts. For any two contracts $G_{j}:M_{j}\longrightarrow X_{j}$ and $%
G_{j}^{\prime }:M_{j}^{\prime }\longrightarrow X_{j}$, define%
\begin{equation*}
G_{j}^{\prime }\sqsupset _{j}G_{j}\Longleftrightarrow \left(
\begin{array}{c}
\text{there exists a measurable surjective function }\iota
_{j}:M_{j}^{\prime }\longrightarrow M_{j}\text{,} \\
\text{admitting a measurable right inverse, and such that} \\
G_{j}^{\prime }\left( m_{j}^{\prime }\right) =G_{j}\left( \iota _{j}\left(
m_{j}^{\prime }\right) \right) \text{ for any }m_{j}^{\prime }\in
M_{j}^{\prime }%
\end{array}%
\right) \text{.}
\end{equation*}

\begin{lemma}
\label{lem:extension_w}Consider any contract space $\mathcal{G}$ and any
allocation $z$ induced by a $\mathcal{G}$-continuation equilibrium $\left[ G%
\text{, }\left( \gamma ,q,p\right) \right] $. For any contract space $%
\mathcal{G}^{\prime }$ with $\mathcal{G}\subseteq \mathcal{G}^{\prime }$,
any $j\in \mathcal{J}$, and any $G_{j}^{\prime }\in \mathcal{G}_{j}^{\prime
} $ with $G_{j}^{\prime }\sqsupset _{j}G_{j}$, there is a $\mathcal{G}%
^{\prime }$-continuation equilibrium in which principal $j$ uses $%
G_{j}^{\prime }$, the other principals use $G_{-j}$, and the induced
allocation is $z$.
\end{lemma}

\begin{proof}
Since $G_{j}^{\prime }\sqsupset _{j}G_{j},$ there exists a measurable
surjective function $\iota _{j}:M_{j}^{\prime }\longrightarrow M_{j}$ with a
measurable right inverse $\rho _{j}:M_{j}\rightarrow M_{j}^{\prime }$ such
that $G_{j}^{\prime }\left( m_{j}^{\prime }\right) =G_{j}\left( \iota
_{j}\left( m_{j}^{\prime }\right) \right) $ for any $m_{j}^{\prime }\in
M_{j}^{\prime }$ and $\iota _{j}(\rho _{j}(m_{j}))=m_{j}$ for every $%
m_{j}\in M_{j}$.

Define $q^{\prime }$ as the push-forward of $q$ under the map $%
(m_{j},m_{-j})\mapsto (\rho _{j}(m_{j}),m_{-j})$. Equivalently, for every
measurable $E\subset M_{j}^{\prime }\times M_{-j}$,
\begin{equation*}
q^{\prime }(\theta )[E]=q(\theta )\left[ \{m\in M:(\rho
_{j}(m_{j}),m_{-j})\in E\}\right] .
\end{equation*}%
Thus messages outside $\rho _{j}(M_{j})\times M_{-j}$ are reached with zero
probability. For every message profile $(m_{j}^{\prime },m_{-j})\in
M_{j}^{\prime }\times M_{-j}$ and every principal $\ell $, set
\begin{equation*}
\gamma _{\ell }^{\prime }(m_{j}^{\prime },m_{-j})=\gamma _{\ell }(\iota
_{j}(m_{j}^{\prime }),m_{-j}),\qquad p_{\ell }^{\prime }(m_{j}^{\prime
},m_{-j})=p_{\ell }(\iota _{j}(m_{j}^{\prime }),m_{-j}).
\end{equation*}

The construction preserves the induced contractible actions because $%
G_{j}^{\prime }(m_{j}^{\prime })=G_{j}(\iota _{j}(m_{j}^{\prime }))$. On the
path, the profile is just a relabeling of the original continuation
equilibrium. Off the path, every possible message under $G_{j}^{\prime }$ is
assigned the same continuation action and belief as the original message
profile to which $\iota _{j}$ maps it. Hence no new payoff is available to
the agent, and each principal's continuation action remains optimal.
Therefore
\begin{equation*}
\left[ (G_{j}^{\prime },G_{-j}),(\gamma ^{\prime },q^{\prime },p^{\prime })%
\right]
\end{equation*}%
is a $\mathcal{G}^{\prime }$-continuation equilibrium inducing the same
allocation $z$.
\end{proof}

\bigskip

\begin{lemma}
\label{lem:max:robust:off-equilibrium-pah}For any contract space $\mathcal{G}%
=\times _{k\in \mathcal{J}}\mathcal{G}_{k}$ with $\mathcal{G}^{\sharp
}\subseteq \mathcal{G}$, a profile $\left[ G=\times _{k\in \mathcal{J}}G_{k}%
\text{, }\left( \gamma ,q,p\right) \right] $ is a $\mathcal{G}$-robust
equilibrium if and only if it is a $\widehat{\mathcal{G}}$-robust
equilibrium with $\widehat{\mathcal{G}}=\times _{k\in \mathcal{J}}\widehat{%
\mathcal{G}}_{k}=\times _{k\in \mathcal{J}}\left[ \left\{ G_{k}\right\} \cup
\mathcal{G}_{k}^{\sharp }\right] $.
\end{lemma}

\begin{proof}
Consider any contract space $\mathcal{G}=\times _{k\in \mathcal{J}}\mathcal{G%
}_{k}$. The "only if" part of Lemma \ref{lem:max:robust:off-equilibrium-pah}
is immediately implied by Definition \ref{def:equilibrium:robust} and $%
\widehat{\mathcal{G}}\subseteq \mathcal{G}$. To prove the "if" part, we fix
any $\widehat{\mathcal{G}}$-robust equilibrium $\left[ G=\left( G_{k}\right)
_{k\in \mathcal{J}}\text{, }\left( \gamma ,q,p\right) \right] $, and aim to
show that it is a $\mathcal{G}$-robust equilibrium. Consider any $j\in
\mathcal{J}$ and any unilateral deviation by $j$, written $G_{j}^{\prime
}:M_{j}^{\prime }\rightarrow X_{j}$, with $G_{j}^{\prime }\in \mathcal{G}%
_{j} $. We can find a contract $G_{j}^{\prime \prime }\in \mathcal{G}%
_{j}^{\sharp }$ such that $G_{j}^{\prime }\sqsupset _{j}G_{j}^{\prime \prime
}$: since $G_{j}^{\prime }$ is admissible, its image lies in $\mathcal{L}%
_{j}^{\ast }$. Choose a measurable feasible-action selector $y_{j}(x_{j})\in
F_{j}(x_{j})$ on this image and let $G_{j}^{\prime \prime }$ offer exactly
the pairs $(x_{j},y_{j}(x_{j}))$. Mapping each message of $G_{j}^{\prime }$
to the corresponding pair in $G_{j}^{\prime \prime }$ gives the required
surjection, and the maintained selector for $G_{j}^{\prime }$ gives a
measurable right inverse. Since $\left[ G\text{, }\left( \gamma ,q,p\right) %
\right] $ is a $\widehat{\mathcal{G}}$-robust equilibrium, there exists a $%
\widehat{\mathcal{G}}$-continuation equilibrium $\left[ \left( G_{j}^{\prime
\prime },G_{-j}\right) \text{, }\left( \gamma ^{\prime \prime },q^{\prime
\prime },p^{\prime \prime }\right) \right] $ such that%
\begin{equation}
V_{j}\left[ G\text{, }\left( \gamma ,q,p\right) \right] \geq V_{j}\left[
\left( G_{j}^{\prime \prime },G_{-j}\right) \text{, }\left( \gamma ^{\prime
\prime },q^{\prime \prime },p^{\prime \prime }\right) \right] \text{.}
\label{at4}
\end{equation}%
By $G_{j}^{\prime }\sqsupset _{j}G_{j}^{\prime \prime }$ and Lemma \ref%
{lem:extension_w}, there is a $\mathcal{G}$-continuation equilibrium with
contract profile $\left( G_{j}^{\prime },G_{-j}\right) $; denote it by $%
\left[ \left( G_{j}^{\prime },G_{-j}\right) \text{, }\left( \gamma ^{\prime
},q^{\prime },p^{\prime }\right) \right] $. It satisfies%
\begin{equation}
V_{j}\left[ \left( G_{j}^{\prime },G_{-j}\right) \text{, }\left( \gamma
^{\prime },q^{\prime },p^{\prime }\right) \right] =V_{j}\left[ \left(
G_{j}^{\prime \prime },G_{-j}\right) \text{, }\left( \gamma ^{\prime \prime
},q^{\prime \prime },p^{\prime \prime }\right) \right] \text{.}  \label{at5}
\end{equation}%
(\ref{at4})-(\ref{at5}) imply that $\left[ G\text{, }\left( \gamma
,q,p\right) \right] $ is a $\mathcal{G}$-robust equilibrium.
\end{proof}

\bigskip

Lemma \ref{lem:max:robust:off-equilibrium-pah} immediately implies that $%
\mathcal{G}^{\sharp }$ establishes a lower bound for the robust equilibrium
allocation sets in the following sense.

\begin{prop}
\label{prop:canonical:lower-bound}For any $\mathcal{E}=\langle \mu
,U,u,(v_{k})_{k\in \mathcal{J}}\rangle $ and any contract space $\mathcal{G}$
with $\mathcal{G}^{\sharp }\subseteq \mathcal{G}$, we have%
\begin{equation*}
\mathcal{Z}^{\text{robust-}\mathcal{G}^{\sharp }\text{-}\mathcal{E}%
}\subseteq \mathcal{Z}^{\text{robust-}\mathcal{G}\text{-}\mathcal{E}}\text{.}
\end{equation*}
\end{prop}

It is not enough to check deviations only in $\mathcal{G}_{j}^{\ast }$. A
singleton-message deviation that induces $x_j$ when $F_j(x_j)=\{y_j,y_j^{%
\prime}\}$ cannot be obtained by relabeling any $\mathcal{G}_{j}^{\ast }$
contract: the singleton contract has no separate messages for the two
recommendations. Thus robustness against $\mathcal{G}_{j}^{\ast }$
deviations alone may fail to deter arbitrary deviations.

Nonetheless, $\mathcal{G}^{\ast }$ provides an upper bound.

\begin{lemma}
\label{lem:max:robust:on-equilibrium-pah}For any allocation $z$ which is
induced by a $\mathcal{G}$-robust equilibrium $\left[ G\text{, }\left(
\gamma ,q,p\right) \right] $ with $\mathcal{G}^{\ast }\subseteq \mathcal{G}$%
, there exists a $\mathcal{G}^{\ast }$-robust equilibrium $\left[ \left(
G_{k}^{G_{k}}\right) _{k\in \mathcal{J}}\text{, }\left( \gamma ^{\prime
},q^{\prime },p^{\prime }\right) \right] $ which induces $z$.
\end{lemma}

\begin{proof}
Consider a $\mathcal{G}$-robust equilibrium $\left[ G\text{, }\left( \gamma
,q,p\right) \right] $ which induces $z$. Lemma \ref{lem:extension_w2}
implies a $\mathcal{G}$-continuation equilibrium $\left[ \left(
G_{k}^{G_{k}}\right) _{k\in \mathcal{J}}\text{, }\left( \gamma ^{\prime
},q^{\prime },p^{\prime }\right) \right] $ which induces $z$. As a result,
we have%
\begin{equation}
V_{j}\left[ G\text{, }\left( \gamma ,q,p\right) \right] =V_{j}\left[ \left(
G_{k}^{G_{k}}\right) _{k\in \mathcal{J}}\text{, }\left( \gamma ^{\prime
},q^{\prime },p^{\prime }\right) \right] \text{, }\forall j\in \mathcal{J}%
\text{.}  \label{at1}
\end{equation}%
We now show that $\left[ \left( G_{k}^{G_{k}}\right) _{k\in \mathcal{J}}%
\text{, }\left( \gamma ^{\prime },q^{\prime },p^{\prime }\right) \right] $
is a $\mathcal{G}^{\ast }$-robust equilibrium. Consider any $j\in \mathcal{J}
$ and any unilateral deviation by $j$ to $G_{j}^{\prime \prime }\in \mathcal{%
G}_{j}^{\ast }$. Since $\left[ G\text{, }\left( \gamma ,q,p\right) \right] $
is a $\mathcal{G}$-robust equilibrium, there exists $\mathcal{G}$%
-continuation equilibrium $\left[ \left( G_{j}^{\prime \prime
},G_{-j}\right) \text{, }\left( \gamma ^{\prime \prime },q^{\prime \prime
},p^{\prime \prime }\right) \right] $ such that%
\begin{equation}
V_{j}\left[ G\text{, }\left( \gamma ,q,p\right) \right] \geq V_{j}\left[
\left( G_{j}^{\prime \prime },G_{-j}\right) \text{, }\left( \gamma ^{\prime
\prime },q^{\prime \prime },p^{\prime \prime }\right) \right] \text{.}
\label{at2}
\end{equation}%
Since $G_{j}^{\prime \prime }\in \mathcal{G}_{j}^{\ast }$, we have $%
G_{j}^{\prime \prime }=G_{j}^{G_{j}^{\prime \prime }}$, and by Lemma \ref%
{lem:extension_w2}, we can find a $\mathcal{G}$-continuation equilibrium $%
\left[ \left( G_{j}^{\prime \prime },\left( G_{k}^{G_{k}}\right) _{k\in
\mathcal{J}\diagdown \left\{ j\right\} }\right) \text{, }\left( \gamma
^{\prime \prime \prime },q^{\prime \prime \prime },p^{\prime \prime \prime
}\right) \right] $ such that%
\begin{equation}
V_{j}\left[ \left( G_{j}^{\prime \prime },\left( G_{k}^{G_{k}}\right) _{k\in
\mathcal{J}\diagdown \left\{ j\right\} }\right) \text{, }\left( \gamma
^{\prime \prime \prime },q^{\prime \prime \prime },p^{\prime \prime \prime
}\right) \right] =V_{j}\left[ \left( G_{j}^{\prime \prime },G_{-j}\right)
\text{, }\left( \gamma ^{\prime \prime },q^{\prime \prime },p^{\prime \prime
}\right) \right] \text{.}  \label{at3}
\end{equation}%
(\ref{at1})-(\ref{at3}) imply that $\left[ \left( G_{k}^{G_{k}}\right)
_{k\in \mathcal{J}}\text{, }\left( \gamma ^{\prime },q^{\prime },p^{\prime
}\right) \right] $ is a $\mathcal{G}^{\ast }$-robust equilibrium.
\end{proof}

Lemma \ref{lem:max:robust:on-equilibrium-pah} immediately implies that $%
\mathcal{G}^{\ast }$ establishes an upper bound for the robust equilibrium
allocation sets in the following sense.

\begin{prop}
\label{prop:canonical:upper-bound}For any $\mathcal{E=}\left\langle \mu
\text{, }U,u\text{, }\left( v_{k}\right) _{k\in \mathcal{J}}\right\rangle $
and any contract space $\mathcal{G}$ with $\mathcal{G}^{\ast }\subseteq
\mathcal{G}$, we have%
\begin{equation*}
\mathcal{Z}^{\text{robust-}\mathcal{G}\text{-}\mathcal{E}}\subseteq \mathcal{%
Z}^{\text{robust-}\mathcal{G}^{\ast }\text{-}\mathcal{E}}\text{.}
\end{equation*}
\end{prop}

\subsection*{Detailed proof of Theorem \protect\ref{thm:minimal:weak-robust}}

Given $\left\vert \mathcal{J}\right\vert \geq 2$, in order to prove Theorem %
\ref{thm:minimal:weak-robust}, it is sufficient to show that $\mathcal{G}%
^{\ast }$ is a canonical contract space due to Proposition \ref%
{thm:lower:weak-robust}. Fix any $\mathcal{E=}\left\langle \mu \text{, }U%
\text{, }u\text{, }\left( v_{k}\right) _{k\in \mathcal{J}}\right\rangle $
and any allocation $z$, and we aim to prove%
\begin{equation}
z\in \mathcal{Z}^{\text{robust-}\mathcal{G}^{\ast }\text{-}\mathcal{E}%
}\Longleftrightarrow \left[ z\text{ is induced by a robust equilibrium}%
\right] \text{.}  \label{G**_minimal}
\end{equation}%
For the \textquotedblleft $\Longleftarrow $\textquotedblright\ direction,
suppose that $z$ is induced by a robust equilibrium $\left[ G\text{, }\left(
\gamma ,q,p\right) \right] $. Because $\left[ G\text{, }\left( \gamma
,q,p\right) \right] $ is a robust equilibrium, it is also a $\overline{%
\mathcal{G}}$-robust equilibrium with $\overline{\mathcal{G}}=\times _{k\in
\mathcal{J}}\overline{\mathcal{G}}_{k}=\times _{k\in \mathcal{J}}\left[
\left\{ G_{k}\right\} \cup \mathcal{G}_{k}^{\ast }\right] $. Thus,
Proposition \ref{prop:canonical:upper-bound} implies $z\in \mathcal{Z}^{%
\text{robust-}\mathcal{G}^{\ast }\text{-}\mathcal{E}}$.

For the \textquotedblleft $\Longrightarrow $\textquotedblright\ direction,
suppose $z$ is induced by a $\mathcal{G}^{\ast }$-robust equilibrium $\left[
G\text{, }\left( \gamma ,q,p\right) \right] $. The profile $G$ itself need
not be robust to deviations outside $\mathcal{G}^{\ast }$. Instead, we
construct an allocation-equivalent robust equilibrium with enlarged message
spaces. Write $G_k:M_k\rightarrow X_k$ for each $k$, set
\begin{equation*}
M:=\times_{k\in\mathcal{J}}M_k,\qquad \mathcal{M}^\circ:=M.
\end{equation*}
Thus the extra coordinate used below records an original message profile,
not an arbitrary feasible action profile.

For each $j\in \mathcal{J}$ and each $\left\langle
G_{j}:M_{j}\longrightarrow X_{j}\right\rangle \in \mathcal{G}_{j}^{\ast }$,
we can extend it to the following contract, $G_{j}^{\circ }:M_{j}\times
\mathcal{M}^{\circ }\longrightarrow X_{j}$ such that%
\begin{equation*}
G_{j}^{\circ }\left( m_j,\widetilde m\right)=G_j(m_j)\text{, }\forall
\left(m_j,\widetilde m\right)\in M_j\times \mathcal{M}^\circ\text{.}
\end{equation*}%
The relation $G_{j}^{\circ }\sqsupset _{j}G_{j}$ holds by construction, and
the only difference between them is that the former has more messages, or
precisely, the $\mathcal{M}^{\circ }$ dimension in the message set of $%
G_{j}^{\circ }$. Since $\mathcal{M}^\circ$ is a product of admissible
standard Borel message spaces and $G_j^\circ$ has the same image as $G_j$, $%
G_j^\circ$ is admissible.

Let $\mathcal{G}_{j}^{\circ }$ denote the set of all such $G_{j}^{\circ }$,
and $\mathcal{G}^{\circ }:=\times _{k\in \mathcal{J}}\mathcal{G}_{k}^{\circ
} $.

\begin{lemma}
\label{lem:extension:reverse}Suppose $\left\vert \mathcal{J}\right\vert \geq
2$. Consider any $G\in \mathcal{G}^{\ast }$, any $j\in \mathcal{J}$ and any $%
G_{j}^{\prime }\in \mathcal{G}_{j}^{\sharp }$ such that $G_{j}\sim
G_{j}^{\prime }$. If allocation $z$ is induced by a $\mathcal{G}^{\ast }$%
-continuation equilibrium $\left[ G\text{, }\left( \gamma ,q,p\right) \right]
$, there exists a $\widehat{\mathcal{G}}^{\circ,j}\left( G_j^{\prime}\right)
$-continuation equilibrium
\begin{equation*}
\left[ \left( G_{j}^{\prime },\left( G_{k}^{\circ }\right) _{k\in \mathcal{J}%
\diagdown \left\{ j\right\} }\right) \text{, }\left( \gamma ^{\prime
},q^{\prime },p^{\prime }\right) \right]
\end{equation*}%
which induces $z$, where
\begin{equation*}
\widehat{\mathcal{G}}^{\circ,j}\left( G_j^{\prime}\right) =\left[ \left\{
G_j^{\prime}\right\} \cup \mathcal{G}_{j}^{\ast }\right] \times \prod_{k\in
\mathcal{J}\diagdown \left\{ j\right\} } \left[ \left\{ G_{k}^{\circ
}\right\} \cup \mathcal{G}_{k}^{\ast }\right] .
\end{equation*}
\end{lemma}

\begin{proof}
Consider any $G\in \mathcal{G}^{\ast }$, where $G_{k}:M_{k}\longrightarrow
X_{k}$ for each $k$. Fix any $j\in \mathcal{J}$ and any $G_{j}^{\prime
}:M_{j}^{\prime }\longrightarrow X_{j}$ in $\mathcal{G}_{j}^{\sharp }$ such
that $G_{j}\sim G_{j}^{\prime }$. Since $G_j^\prime\in\mathcal{G}_j^\sharp$,
the projection $M_j^\prime\to G_j^\prime(M_j^\prime)$ admits a measurable
right inverse $\rho_j^\prime$. Define
\begin{equation*}
\phi_j(m_j):=\rho_j^\prime(G_j(m_j)).
\end{equation*}
Then $\phi_j:M_j\to M_j^\prime$ is measurable and satisfies $G_{j}\left(
m_{j}\right) =G_{j}^{\prime }\left[ \phi _{j}\left( m_{j}\right) \right] $
for every $m_{j}\in M_{j}$.

Suppose $z$ is induced by a $\mathcal{G}^{\ast }$-continuation equilibrium $%
\left[ G\text{, }\left( \gamma ,q,p\right) \right] $. We now construct a $%
\widehat{\mathcal{G}}^{\circ,j}\left( G_j^{\prime}\right) $-continuation
equilibrium $\left[ \left( G_{j}^{\prime },\left( G_{k}^{\circ }\right)
_{k\in \mathcal{J}\diagdown \left\{ j\right\} }\right) \text{, }\left(
\gamma ^{\prime },q^{\prime },p^{\prime }\right) \right] $ which induces $z$%
. We first construct an injective function%
\begin{gather*}
\Psi :\times _{k\in \mathcal{J}}M_{k}\longrightarrow M_{j}^{\prime }\times %
\left[ \times _{k\in \mathcal{J}\diagdown \left\{ j\right\} }\left(
M_{k}\times \mathcal{M}^{\circ }\right) \right] \text{,} \\
\Psi \left[ \left( x_{k},y_{k}\right) _{k\in \mathcal{J}}\right] =\left[
\phi _{j}\left( x_{j},y_{j}\right) ,\text{ }\left[ \left( x_{\widetilde{k}%
},y_{\widetilde{k}}\right) ,\left( x_{k},y_{k}\right) _{k\in \mathcal{J}}%
\right] _{_{\widetilde{k}\in \mathcal{J}\diagdown \left\{ j\right\} }}\right]
\text{.}
\end{gather*}%
Because $\left\vert \mathcal{J}\right\vert \geq 2$, the additional $\mathcal{%
M}^\circ$ coordinate sent to every $k\neq j$ records the original message
profile. Thus, $\Psi $ is a bijection between $M=\times _{k\in \mathcal{J}%
}M_{k}$ and $\Psi \left( M\right) $. Let $\Psi ^{-1}$ denote the inverse
function,
\begin{equation*}
\Psi ^{-1}:\Psi \left( M\right) \longrightarrow M\text{, and }\Psi ^{-1}%
\left[ \Psi \left( m\right) \right] =m\text{, }\forall m\in M\text{.}
\end{equation*}%
For each $\theta \in \Theta $, we have $q\left( \theta \right) \in \triangle
\left( M\right) $, and let $q^{\prime }\left( \theta \right) $ be the
push-forward distribution under $\Psi $:
\begin{equation*}
q^{\prime }\left( \theta \right) \left[ E\right] =q\left( \theta \right) %
\left[ \Psi ^{-1}\left( E\cap \Psi \left( M\right) \right) \right]
\end{equation*}
for every measurable set in the enlarged message space. Hence messages
outside $\Psi(M)$ are reached with zero probability.

For each $k$, let $\rho_k:G_k(M_k)\rightarrow M_k$ be the measurable
selector associated with the admissible contract $G_k$. Because $%
G_j^{\prime}\sim G_j$, every contractible-action profile generated by $%
\left(G_j^{\prime},G_{-j}^\circ\right)$ lies in $\times_{k\in\mathcal{J}%
}G_k(M_k)$. Define the measurable map
\begin{equation*}
\xi:\times_{k\in\mathcal{J}}G_k(M_k)\longrightarrow M,\qquad \xi[(x_k)_{k\in%
\mathcal{J}}]:=(\rho_k(x_k))_{k\in\mathcal{J}}.
\end{equation*}
Furthermore, for each principal $\ell \in \mathcal{J}$, define $\gamma
_{\ell }^{\prime }$ and $p_{\ell }^{\prime }$ as follows. For any $\left(
m_{k}\right) _{k\in \mathcal{J}}\in M_{j}^{\prime }\times \left[ \times
_{k\in \mathcal{J}\diagdown \left\{ j\right\} }\left( M_{k}\times \mathcal{M}%
^{\circ }\right) \right] $, {\small
\begin{equation*}
\gamma _{\ell }^{\prime }\left[ \left( m_{k}\right) _{k\in \mathcal{J}}%
\right] =\left\{
\begin{tabular}{ll}
$\gamma _{\ell }\left[ \left( x_{k},y_{k}\right) _{k\in \mathcal{J}}\right]
\text{,}$ &
\begin{tabular}{l}
if $\exists \left( x_{k},y_{k}\right) _{k\in \mathcal{J}}\in M$ such \\
that $\left( m_{k}\right) _{k\in \mathcal{J}}=\Psi \left[ \left(
x_{k},y_{k}\right) _{k\in \mathcal{J}}\right] $,%
\end{tabular}
\\
&  \\
$\gamma _{\ell }\left[ \xi \left( G_{j}^{\prime }\left( m_{j}\right) ,\left[
G_{k}^{\circ }\left( m_{k}\right) \right] _{k\in \mathcal{J}\diagdown
\left\{ j\right\} }\right) \right] \text{,}$ & otherwise.%
\end{tabular}%
\right.
\end{equation*}
given his belief \begingroup%
\begin{equation*}
p_{\ell }^{\prime }\left[ \left( m_{k}\right) _{k\in \mathcal{J}}\right]
=\left\{
\begin{tabular}{ll}
$p_{\ell }\left[ \left( x_{k},y_{k}\right) _{k\in \mathcal{J}}\right] \text{,%
}$ &
\begin{tabular}{l}
if $\exists \left( x_{k},y_{k}\right) _{k\in \mathcal{J}}\in M$ such \\
that $\left( m_{k}\right) _{k\in \mathcal{J}}=\Psi \left[ \left(
x_{k},y_{k}\right) _{k\in \mathcal{J}}\right] $,%
\end{tabular}
\\
&  \\
$p_{\ell }\left[ \xi \left( G_{j}^{\prime }\left( m_{j}\right) ,\left[
G_{k}^{\circ }\left( m_{k}\right) \right] _{k\in \mathcal{J}\diagdown
\left\{ j\right\} }\right) \right] \text{,}$ & otherwise.%
\end{tabular}%
\right. .
\end{equation*}
\endgroup
}

The map $\Psi $ embeds each equilibrium message profile from the original
profile into the enlarged message space. On the image of $\Psi $, actions
and beliefs coincide with the original equilibrium; off that image, $\xi $
selects a feasible action profile and the associated beliefs make the
specified continuation actions optimal. The enlarged message space also
creates no profitable deviation for the agent. Every message in $\Psi(M)$
induces exactly the payoff of the corresponding original message profile $%
\Psi^{-1}(\cdot)$. Every message outside $\Psi(M)$ induces the same
contractible-action profile and continuation-action profile as the original
message profile selected by $\xi$. Thus every payoff attainable by an agent
deviation in the enlarged game was already attainable by a message deviation
in the original $\mathcal{G}^{\ast}$-continuation equilibrium. The agent's
incentive and participation constraints therefore remain satisfied. Hence $%
\left[ \left( G_{j}^{\prime },\left( G_{k}^{\circ }\right) _{k\in \mathcal{J}%
\diagdown \left\{ j\right\} }\right) \text{, }\left( \gamma ^{\prime
},q^{\prime },p^{\prime }\right) \right] $ is a $\widehat{\mathcal{G}}%
^{\circ ,j}\left( G_{j}^{\prime }\right) $-continuation equilibrium that is
allocation-equivalent to $\left[ G\text{, }\left( \gamma ,q,p\right) \right]
$.
\end{proof}

\bigskip

The intuition of Lemma \ref{lem:extension:reverse} is clear. When the proof
replicates the original equilibrium, the apparent problem is that $%
G_{j}^{\prime }$ may not provide enough messages to reproduce the messages
sent to $j$ under $G_{j}$. The construction uses the enlarged message sets
in the $G_{k}^{\circ }$ contracts of the other principals to encode the
missing information. This scheme works because contracting is public: every
principal observes the full profile of messages, so the messages originally
sent to $j$ can be represented through the additional $\mathcal{M}^{\circ }$
coordinates sent to the other principals.

Using Lemma \ref{lem:extension:reverse}, we can establish the following
lemma.

\begin{lemma}
\label{lem:robust:off-equilibrium}Suppose $\left\vert \mathcal{J}\right\vert
\geq 2$. For any $\mathcal{G}^{\ast }$-robust equilibrium that induces $z$,
there exists a $\widehat{\mathcal{G}}^{\circ }$-robust equilibrium with
contract profile $G^{\circ }=\left( G_{k}^{\circ }\right) _{k\in \mathcal{J}%
} $ that induces $z$, where $\widehat{\mathcal{G}}^{\circ }=\times _{k\in
\mathcal{J}}\left[ \left\{ G_{k}^{\circ }\right\} \cup \mathcal{G}%
_{k}^{\sharp }\right] $.
\end{lemma}

\begin{proof}
Suppose $z$ is induced by a $\mathcal{G}^{\ast }$-robust equilibrium $\left[
G\text{, }\left( \gamma ,q,p\right) \right] $. Let $G^{\circ }=\left(
G_{k}^{\circ }\right) _{k\in \mathcal{J}}$ and $G_{-j}^{\circ }=\left(
G_{k}^{\circ }\right) _{k\in \mathcal{J}\diagdown \left\{ j\right\} }$.
Since $G_{j}^{\circ }\sqsupset _{j}G_{j}$ for every $j\in \mathcal{J}$, we
can apply Lemma \ref{lem:extension_w} inductively on $\mathcal{J}$ to show
the existence of a $\widehat{\mathcal{G}}^{\circ }$-continuation equilibrium
$\left[ G^{\circ }\text{, }\left( \gamma ^{\prime },q^{\prime },p^{\prime
}\right) \right] $ which induces $z$. As a result,%
\begin{equation}
V_{j}\left[ G\text{, }\left( \gamma ,q,p\right) \right] =V_{j}\left[
G^{\circ }\text{, }\left( \gamma ^{\prime },q^{\prime },p^{\prime }\right) %
\right] \text{, }\forall j\in \mathcal{J}\text{.}  \label{ttk1}
\end{equation}

To show that $\left[ G^{\circ }\text{, }\left( \gamma ^{\prime },q^{\prime
},p^{\prime }\right) \right] $ is a $\widehat{\mathcal{G}}^{\circ }$-robust
equilibrium, consider any $j\in \mathcal{J}$ and any deviating mechanism $%
\left\langle G_{j}^{\prime }:M_{j}^{\prime }\longrightarrow
X_{j}\right\rangle \in \mathcal{G}_{j}^{\sharp }$. We can find $%
G_{j}^{\prime \prime }\in \mathcal{G}_{j}^{\ast }$ such that $G_{j}^{\prime
\prime }\sqsupset _{j}G_{j}^{\prime }$, which further implies $G_{j}^{\prime
\prime }\sim G_{j}^{\prime }$: let $G_{j}^{\prime \prime }$ offer all
feasible recommendations for every contractible action in $G_{j}^{\prime
}(M_{j}^{\prime })$. Because $G_j^\prime\in\mathcal{G}_j^\sharp$, the
projection $M_j^\prime\to G_j^\prime(M_j^\prime)$ admits a measurable right
inverse $\rho_j^\prime$. Relabel $G_j^{\prime\prime}$ into $G_j^\prime$ by
mapping each message already present in $G_j^\prime$ to itself and mapping
each additional message $(x_j,y_j)$ to $\rho_j^\prime(x_j)$. This measurable
relabeling preserves the contractible action, so $G_j^{\prime\prime}\sim
G_j^\prime$. Since $\left[ G\text{, }\left( \gamma ,q,p\right) \right] $ is
a $\mathcal{G}^{\ast }$-robust equilibrium and $G_{j}^{\prime \prime }\in
\mathcal{G}_{j}^{\ast }$, there exists a $\mathcal{G}^{\ast }$-continuation
equilibrium $\left[ \left( G_{j}^{\prime \prime },G_{-j}\right) \text{, }%
\left( \gamma ^{\prime \prime },q^{\prime \prime },p^{\prime \prime }\right) %
\right] $ such that
\begin{equation}
V_{j}\left[ G\text{, }\left( \gamma ,q,p\right) \right] \geq V_{j}\left[
\left( G_{j}^{\prime \prime },G_{-j}\right) \text{, }\left( \gamma ^{\prime
\prime },q^{\prime \prime },p^{\prime \prime }\right) \right] \text{.}
\label{ttk2}
\end{equation}%
By Lemma \ref{lem:extension:reverse} and $G_{j}^{\prime \prime }\sim
G_{j}^{\prime }$, there is a $\widehat{\mathcal{G}}^{\circ ,j}\left(
G_{j}^{\prime }\right) $-continuation equilibrium with contract profile $%
\left( G_{j}^{\prime },G_{-j}^{\circ }\right) $. This contract space is
contained in the present $\widehat{\mathcal{G}}^{\circ }$, so the same
profile is also a $\widehat{\mathcal{G}}^{\circ }$-continuation equilibrium.
Denote it by $\left[ \left( G_{j}^{\prime },G_{-j}^{\circ }\right) \text{, }%
\left( \widehat{\gamma },\widehat{q},\widehat{p}\right) \right] $. It
induces the same allocation as the equilibrium in (\ref{ttk2}), so%
\begin{equation}
V_{j}\left[ \left( G_{j}^{\prime \prime },G_{-j}\right) \text{, }\left(
\gamma ^{\prime \prime },q^{\prime \prime },p^{\prime \prime }\right) \right]
=V_{j}\left[ \left( G_{j}^{\prime },G_{-j}^{\circ }\right) \text{, }\left(
\widehat{\gamma },\widehat{q},\widehat{p}\right) \right] \text{.}
\label{ttk3}
\end{equation}%
(\ref{ttk1}), (\ref{ttk2}) and (\ref{ttk3}) imply%
\begin{equation}
V_{j}\left[ G^{\circ }\text{, }\left( \gamma ^{\prime },q^{\prime
},p^{\prime }\right) \right] \geq V_{j}\left[ \left( G_{j}^{\prime
},G_{-j}^{\circ }\right) \text{, }\left( \widehat{\gamma },\widehat{q},%
\widehat{p}\right) \right] \text{.}  \label{ttk4}
\end{equation}%
Therefore, $\left[ G^{\circ }\text{, }\left( \gamma ^{\prime },q^{\prime
},p^{\prime }\right) \right] $ is a $\widehat{\mathcal{G}}^{\circ }$-robust
equilibrium.
\end{proof}

\bigskip

Lemmas \ref{lem:max:robust:off-equilibrium-pah} and \ref%
{lem:robust:off-equilibrium} provide a concise proof of Proposition \ref%
{prop:robust:equilibrium}

\bigskip

\noindent \textbf{Proof of Proposition \ref{prop:robust:equilibrium} }%
Consider any $\mathcal{G}^{\ast }$-robust equilibrium $\left[ G\text{, }%
\left( \gamma ,q,p\right) \right] $ that induces $z$. By Lemma \ref%
{lem:robust:off-equilibrium}, there is a $\widehat{\mathcal{G}}^{\circ }$%
-robust equilibrium $\left[ G^{\circ }\text{, }\left( \gamma ^{\prime
},q^{\prime },p^{\prime }\right) \right] $ that induces $z$, where
\begin{equation*}
G^{\circ }=\left( G_{k}^{\circ }\right) _{k\in \mathcal{J}}\quad \text{and}%
\quad \widehat{\mathcal{G}}^{\circ }=\times _{k\in \mathcal{J}}\left[
\left\{ G_{k}^{\circ }\right\} \cup \mathcal{G}_{k}^{\sharp }\right] .
\end{equation*}%
We now show that this equilibrium is robust. Consider any contract space $%
\mathcal{G}=\times _{k\in \mathcal{J}}\mathcal{G}_{k}$ with $G^{\circ }\in
\mathcal{G}$. Let $\overline{\mathcal{G}}=\times _{k\in \mathcal{J}}\left[
\mathcal{G}_{k}\cup \widehat{\mathcal{G}}_{k}^{\circ }\right] $. Lemma \ref%
{lem:max:robust:off-equilibrium-pah} implies that $\left[ G^{\circ }\text{, }%
\left( \gamma ^{\prime },q^{\prime },p^{\prime }\right) \right] $ is a $%
\overline{\mathcal{G}}$-robust equilibrium. Since $\mathcal{G}\subseteq
\overline{\mathcal{G}}$, it is also a $\mathcal{G}$-robust equilibrium. $%
\blacksquare $

\bigskip

We now prove Theorem \ref{thm:minimal:weak-robust} below by recapping the
logic presented above.

\bigskip

\noindent \textbf{Proof of Theorem \ref{thm:minimal:weak-robust}}. Because
of Proposition \ref{thm:lower:weak-robust}, we only need to show that (\ref%
{G**_minimal}) holds given any $\left\langle \mu \text{, }u\text{, }\left(
v_{k}\right) _{k\in \mathcal{J}}\right\rangle $ and any allocation $z$. For
the \textquotedblleft $\Longleftarrow $\textquotedblright\ direction,
suppose $z$ is induced by an unrestricted robust equilibrium. Since that
profile is robust against every admissible contract space containing its
contract profile, Proposition \ref{prop:canonical:upper-bound} gives a $%
\mathcal{G}^{\ast}$-robust equilibrium inducing the same allocation. For the
\textquotedblleft $\Longrightarrow $\textquotedblright\ direction,
Proposition \ref{prop:robust:equilibrium} embeds the given $\mathcal{G}%
^{\ast}$-robust equilibrium in enlarged message spaces and restores
robustness to arbitrary deviations. Hence (\ref{G**_minimal}) holds,
completing the proof of Theorem \ref{thm:minimal:weak-robust}. $\blacksquare
$

\subsection*{Proof of Proposition \protect\ref{proposition_single_optimal}}

By Proposition \ref{proposition_single}, the contracting value within the
atomic equilibrium class is attainable with a single contractible-action
offer. Fix such an offer $x$. After observing the principal's
non-contractible action $y$, type $\theta$ stays if and only if $%
u(x,y,\theta)\geq 0$. Hence the principal's payoff from the pair $(x,y)$ is
\begin{equation*}
\int_{\Theta}\mathbf{1}_{\{u(x,y,\theta)\geq 0\}}v(x,y,\theta)\,\mu[d\theta].
\end{equation*}
The pair is relevant only when some type is willing to stay, that is, when $%
(x,y)\in C$; for a fixed $x$, this means $y\in C(x)$, and the feasible
contractible offers are exactly $x\in B$. Therefore the best payoff
generated by a given single offer $x$ is the inner maximization in (\ref%
{optimal_contract}), and optimizing over all feasible offers gives the outer
maximization. Thus a single offer $x^\ast$ is optimal if and only if it
solves (\ref{optimal_contract}). $\blacksquare $

\subsection*{Calculation for the common-agency labor example}

The zero offer cannot yield positive profit, so focus on $x_{j}>0$. For each
firm $j$, the worker remains if and only if
\begin{equation*}
u_{j}\geq 0 \iff \theta\geq t_{j}:=\frac{y_{j}^{2}}{x_{j}}, \qquad j=1,2.
\end{equation*}
Let
\begin{equation*}
a_{j}:=\max\left\{3,\frac{y_{j}^{2}}{x_{j}}\right\}.
\end{equation*}
Then the worker accepts firm $j$ in the continuation stage if and only if $%
\theta\in[a_{j},4]$. Fix $x_{-j}$ and write $A_{j}=1+\beta x_{-j}$. For $%
t\in \lbrack 3,4)$, expected profit is
\begin{equation*}
\Pi _{j}(x_{j},t\mid x_{-j})=(4-t)\left( A_{j}\sqrt{x_{j}t}\cdot \frac{t+4}{2%
}-x_{j}^{2}\right) =(4-t)\left( A_{j}K(t)\sqrt{x_{j}}-x_{j}^{2}\right) ,
\end{equation*}%
where $K(t):=\sqrt{t}(t+4)/2$. For $t\leq 3$, all types stay and the best
no-exit cutoff is $t=3$. For each $t\in \lbrack 3,4)$, the optimal
contractible action solves
\begin{equation*}
x_{j}(t\mid x_{-j})=\left( \frac{A_{j}K(t)}{4}\right) ^{2/3}=\left( A_{j}%
\frac{\sqrt{t}(t+4)}{8}\right) ^{2/3}.
\end{equation*}%
Substitution shows that the $t$-dependent part of firm $j$'s value is
proportional to
\begin{equation*}
(4-t)t^{2/3}(t+4)^{4/3}.
\end{equation*}%
Its log derivative is
\begin{equation*}
\frac{32+4t-9t^{2}}{3t(t+4)(4-t)}<0\qquad \text{for all }t\in \lbrack 3,4),
\end{equation*}%
so the optimal cutoff is $t^{\ast }=3$. Hence firm $j$'s stage-1 best
response is
\begin{equation*}
BR_{j}^{x}(x_{-j})=\left[ \left( 1+\beta x_{-j}\right) \frac{7\sqrt{3}}{8}%
\right] ^{2/3}.
\end{equation*}

\subsection*{Formal common-agency assumptions for Proposition \protect\ref%
{prop_pooling_common}}

The following assumptions spell out the bundled common-agency reduction
condition used in Section \ref{sec:appII}. We impose a separability
condition on the agent's utility.

\begin{assum}
\label{Ass_multi_separability} $u(x,y,\theta)$ is separable across
principals:
\begin{equation*}
u(x,y,\theta)=u_{1}(x_{1},y_{1},\theta)+u_{2}(x_{2},y_{2},\theta).
\end{equation*}
\end{assum}

Separability makes a unilateral deviation comparable to the single-principal
problem: the agent's continuation utility from principal $j$ can be
evaluated holding the other relationship fixed.

We also restrict cross-principal payoff dependence to the contractible
actions of other principals.

\begin{assum}
\label{Ass_multi_payoff_dependence} For each principal $j$, $%
v_j(x,y,\theta)=\widetilde v_j(x_j,y_j,x_{-j},\theta)$. If the agent exits
from principal $j$, principal $j$'s payoff is zero.
\end{assum}

Thus a unilateral deviation by principal $j$ does not affect principal $j$'s
payoff through the other principal's non-contractible action; the other
relationship enters only through the fixed contractible offer $x_{-j}$.

We maintain the same monotonicity conditions as in the single-principal case.

\begin{assum}
\label{Ass_multi_theta} For each $j=1,2$, $u_{j}(x_{j},y_{j},\theta)$ is
strictly increasing in $\theta$, and $v_{j}(x,y,\theta)$ is non-decreasing
in $\theta$.
\end{assum}

Next, each principal prefers a higher non-contractible action, whereas the
agent prefers a lower one.

\begin{assum}
\label{Ass_multi_opp_monotone}\footnote{%
If instead $\frac{\partial v_{j}}{\partial y_{j}}<0$ and $\frac{\partial
u_{j}}{\partial y_{j}}>0$, one can redefine $y_{j}^{\prime }=-y_{j}$ and
apply the same analysis.} For each $j=1,2$,
\begin{equation*}
\frac{\partial v_{j}}{\partial y_{j}}>0\qquad \text{and}\qquad \frac{%
\partial u_{j}}{\partial y_{j}}<0.
\end{equation*}
\end{assum}

We also assume differentiability and impose single crossing separately for
each bilateral component.

\begin{assum}
\label{Ass_multi_single_crossing} For each $j=1,2$, $u_{j}(x_{j},y_{j},%
\theta)$ is differentiable in $(x_{j},y_{j})$ and satisfies the
single-crossing property in Assumption \ref{assu_single_crossing}, with $%
u_{j},x_{j},$ and $y_{j}$ replacing $u,x,$ and $y$.
\end{assum}

This is the bilateral version of the single-crossing condition used in
Section \ref{sec:appI}.

\begin{assum}[Bilateral local regularity]
\label{Ass_multi_local_rent_extraction} For each principal $j$ and each
fixed contractible action profile of the other principals, the bilateral
continuation problem with primitives $u_j(x_j,y_j,\theta)$ and $\widetilde
v_j(x_j,y_j,x_{-j},\theta)$ satisfies the local uniform-continuity and
integrability hypotheses used in Lemma \ref{lem:rent_extraction_common_slack}
whenever that lemma is applied to an on-path participating action-pair cell.
\end{assum}

\begin{assum}[Bilateral participating action-pair atomicity]
\label{Ass_multi_no_thin_randomization} For each principal $j$ and each
fixed contractible action profile of the other principals, the bilateral
continuation problem with primitives $u_j(x_j,y_j,\theta)$ and $\widetilde
v_j(x_j,y_j,x_{-j},\theta)$ is restricted to atomic
participating-action-pair equilibria for the aggregate measure of
participating recommendations.
\end{assum}

\begin{assum}[Principal-separable continuation]
\label{Ass_multi_principal_separable_continuation} Following any unilateral
deviation by principal $j$, the no-safe-deviation test admits a continuation
equilibrium in which the payoff contribution from the non-deviating
relationship is independent of principal $j$'s message. Equivalently,
holding the non-deviating contractible offer $x_{-j}^{\ast}$ fixed, the
agent's incentive comparisons across principal $j$'s messages are governed
by the bilateral utility $u_j(x_j,y_j,\theta)$, and principal $j$'s payoff
is $\widetilde v_j(x_j,y_j,x_{-j}^{\ast},\theta)$.
\end{assum}

Assumption \ref{Ass_multi_principal_separable_continuation} is automatic
when non-deviating principals' continuation actions are fixed at the
simple-offer continuation values, or more generally when those actions do
not vary with the deviating principal's message.

\subsection*{Proof of Proposition \protect\ref{prop_pooling_common}}

Fix an equilibrium $(x^{\ast },y^{\ast })$ of the single contractible-action
offer game and a principal $j$. Consider any deviation by $j$ to an
arbitrary mechanism. To rule out a safe deviation, it is enough to exhibit
one continuation equilibrium following the deviation that does not improve
principal $j$'s payoff. By the extended taxation principle, the deviation
can be represented, for continuation-payoff purposes, as a menu with
recommendations. Let $L_j$ be the induced set of contractible actions, and
take a principal-separable continuation equilibrium, as allowed by
Assumption \ref{Ass_multi_principal_separable_continuation}, with
continuation strategy $q:\Theta\rightarrow\Delta(L_j\times\mathbb{R}_{+})$
over chosen contractible actions and recommended non-contractible actions.

Holding the non-deviating offer $x_{-j}^{\ast}$ fixed, Assumption \ref%
{Ass_multi_separability} makes the agent's incentives and exit decision with
respect to principal $j$ depend only on $u_j(x_j,y_j,\theta)$ in that
continuation equilibrium, while Assumption \ref{Ass_multi_payoff_dependence}
makes principal $j$'s payoff depend on the other relationship only through $%
x_{-j}^{\ast}$. The deviating continuation problem is therefore the
single-principal problem with primitives $u_j(x_j,y_j,\theta)$ and $%
\widetilde v_j(x_j,y_j,x_{-j}^{\ast},\theta)$; Assumptions \ref%
{Ass_multi_opp_monotone} and \ref{Ass_multi_single_crossing} supply the
monotonicity and single-crossing structure, while Assumptions \ref%
{Ass_multi_local_rent_extraction} and \ref{Ass_multi_no_thin_randomization}
supply the local regularity and atomic participating-action-pair conditions
used in Proposition \ref{proposition_single}.

By Proposition \ref{proposition_single}, principal $j$'s payoff from this
deviating menu is no larger than the payoff from some single
contractible-action offer $\bar{x}_{j}\in L_{j}$ against the same $%
x_{-j}^{\ast }$. Thus a profitable complex safe deviation would imply a
profitable single-offer deviation, contradicting equilibrium of the
single-offer game. Since $j$ was arbitrary, the allocation induced by $%
(x^{\ast },y^{\ast })$ is robust. $\blacksquare $

\clearpage
\setcounter{page}{1} \renewcommand{\thepage}{PC-\arabic{page}}

\begin{center}
{\Large \textbf{Private Contracting with Imperfect Commitment}}\\[0.75em]
{\large \textbf{Contracting with Imperfect Commitment: Minimal Canonical
Contracts}}\\[0.5em]
Seungjin Han and Siyang Xiong
\end{center}

\noindent In private contracting, principals privately offer contracts to
the agent, and the agent sends messages privately, one to each principal.
Principal $j$ observes only his own mechanism and the message sent to him;
he does not observe the other principals' mechanisms or the messages sent to
them. A principal's off-path belief about the other principals' mechanisms,
messages, and continuation play therefore depends on the contract space
allowed in the model. This section gives the corresponding
private-contracting solution concepts and identifies the canonical contract
space.

The private-contracting analysis uses the current notation from Sections \ref%
{sec:model} and \ref{sec:canonical}. For each principal, write
\begin{equation*}
Z_j:=\{(x_j,y_j)\in X_j\times Y_j:y_j\in F_j(x_j)\},\qquad Z:=\times_{j\in%
\mathcal{J}}Z_j .
\end{equation*}
The finite multiple-principal minimality result below uses $|Z|$ for the
cardinality of the feasible action-pair profiles.

Throughout this section, message spaces, contract spaces, continuation
strategy spaces, and belief spaces are restricted to the admissible
standard-Borel objects described in Section \ref{sec:model}. In particular,
the measurable selectors and regular conditional probabilities used in the
constructions below are assumed to exist whenever the corresponding
construction is invoked. These requirements are automatic in the finite
environments used for the minimality results.

\section*{Solution Concepts}

Fix a contract space $\mathcal{G}=\times_{j\in\mathcal{J}}\mathcal{G}_j$. A
private-contracting profile is
\begin{equation*}
\left[ G=(G_j:M_j\to X_j)_{j\in\mathcal{J}},(\gamma,q,p)\right],
\end{equation*}
where $G\in\mathcal{G}$, $M=\times_{j\in\mathcal{J}}M_j$, the agent's
message strategy is $q:\Theta\to\Delta(M)$, principal $j$'s continuation
strategy is $\gamma_j:M_j\to Y_j$ with
\begin{equation*}
\gamma_j(m_j)\in F_j(G_j(m_j))\qquad\text{for every }m_j\in M_j,
\end{equation*}
and principal $j$'s on-path belief is
\begin{equation*}
p_j:M_j\to\Delta(\Theta\times M_{-j}).
\end{equation*}
Let $\mathcal{Y}_j(G_j)$ denote the set of feasible continuation strategies
for principal $j$ after contract $G_j$.

\begin{define}[Bayesian consistency]
Beliefs are Bayesian consistent if, for every principal $j$, every
measurable $E_1\subseteq\Theta$, $E_2\subseteq M_j$, and $E_3\subseteq
M_{-j} $,
\begin{equation*}
\int_{E_1}q(\theta)[E_2\times E_3]\,\mu(d\theta)
=\int_{\Theta}\int_{E_2}p_j(m_j)[E_1\times E_3]\,
q_{M_j}(\theta)(dm_j)\,\mu(d\theta),
\end{equation*}
where $q_{M_j}(\theta)$ is the marginal of $q(\theta)$ on $M_j$.
\end{define}

The distribution over messages to principal $j$ induced by $(\mu,q)$ is
denoted by $\mu_{M_j}^{q}$: for every measurable $E\subseteq M_j$,
\begin{equation*}
\mu_{M_j}^{q}(E):=\int_{\Theta}q_{M_j}(\theta)[E]\,\mu(d\theta).
\end{equation*}
The on-path messages of principal $j$ are the elements of $\func{supp}%
\mu_{M_j}^{q}$.

For off-path messages, principal $j$ may hold a belief over the type, the
other principals' messages, the other principals' contracts, and the other
principals' continuation strategies. Given $\mathcal{G}$, define
\begin{equation*}
\Delta_{-j}^{\mathcal{G}}:=
\Delta\left(\left\{(\theta,(m_k,G_k,\gamma_k)_{k\neq j}):
\begin{array}{l}
\theta\in\Theta,\ G_k:M_k\to X_k\text{ belongs to }\mathcal{G}_k, \\
m_k\in M_k,\ \gamma_k\in\mathcal{Y}_k(G_k),\ k\neq j%
\end{array}%
\right\}\right).
\end{equation*}

\begin{define}[principal $j$'s legitimate belief]
A legitimate belief for principal $j$ is a map $\eta_j:M_j\to\Delta_{-j}^{%
\mathcal{G}}$ such that, whenever $m_j\in\func{supp}\mu_{M_j}^{q}$, $%
\eta_j(m_j)$ puts probability one on the actual $(G_{-j},\gamma_{-j})$ and
has marginal $p_j(m_j)$ on $\Theta\times M_{-j}$.
\end{define}

\begin{define}[agent's incentive compatibility]
The agent's incentive and participation condition is
\begin{multline*}
\int_M u\left((G_k(m_k))_{k\in\mathcal{J}}, (\gamma_k(m_k))_{k\in\mathcal{J}%
},\theta\right)q(\theta)(dm) \\
\geq \max\left\{ U(\theta), \sup_{m\in M}u\left((G_k(m_k))_{k\in\mathcal{J}%
}, (\gamma_k(m_k))_{k\in\mathcal{J}},\theta\right) \right\}
\end{multline*}
for every $\theta\in\Theta$.
\end{define}

\begin{define}[principal $j$'s incentive compatibility]
Principal $j$'s continuation optimality requires that there exist a
legitimate belief $\eta_j$ such that, for every $m_j\in M_j$ and every $%
y_j\in F_j(G_j(m_j))$,
\begin{multline*}
\int v_j\left((G_j(m_j),\widetilde G_{-j}(m_{-j})),
(\gamma_j(m_j),\widetilde\gamma_{-j}(m_{-j})),\theta\right)
\eta_j(m_j)(d\theta,dm_{-j},d\widetilde G_{-j},d\widetilde\gamma_{-j}) \\
\geq \int v_j\left((G_j(m_j),\widetilde G_{-j}(m_{-j})),
(y_j,\widetilde\gamma_{-j}(m_{-j})),\theta\right)
\eta_j(m_j)(d\theta,dm_{-j},d\widetilde G_{-j},d\widetilde\gamma_{-j}).
\end{multline*}
\end{define}

\begin{define}[$\mathcal{G}$-continuation private equilibrium]
A profile $[G,(\gamma,q,p)]$ is a $\mathcal{G}$-continuation private
equilibrium if $G\in\mathcal{G}$, beliefs are Bayesian consistent, the
agent's incentive and participation condition holds, and every principal's
continuation optimality condition holds.
\end{define}

The ex ante payoff of principal $j$ is
\begin{equation*}
V_j[G,(\gamma,q,p)] :=\int_{\Theta}\int_M v_j\left((G_k(m_k))_{k\in\mathcal{J%
}}, (\gamma_k(m_k))_{k\in\mathcal{J}},\theta\right)
q(\theta)(dm)\,\mu(d\theta).
\end{equation*}

\begin{define}[$\mathcal{G}$-robust private equilibrium]
\label{def:G-robust-private} A profile $[G,(\gamma,q,p)]$ is a $\mathcal{G}$%
-robust private equilibrium if it is a $\mathcal{G}$-continuation private
equilibrium and, for every principal $j$ and every deviation $%
G_{j}^{\prime}\in\mathcal{G}_j$, there exists a post-deviation profile
\begin{equation*}
\left[ G^{\prime}=(G_k^{\prime})_{k\in\mathcal{J}},
(\gamma^{\prime},q^{\prime},p^{\prime})\right]
\end{equation*}
such that
\begin{equation*}
\left[(G_k^{\prime},\gamma_k^{\prime},p_k^{\prime})_{k\neq j}\right] =\left[%
(G_k,\gamma_k,p_k)_{k\neq j}\right],
\end{equation*}
$G_{j}^{\prime}$ is the deviating contract, $p^{\prime}$ is Bayesian
consistent with respect to $q^{\prime}$ under $[G^{\prime},(\gamma^{%
\prime},q^{\prime},p^{\prime})]$, the agent's incentive compatibility holds
under $[G^{\prime},(\gamma^{\prime},q^{\prime},p^{\prime})]$, principal $j$%
's incentive compatibility holds under $[G^{\prime},(\gamma^{\prime},q^{%
\prime},p^{\prime})]$, and
\begin{equation*}
V_j[G,(\gamma,q,p)]\geq
V_j[G^{\prime},(\gamma^{\prime},q^{\prime},p^{\prime})].
\end{equation*}
\end{define}

The equality for the non-deviating principals records the private
information structure: after principal $j$ deviates, every other principal
continues with the same mechanism, continuation strategy, and belief system
because he observes neither $G_{j}^{\prime}$ nor the message sent to
principal $j$.

A profile is a continuation private equilibrium, respectively robust private
equilibrium, if it is a $\mathcal{G}$-continuation private equilibrium,
respectively $\mathcal{G}$-robust private equilibrium, for every contract
space $\mathcal{G}$ containing its contract profile.

\section*{Single-Principal Case}

With a single principal, private and public contracting coincide: there are
no other principals' mechanisms or messages to be hidden. Hence a $\mathcal{G%
}$-continuation private equilibrium is exactly the $\mathcal{G}$%
-continuation equilibrium defined in the main text. The notions of canonical
and minimal canonical contract spaces are therefore the same as in the
public-contracting single-principal problem.

\begin{theo}
\label{thm:private_single_current} Suppose $|\Theta|\geq |X|$ and $|\mathcal{%
J}|=1$. Then $\mathcal{G}^{\ast}$ is a minimal canonical contract space for
continuation private equilibrium.
\end{theo}

\begin{proof}
This is Theorem \ref{thm:minimal:weak} applied to the private-contracting
notation. With one principal, the agent sends only one message, and no other
principal observes or fails to observe anything. Thus the private and public
continuation-equilibrium restrictions are identical.
\end{proof}

\section*{Multiple-Principal Case}

For the remainder of this section, suppose $|\mathcal{J}|\geq 2$. With
multiple principals, private contracting differs from public contracting
because principal $j$'s mechanism and the agent's message to principal $j$
are not observed by the other principals. For robust private equilibrium,
menus with recommendations alone are not minimal. One must also allow plain
menus of contractible actions. Define
\begin{equation*}
\widehat X_j:=\{x_j\in X_j: |F_j(x_j)|\geq 2\}.
\end{equation*}
For every $L_j\in\mathcal{L}_j^\ast$ with $L_j\cap\widehat
X_j\neq\varnothing $, let
\begin{equation*}
B_j^{L_j}:L_j\to X_j,\qquad B_j^{L_j}(x_j)=x_j,
\end{equation*}
be the plain-menu contract that lets the agent select a contractible action
but not recommend a non-contractible action. Let
\begin{equation*}
\mathcal{G}_j^\bullet :=\{B_j^{L_j}:L_j\in\mathcal{L}_j^\ast,\
L_j\cap\widehat X_j\neq\varnothing\}.
\end{equation*}
If $L_j\cap\widehat X_j=\varnothing$, every feasible $F_j(x_j)$ is a
singleton on $L_j$, so the corresponding plain menu is behaviorally
equivalent to the menu-with-recommendations contract in $\mathcal{G}_j^\ast$.

Define the private canonical contract space by
\begin{equation*}
\mathcal{G}_j^{\mathrm{p}}:=\mathcal{G}_j^\ast\sqcup \mathcal{G}%
_j^\bullet,\qquad \mathcal{G}^{\mathrm{p}}:=\times_{j\in\mathcal{J}}
\mathcal{G}_j^{\mathrm{p}}.
\end{equation*}
The disjoint-union notation means that a plain menu and a
menu-with-recommendations are treated as distinct contracts when their
message spaces differ.

\begin{define}
A contract space $\mathcal{G}$ is canonical for robust private equilibrium
if, for every payoff environment,
\begin{equation*}
\mathcal{Z}^{\text{robust-private-}\mathcal{G}\text{-}\mathcal{E}} =\mathcal{%
Z}^{\text{robust-private-}\mathcal{E}},
\end{equation*}
where the left side is the set of allocations induced by $\mathcal{G}$%
-robust private equilibria and the right side is the set of allocations
induced by unrestricted robust private equilibria.
\end{define}

The following relation is used repeatedly. For two contracts $%
G_{j}^{\prime}:M_{j}^{\prime}\to X_j$ and $G_j:M_j\to X_j$, write $%
G_{j}^{\prime}\sqsupset_j G_j$ if there is a measurable surjection $%
\iota_j:M_{j}^{\prime}\to M_j$, admitting a measurable right inverse, such
that
\begin{equation*}
G_{j}^{\prime}(m_{j}^{\prime})=G_j(\iota_j(m_{j}^{\prime})), \qquad \forall
m_{j}^{\prime}\in M_{j}^{\prime}.
\end{equation*}
Thus $G_{j}^{\prime}$ is a refinement or relabeling of $G_j$.

\begin{lemma}[private relabeling]
\label{lem:private_relabeling_current} Fix a contract space $\mathcal{G}$
and a $\mathcal{G}$-continuation private equilibrium inducing allocation $z$%
. If principal $j$'s contract is replaced by a contract $G_{j}^{\prime}$
with $G_{j}^{\prime}\sqsupset_j G_j$, then there is a continuation private
equilibrium of the enlarged contract space with contract profile $%
(G_{j}^{\prime},G_{-j})$ that induces the same allocation $z$. The same
construction also preserves the post-deviation requirements in Definition %
\ref{def:G-robust-private} when it is applied to the deviating principal's
contract.
\end{lemma}

\begin{proof}
Let $\iota_j:M_{j}^{\prime}\to M_j$ and $\rho_j:M_j\to M_{j}^{\prime}$ be
the surjection and a right inverse. Let $\widehat M_j^{\prime}:=\rho_j(M_j)$%
; then $\iota_j:\widehat M_j^{\prime}\to M_j$ is a bijection. The agent uses
only messages in $\widehat M_j^{\prime}$ and treats $\rho_j(m_j)$ as the old
message $m_j$. Formally, for measurable $E_j^{\prime}\subseteq M_j^{\prime}$
and $E_{-j}\subseteq M_{-j}$,
\begin{equation*}
q^{\prime}(\theta)[E_j^{\prime}\times E_{-j}]
=q(\theta)[\iota_j(E_j^{\prime}\cap\widehat M_j^{\prime})\times E_{-j}].
\end{equation*}
Set
\begin{equation*}
\gamma_j^{\prime}(m_{j}^{\prime})=\gamma_j(\iota_j(m_{j}^{\prime})),\qquad
\gamma_k^{\prime}=\gamma_k\quad (k\neq j).
\end{equation*}
For principal $j$, set $p_j^{\prime}(m_{j}^{\prime})=p_j(\iota_j(m_{j}^{%
\prime}))$. For each $k\neq j$, define $p_k^{\prime}$ by pushing $p_k$
forward through the map that replaces the old $j$-message $m_j$ by $%
\rho_j(m_j)$ and leaves all other coordinates unchanged. Legitimate beliefs
are transformed in the same way on path; off path, copy the old legitimate
belief after the corresponding old message. Because $G_{j}^{\prime}(m_{j}^{%
\prime})=G_j(\iota_j(m_{j}^{\prime}))$, every payoff available to the agent
or to a principal under the relabeled profile was already available under
the original profile. Incentive compatibility, participation, and
continuation optimality are therefore preserved, and the induced allocation
is unchanged.
\end{proof}

\begin{lemma}[one-principal private recommendation expansion]
\label{lem:private_recommendation_expansion_current} Fix a contract space $%
\mathcal{G}$ and a $\mathcal{G}$-continuation private equilibrium inducing
allocation $z$. For any principal $j$, replacing only $G_j$ by the
menu-with-recommendations contract $G_j^{G_j}\in \mathcal{G}_j^\ast$ yields
a continuation private equilibrium of the enlarged contract space that
induces the same allocation $z$.
\end{lemma}

\begin{proof}
Let $L_j=G_j(M_j)$ and let $G_j^{G_j}$ have message space
\begin{equation*}
M_j^{G_j}:=\{(x_j,y_j):x_j\in L_j,\ y_j\in F_j(x_j)\}, \qquad
G_j^{G_j}(x_j,y_j)=x_j .
\end{equation*}
The original equilibrium induces a probability measure $\pi_j$ on $%
\Theta\times M_j^{G_j}\times M_{-j}$ by
\begin{equation*}
\pi_j(E)=\int_{\Theta}q(\theta)\left[ \left\{m\in
M:\left(\theta,(G_j(m_j),\gamma_j(m_j)),m_{-j}\right)\in E \right\}\right]%
\mu(d\theta)
\end{equation*}
for every measurable $E\subseteq \Theta\times M_j^{G_j}\times M_{-j}$. Let $%
\pi_{j,M_j^{G_j}}$ be its marginal on $M_j^{G_j}$, and let $%
\pi_{j,\Theta\times M_{-j}}(x_j,y_j)$ denote a regular conditional
distribution on $\Theta\times M_{-j}$ given $(x_j,y_j)$ whenever $(x_j,y_j)$
is in the support of $\pi_{j,M_j^{G_j}}$. For each $x_j\in L_j$, fix a
selector $\sigma_j(x_j)\in M_j$ satisfying $G_j(\sigma_j(x_j))=x_j$.

Define the new message strategy as the push-forward of $q$ under
\begin{equation*}
T_j(m):=\left((G_j(m_j),\gamma_j(m_j)),m_{-j}\right).
\end{equation*}
Equivalently, for every measurable $E^{\prime}\subseteq M_j^{G_j}\times
M_{-j}$,
\begin{equation*}
q^{\prime}(\theta)[E^{\prime}] =q(\theta)\left[ \left\{m\in
M:\left((G_j(m_j),\gamma_j(m_j)),m_{-j}\right)\in E^{\prime}\right\} \right].
\end{equation*}
For $k\neq j$, keep $G_k$ and $\gamma_k$ fixed. Principal $j$'s continuation
strategy is
\begin{equation*}
\gamma_j^{\prime}(x_j,y_j)=
\begin{cases}
y_j, & \text{if }(x_j,y_j)\in\func{supp}\pi_{j,M_j^{G_j}}, \\
\gamma_j(\sigma_j(x_j)), & \text{otherwise.}%
\end{cases}%
\end{equation*}
For beliefs, set $p_k^{\prime}$ for $k\neq j$ equal to the push-forward of $%
p_k$ under the replacement of the old $j$-message $m_j$ by $%
(G_j(m_j),\gamma_j(m_j))$. Principal $j$'s belief is
\begin{equation*}
p_j^{\prime}(x_j,y_j)=
\begin{cases}
\pi_{j,\Theta\times M_{-j}}(x_j,y_j), & \text{if }(x_j,y_j)\in\func{supp}%
\pi_{j,M_j^{G_j}}, \\
p_j(\sigma_j(x_j)), & \text{otherwise.}%
\end{cases}%
\end{equation*}
Legitimate off-path beliefs are transformed by the same selector $\sigma_j$.

This is the old private-contracting construction with only principal $j$'s
private message recoded. The new profile replicates the distribution over
types, contractible actions, and continuation actions induced by the old
profile. Bayesian consistency is preserved by construction of the regular
conditional beliefs and the push-forward beliefs. The agent's incentive and
participation condition and all principals' continuation optimality
conditions are inherited from the original profile. Hence the resulting
profile is a continuation private equilibrium and induces the same
allocation $z$.
\end{proof}

\begin{prop}
\label{prop:private_canonical_upper_current} For every payoff environment
and every contract space $\mathcal{G}$ with $\mathcal{G}^{\mathrm{p}%
}\subseteq\mathcal{G}$,
\begin{equation*}
\mathcal{Z}^{\text{robust-private-}\mathcal{G}\text{-}\mathcal{E}} \subseteq
\mathcal{Z}^{\text{robust-private-}\mathcal{G}^{\mathrm{p}}\text{-}\mathcal{E%
}}.
\end{equation*}
\end{prop}

\begin{proof}
Take a $\mathcal{G}$-robust private equilibrium $[G,(\gamma,q,p)]$ inducing $%
z$. Apply Lemma \ref{lem:private_recommendation_expansion_current}
successively, one principal at a time, to obtain an allocation-equivalent
continuation private equilibrium $[\bar G,(\bar\gamma,\bar q,\bar p)]$,
where $\bar G_k=G_k^{G_k}\in\mathcal{G}_k^\ast$ for every $k$. Choose the
principalwise construction once and for all. For each principal $k$, let
\begin{equation*}
\tau_k(m_k):=(G_k(m_k),\gamma_k(m_k))
\end{equation*}
be the canonical message generated by the old message $m_k$, and fix the
selectors used to define continuation actions and off-path beliefs outside
the image of $\tau_k$. Thus $(\bar G_k,\bar\gamma_k,\bar p_k)$ is a fixed
recoding of the original profile: $\bar G_k$ and $\bar\gamma_k$ are
determined by $(G_k,\gamma_k)$, while $\bar p_k$ is the corresponding
push-forward of $p_k$ through the fixed recoding maps for the other
principals' private message coordinates.

Now fix a principal $j$ and a deviation $G_{j}^{\prime}\in\mathcal{G}_j^{%
\mathrm{p}}$. Since $\mathcal{G}^{\mathrm{p}}\subseteq\mathcal{G}$, the
original $\mathcal{G}$-robust equilibrium provides a post-deviation profile
after $G_{j}^{\prime}$ in which every non-deviating principal $k\neq j$
keeps the old mechanism, continuation strategy, and belief system $%
(G_k,\gamma_k,p_k)$, and principal $j$ obtains no more than his original
payoff. Write this post-deviation profile as
\begin{equation*}
\left[(G_{j}^{\prime},G_{-j}),(\gamma^{d},q^{d},p^{d})\right],
\end{equation*}
where $(G_k,\gamma_k,p_k)_{k\neq j}$ are the original non-deviating private
objects.

Replicate this post-deviation profile after replacing the non-deviating
contracts $G_k$ by $\bar G_k=G_k^{G_k}$. The replicated message strategy $%
\widehat q$ is the push-forward of $q^d$ by the map
\begin{equation*}
(m_j^{\prime},m_{-j})\mapsto (m_j^{\prime},(\tau_k(m_k))_{k\neq j}).
\end{equation*}
Equivalently, for every measurable $E\subseteq
M_j^{\prime}\times\prod_{k\neq j}M_k^{G_k}$,
\begin{equation*}
\widehat q(\theta)[E] =q^{d}(\theta)\left[ \left\{(m_j^{\prime},m_{-j}):
(m_j^{\prime},(\tau_k(m_k))_{k\neq j})\in E\right\}\right].
\end{equation*}
Principal $j$ keeps the continuation strategy from the original
post-deviation profile. Each non-deviating principal $k\neq j$ uses the
fixed continuation strategy $\bar\gamma_k$ from the initial
canonicalization: on messages in the image of $\tau_k$, it chooses the
recommended continuation action, and off this image it uses the fixed
selector chosen above. Principal $j$'s on-path beliefs are the push-forward
of $p_j^d$ through the same maps $\tau_k$, $k\neq j$. For every
non-deviating principal $k\neq j$, the resulting private object is exactly
the fixed $(\bar G_k,\bar\gamma_k,\bar p_k)$, because it is obtained by the
same principalwise recoding and the same belief push-forward as before the
deviation. Principal $j$'s legitimate belief is pushed forward through the
same recoding of the non-deviating principals' private message coordinates.

This is the old post-deviation replication argument, applied only to the
non-deviating principals' private message coordinates. It does not require
any non-deviating principal to observe $G_j^{\prime}$ or the message sent to
principal $j$. It preserves the joint distribution over types and realized
action pairs, so Bayesian consistency, the agent's incentive condition,
principal $j$'s continuation optimality, and principal $j$'s payoff are
preserved. Therefore every deviation in $\mathcal{G}_j^{\mathrm{p}}$ is met
by an admissible post-deviation profile that leaves the deviator no better
off, so $[\bar G,(\bar\gamma,\bar q,\bar p)]$ is a $\mathcal{G}^{\mathrm{p}}$%
-robust private equilibrium inducing $z$.
\end{proof}

\begin{prop}
\label{prop:private_canonical_lower_current} For every payoff environment,
\begin{equation*}
\mathcal{Z}^{\text{robust-private-}\mathcal{G}^{\mathrm{p}}\text{-}\mathcal{E%
}} \subseteq \mathcal{Z}^{\text{robust-private-}\mathcal{E}}.
\end{equation*}
\end{prop}

\begin{proof}
Let $[G,(\gamma,q,p)]$ be a $\mathcal{G}^{\mathrm{p}}$-robust private
equilibrium. We show that it is robust against any admissible private
contract. Fix a principal $j$ and an arbitrary deviation $%
G_{j}^{\prime}:M_{j}^{\prime}\to X_j$. Let $L_j:=G_{j}^{\prime}(M_{j}^{%
\prime})$. By admissibility of $G_j^{\prime}$, $L_j\in\mathcal{L}_j^\ast$.

If $L_j\cap\widehat X_j\neq\varnothing$, let $G_{j}^{\prime%
\prime}:=B_j^{L_j}\in \mathcal{G}_j^\bullet$. If $L_j\cap\widehat
X_j=\varnothing$, let $G_{j}^{\prime\prime}$ be the unique
menu-with-recommendations contract in $\mathcal{G}_j^\ast$ with image $L_j$;
because each $F_j(x_j)$ is a singleton on $L_j$, this contract is
behaviorally a plain menu. In both cases $G_{j}^{\prime}\sqsupset_j
G_{j}^{\prime\prime}$.

Robustness against the private-canonical deviation $G_{j}^{\prime\prime}$
gives a post-deviation profile after $(G_{j}^{\prime\prime},G_{-j})$ that
keeps all non-deviating principals' mechanisms, continuation strategies, and
beliefs fixed, and that does not improve principal $j$'s payoff. The
relabeling construction in Lemma \ref{lem:private_relabeling_current},
applied only to the deviating principal's contract, gives a corresponding
post-deviation profile after $(G_{j}^{\prime},G_{-j})$. This application
does not require the post-deviation profile after $G_{j}^{\prime\prime}$ to
be a full continuation private equilibrium for all principals. It preserves
exactly the requirements imposed after a private deviation: the agent's
incentive condition, principal $j$'s continuation optimality, principal $j$%
's payoff, and the fixed private objects $(G_k,\gamma_k,p_k)_{k\neq j}$ of
all non-deviating principals. Hence the arbitrary deviation $G_{j}^{\prime}$
is not a safe profitable deviation. Since $j$ and $G_{j}^{\prime}$ were
arbitrary, the original profile is an unrestricted robust private
equilibrium.
\end{proof}

\begin{theo}
\label{thm:private_canonical_current} Suppose $|\mathcal{J}|\geq 2$. Then $%
\mathcal{G}^{\mathrm{p}}$ is canonical for robust private equilibrium.
\end{theo}

\begin{proof}
This follows immediately from Propositions \ref%
{prop:private_canonical_upper_current} and \ref%
{prop:private_canonical_lower_current}.
\end{proof}

\section*{Minimality}

The minimality statement below is the finite-action version of the private
canonicality result. Assume in this subsection that each $X_j$ and each
feasible set $F_j(x_j)$ is finite, and that $|\Theta|\geq |Z|$.

\begin{define}
A canonical contract space $\mathcal{G}$ is minimal for robust private
equilibrium if every canonical contract space $\mathcal{G}^{\prime}$ for
robust private equilibrium satisfies $|\mathcal{G}^{\prime}|\geq|\mathcal{G}%
| $.
\end{define}

\begin{lemma}
\label{lem:private_minimal_rich_current} Suppose $|\mathcal{J}|\geq 2$. If $%
\mathcal{G}^{\prime }$ is canonical for robust private equilibrium, then for
every principal $j$ and every non-empty $L_{j}\subseteq X_{j}$, $\mathcal{G}%
_{j}^{\prime }$ contains a contract $\left\langle G_{j}:M_{j}\rightarrow
X_{j}\right\rangle $ with image $L_{j}$ and at least $|F_{j}(x_{j})|$
messages above each $x_{j}\in L_{j}$, i.e.,
\begin{equation*}
\left\vert \left\{ m_{j}\in M_{j}:G_{j}\left( m_{j}\right) =x_{j}\right\}
\right\vert \geq |F_{j}(x_{j})|,\quad \forall x_j\in L_{j}=G_{j}(M_{j}).
\end{equation*}
\end{lemma}

\begin{proof}
Fix $j$ and $L_j$. Since $|\Theta|\geq |Z|$, choose a surjection
\begin{equation*}
\varphi_j:\Theta\to \{(x_j,y_j):x_j\in L_j,\ y_j\in F_j(x_j)\}.
\end{equation*}
Use the payoff environment from the old lower-bound construction:
\begin{equation*}
x_j\notin L_j \quad\Longrightarrow\quad u(x,y,\theta)=-v_k(x,y,\theta)=8
\end{equation*}
for every $(x,y,\theta,k)\in X\times Y\times\Theta\times\mathcal{J}$, and,
when $x_j\in L_j$,
\begin{equation*}
u(x,y,\theta)=v_k(x,y,\theta)=
\begin{cases}
1, & \text{if }(x_j,y_j)=\varphi_j(\theta), \\
0, & \text{otherwise,}%
\end{cases}%
\end{equation*}
for every principal $k$. Thus, given $x_j\in L_j$, all players agree that
the pair $\varphi_j(\theta)$ is uniquely best; if $x_j\notin L_j$, the agent
strictly prefers that action while every principal strictly dislikes it.

The menu-with-recommendations contract with image $L_j$ supports a robust
private equilibrium in which the agent sends the message $\varphi_j(\theta)$
to principal $j$, and principal $j$ chooses the recommended continuation
action. Deviations by principal $j$ that include an action outside $L_j$ are
deterred because the agent would choose such an action and give principal $j$
payoff $-8$. Deviations whose image is contained in $L_j$ cannot improve on
payoff one at each state. Deviations by the other principals do not affect
this payoff construction.

Canonicality of $\mathcal{G}^{\prime}$ requires a $\mathcal{G}^{\prime}$%
-robust private equilibrium inducing the same allocation. Principal $j$'s
contract in such an equilibrium must have image exactly $L_j$: every $x_j\in
L_j$ is used by the allocation, while any available message inducing $%
x_j\notin L_j$ would be chosen by the agent and would change the allocation.
Moreover, for every $x_j\in L_j$ and every $y_j\in F_j(x_j)$, the surjection
uses the pair $(x_j,y_j)$ at some state. Since principal $j$'s continuation
action is a function of the message he receives, the contract must have at
least $|F_j(x_j)|$ distinct messages above each $x_j\in L_j$.
\end{proof}

\begin{lemma}
\label{lem:private_minimal_plain_current} Suppose $|\mathcal{J}|\geq 2$. If $%
\mathcal{G}^{\prime }$ is canonical for robust private equilibrium, then for
every principal $j$ and every non-empty $L_{j}\subseteq X_{j}$ with $%
L_{j}\cap \widehat{X}_{j}\neq \varnothing $, $\mathcal{G}_{j}^{\prime }$
contains a contract with image $L_{j}$ that has fewer than $|F_{j}(x_{j})|$
messages above some $x_{j}\in L_{j}$.
\end{lemma}

\begin{proof}
Suppose not. Then there are a principal $j$, a non-empty $L_j\subseteq X_j$
with $L_j\cap\widehat X_j\neq\varnothing$, and $\widetilde x_j\in
L_j\cap\widehat X_j$ such that every $\mathcal{G}_{j}^{\prime}$ contract
with image $L_j$ has at least $|F_j(x_j)|$ messages above every $x_j\in L_j$%
. Choose two distinct $\widetilde y_j^1,\widetilde y_j^2\in F_j(\widetilde
x_j)$.

Pick $k\neq j$. By Assumption \ref{assm:non-trivial} and Lemma \ref%
{lem:private_minimal_rich_current}, principal $k$ has a contract $\widetilde
G_k:\widetilde M_k\to X_k$ in $\mathcal{G}_{k}^{\prime}$ with two distinct
messages $\widetilde m_k^1,\widetilde m_k^2$ that can induce two distinct
feasible action pairs
\begin{equation*}
a_k^1=(\widetilde x_k^1,\widetilde y_k^1),\qquad a_k^2=(\widetilde
x_k^2,\widetilde y_k^2).
\end{equation*}
Principals other than $j$ and $k$ receive payoff zero at every outcome. Set $%
u=v_k$, so the agent and principal $k$ have the same preferences.

Choose two states $\theta^1,\theta^2$ and, for each $x\in
L_j\setminus\{\widetilde x_j\}$, an auxiliary state $\theta^x$. Put the
uniform prior on
\begin{equation*}
S:=\{\theta^1,\theta^2\}\cup \{\theta^x:x\in L_j\setminus\{\widetilde
x_j\}\}, \qquad \mu(\{\theta\})=\kappa>0\quad(\theta\in S),
\end{equation*}
and make all states outside $S$ payoff-irrelevant. Choose payoffs bounded in
absolute value by $100/\kappa$. At every state in $S$, if $x_j\notin L_j$,
set
\begin{equation*}
u(x,y,\theta)=-v_j(x,y,\theta)=\frac{100}{\kappa}.
\end{equation*}
At each auxiliary state $\theta^x$, for $x_j\in L_j$, set
\begin{equation*}
u(x,y,\theta^x)=v_k(x,y,\theta^x)=v_j(x,y,\theta^x)=
\begin{cases}
\frac{100}{\kappa}, & \text{if }x_j=x, \\
-\frac{100}{\kappa}, & \text{if }x_j\neq x.%
\end{cases}%
\end{equation*}
At $\theta^1$ and $\theta^2$, any outcome outside
\begin{equation*}
\{(\widetilde x_j,\widetilde y_j^1), (\widetilde x_j,\widetilde y_j^2)\}
\times\{a_k^1,a_k^2\}
\end{equation*}
gives the agent, principal $k$, and principal $j$ payoff $-100/\kappa$. On
the four relevant outcomes, payoffs are:
\begin{equation*}
\begin{array}{c|cc}
\multicolumn{3}{c}{\theta=\theta^1} \\
& a_k^1 & a_k^2 \\ \hline
(\widetilde x_j,\widetilde y_j^1) & (2,2) & (0,0) \\
(\widetilde x_j,\widetilde y_j^2) & (0,-100/\kappa) & (0,-100/\kappa)%
\end{array}%
\end{equation*}
and
\begin{equation*}
\begin{array}{c|cc}
\multicolumn{3}{c}{\theta=\theta^2} \\
& a_k^1 & a_k^2 \\ \hline
(\widetilde x_j,\widetilde y_j^1) & (2,2) & (0,0) \\
(\widetilde x_j,\widetilde y_j^2) & (0,0) & (2,1).%
\end{array}%
\end{equation*}
In each entry, the first coordinate is the common payoff of the agent and
principal $k$, and the second coordinate is principal $j$'s payoff.

By Lemma \ref{lem:private_minimal_rich_current}, there is a contract $%
\widetilde G_j:\widetilde M_j\to X_j$ in $\mathcal{G}_{j}^{\prime}$ with
image $L_j$. By the maintained contradiction hypothesis, choose distinct
messages $\widetilde m_j^1,\widetilde m_j^2\in\widetilde M_j$ with
\begin{equation*}
\widetilde G_j(\widetilde m_j^1)= \widetilde G_j(\widetilde
m_j^2)=\widetilde x_j.
\end{equation*}
Consider the profile in which principal $j$ offers $\widetilde G_j$ and
principal $k$ offers $\widetilde G_k$. At $\theta^1$, the agent sends $%
(\widetilde m_j^1,\widetilde m_k^1)$ and the induced action pairs are $%
(\widetilde x_j,\widetilde y_j^1)$ and $a_k^1$. At $\theta^2$, the agent
sends $(\widetilde m_j^2,\widetilde m_k^2)$ and the induced action pairs are
$(\widetilde x_j,\widetilde y_j^2)$ and $a_k^2$. At each auxiliary state $%
\theta^x$, the agent sends a message to principal $j$ inducing contractible
action $x$.

This is a $\mathcal{G}^{\prime}$-robust private equilibrium. The agent and
principal $k$ attain their maximal feasible payoff at the special states.
Principal $j$ does not profit from a deviation whose image contains an
action outside $L_j$, because the agent can select such an action and give
principal $j$ payoff $-100/\kappa$. He does not profit from a deviation
whose image omits an action in $L_j$, because the corresponding auxiliary
state gives him the large negative payoff. If his deviation has image
exactly $L_j$, the maintained contradiction hypothesis supplies two distinct
messages above $\widetilde x_j$, so the same state-separating continuation
can be replicated. The continuation responses keep all non-deviating
principals' private objects fixed, as required under private contracting.
Deviations by other principals are deterred by the same payoff construction
or are payoff-irrelevant.

Now consider the unrestricted private-contracting problem. Principal $j$ can
deviate to the plain menu $B_j^{L_j}:L_j\to X_j$, $B_j^{L_j}(x_j)=x_j$. At
the auxiliary states, the large payoffs force the agent to choose the
corresponding $x\in L_j$. At the two special states, the agent must send the
single message $\widetilde x_j$ to principal $j$. Thus principal $j$ cannot
condition his non-contractible action on whether the state is $\theta^1$ or $%
\theta^2$. This is exactly the private-contracting restriction: the agent's
message to a non-deviating principal cannot be used to communicate the state
to principal $j$ after his private deviation. Given the pooled message $%
\widetilde x_j$, $\widetilde y_j^1$ is
the unique optimal continuation action for principal $j$: choosing $%
\widetilde y_j^2$ gives him the large negative payoff at $\theta^1$, while
the gain at $\theta^2$ is bounded. Given $\widetilde y_j^1$, the agent and
principal $k$ uniquely select the message inducing $a_k^1$ at both special
states. Hence every admissible post-deviation continuation after the plain
menu gives principal $j$ payoff $2$ at both $\theta^1$ and $\theta^2$, with
the same auxiliary-state payoffs as before. In the separating allocation
above, principal $j$ receives payoff $2$ at $\theta^1$ and payoff $1$ at $%
\theta^2$. Thus the plain menu is a safe profitable deviation from any
profile inducing the separating allocation, contradicting canonicality of $%
\mathcal{G}^{\prime}$.
\end{proof}

\begin{theo}
\label{thm:private_minimal_current} Suppose $|\Theta|\geq |Z|$, each $X_j$
and each $F_j(x_j)$ is finite, and $|\mathcal{J}|\geq 2$. Then $\mathcal{G}^{%
\mathrm{p}}$ is a minimal canonical contract space for robust private
equilibrium.
\end{theo}

\begin{proof}
For each non-empty $L_j\subseteq X_j$, Lemma \ref%
{lem:private_minimal_rich_current} requires one recommendation-rich
contract, corresponding to the contract in $\mathcal{G}_j^\ast$ with image $%
L_j$. For every $L_j$ that intersects $\widehat X_j$, Lemma \ref%
{lem:private_minimal_plain_current} requires an additional
recommendation-poor contract, corresponding to the plain menu in $\mathcal{G}%
_j^\bullet$. If $L_j$ does not intersect $\widehat X_j$, the plain menu is
behaviorally identical to the unique $\mathcal{G}_j^\ast$ contract with
image $L_j$, so no additional contract is needed. Therefore every canonical
private contract space has cardinality at least $|\mathcal{G}^{\mathrm{p}}|$%
. Theorem \ref{thm:private_canonical_current} establishes that $\mathcal{G}^{%
\mathrm{p}}$ is canonical, so it is minimal.
\end{proof}

\section*{Implications}

The single-principal results and applications in Sections \ref{sec:appI} and %
\ref{sec:revisable} are unchanged under private contracting, because with
one principal there is no distinction between public and private
contracting. The common-agency application in Section \ref{sec:appII} also
remains valid under private contracting when deviations are evaluated with
the private canonical contract space $\mathcal{G}^{\mathrm{p}}$.

The multiple-principal result is where private contracting differs from
public contracting. Because other principals do not observe principal $j$'s
mechanism or the message sent to him, deviations to plain menus cannot be
evaluated by using other principals' public messages to restore missing
recommendations. The canonical private contract space therefore contains
both menus-with-recommendations and plain menus. In intrinsic common agency,
this private observability structure can also weaken the separability
assumptions needed in an application: if non-deviating principals' contracts
do not allow the agent to communicate principal $j$'s deviation, then
principal $j$ takes their contractible and non-contractible actions as fixed.

\end{document}